\documentclass[a4paper,DIV=14]{scrartcl}
\usepackage[utf8]{inputenc}

\newenvironment{keywords}%
{\begin{quotation} \noindent \textbf{Keywords:}}% oder „Keywords:”
{\end{quotation}}

\setkomafont{title}{\normalfont\bfseries}
\setkomafont{section}{\normalfont\bfseries\LARGE}
\setkomafont{subsection}{\normalfont\bfseries\Large}
\setkomafont{subsubsection}{\normalfont\bfseries\large}
\setkomafont{subparagraph}{\normalfont\bfseries}

\usepackage[numbered]{bookmark}
\usepackage{caption,longtable}
\usepackage{amsmath}
\usepackage{amsthm}
\usepackage{amssymb}
\usepackage{microtype}%if unwanted, comment out or use option "draft"
\usepackage{stmaryrd}
\usepackage[table]{xcolor}
\usepackage{tikz}
\usepackage{tikz-qtree}
\usetikzlibrary{decorations,calc}
\usetikzlibrary{positioning}
\usetikzlibrary{shapes.geometric}
\usepackage{listings}
\usepackage{graphicx}

\newcommand\TTstrut{\rule{0pt}{3.5ex}}         % = `top' strut
\newcommand\Tstrut{\rule{0pt}{2.6ex}}         % = `top' strut
\newcommand\Bstrut{\rule[-0.9ex]{0pt}{0pt}}   % = `bottom' strut

\theoremstyle{plain}
\newtheorem{theorem}{Theorem}[section]

\newtheorem{proposition}[theorem]{Proposition}

\theoremstyle{definition}
\newtheorem{definition}[theorem]{Definition}
\newtheorem{example}[theorem]{Example}

\newcolumntype{?}[1]{!{\vrule width #1}}

\newcommand{\N}{\ensuremath{\mathbb{N}}}
\newcommand{\No}{\ensuremath{\mathbb{N}_0}}
\newcommand{\set}[2]{\ensuremath{\left\{ #1 \mid #2 \right\}}}

\newcommand{\multiset}[1]{\ensuremath{ {\{}\hspace{-0.8mm}{\{} #1 {\}}\hspace{-0.8mm}{\}} }}

\newcommand{\compmodel}{\ensuremath{\mathcal{M}}}
\newcommand{\cmcm}[1][k]{\ensuremath{\mathrm{CM}(#1)}}

\newcommand{\cmvm}[1][V,F]{\ensuremath{\mathrm{VM}(#1)}}
\newcommand{\cmspace}{\ensuremath{S}}
\newcommand{\opset}{\ensuremath{O}}
\newcommand{\predset}{\ensuremath{P}}
\newcommand{\predseq}{\ensuremath{\mathrm{pw}}}

\begin{document}

\title{A Systematic Approach to Programming} 
\author{Maurice Chandoo\footnote{Leibniz Universität Hannover,
        Institut für Theoretische Informatik,
        Appelstr.~4, 30167 Hannover, Germany; \hskip 2em      
        E-Mail: chandoo@thi.uni-hannover.de }}
\date{\vspace{-5ex}}

\maketitle

\begin{abstract}\noindent\textbf{Abstract.}
    Programming is commonly understood as the process of writing a computer program which accomplishes a specific task. 
    The average programmer often struggles to write a program which correctly solves the given task for all instances.  
    Failure to write a correct program indicates that he is unable to express what he wants the computer to do or that he is unable to solve the task himself. We assume that the latter is not the case, i.e.~the programmer knows a correct algorithm. 
    
    High-level programming languages with automatic memory management have significantly facilitated the task of expressing what we want a computer to do without having to be aware of its architecture. However, the challenge of translating a mental representation of an algorithm into code remains.
    While such programming languages make it possible that the code only needs to reflect the essentials of an algorithm and nothing else, they obviously cannot simplify this task any further.
    Thus, a method that guides this translation process is desirable. Existing formal programming methods fail to support this process. In fact, they complicate it only further by requiring the programmer to justify why the algorithm to be implemented is correct while translating it into code.     
    In contrast, we start out with the assumption that the algorithm to be implemented is correct and only focus on ensuring that the code faithfully implements the algorithm.     
    
    We describe a programming method for systematically implementing sequential algorithms in any imperative or functional programming language. The method is based on the premise that it is easy to write down how an algorithm proceeds on a concrete input.
    This information---which we call execution trace---is used as a starting point to derive the desired program. In contrast to test-driven development the program is directly constructed from the test cases instead of written separately. The program's operations, control flow and predicates guiding the control flow  are worked out separately, which saves the programmer from having to think about them simultaneously. This reduces the likelihood of introducing errors. We demonstrate our method for two examples and discuss its utility.  The method only requires pen and paper. However, a computer can facilitate its usage by taking care of certain trivial tasks. 
              
    Additionally, we provide a formal framework for programming in terms of execution traces. We consider the central concern of any such programming method: how to derive a program from a set of traces? We prove that this problem is equivalent to solving a certain coloring problem. In general, solving this problem by hand is impractical. Our method side steps this issue by leveraging the programmer's intuition. 
\end{abstract}


\section{Introduction}
We consider programming to be the thought process in which an algorithm is implemented as a program in some programming language. In practice, it is often conducted in an ad hoc fashion which can be summarized as trial and error: write some code, execute it to verify whether it behaves correctly, if a bug is detected modify the code and repeat this process until the code is believed to work as intended. Is there a more systematic alternative to ad hoc programming?

%At first sight, step-wise refinement seems to be a candidate \cite{wirth}. However, at closer inspection it becomes apparent that step-wise refinement does not say anything about programming (in our sense). Instead, it deals with the separate issue of algorithm design which comes before programming. Suppose you want to find an algorithm which solves a problem $X$. Assume that there is a simple algorithm for $X$ which uses subroutines for solving two other problems $Y$ and $Z$. In complexity theory one would say that $X$ reduces to $Y$ and $Z$. This allows us to split the task of finding an algorithm for $X$ into three parts: finding algorithms for $Y,Z$ and finding an algorithm for $X$ which uses the other two as subroutine. Step-wise refinement is a recursive procedure which describes how to construct a program for an algorithmic problem $X$. If $X$ is sufficiently simple (base case of the recursion) then write a program for it. Otherwise, reduce $X$ to some simpler problems such as $Y,Z$ in our example. Write a program for $X$ using placeholders for programs which compute $Y$ and $Z$. Recursively apply this procedure to $Y$ and $Z$. How to implement the algorithms that are found during this recursion is not covered by step-wise refinement.

The issue with ad hoc programming is its high potential of producing incorrect programs.
Consequently, it is important to verify a program's correctness, which can be done by testing or finding a  correctness proof. However, the inherent problem with verifying a program's correctness after its construction is that it does nothing to facilitate the actual act of programming \cite{dijk1}.
%\cite{dijk1}
%This has been also noticed by Dijkstra who said \cite{dijk1}: 
%\begin{quote}
%``When correctness concerns come as an afterthought and correctness proofs have to be
%given once the program is already completed, the programmer can indeed expect severe
%troubles. If, however, he adheres to the discipline to produce the correctness proofs as he
%programs along, he will produce program and proof with less effort than programming
%alone would have taken.''
%\end{quote}

Invariant-based programming \cite{back} and matrix code \cite{emdenmc,emdencbc} are methods where a program and its correctness proof are constructed hand in hand and thus constitute systematic alternatives to ad hoc programming. Both methods share the same understanding of what it means for a program to be correct. 
But what does it mean for a program to be correct?
  Suppose you want to sort an array of integers. It is plausible to call any program which accomplishes this correct. There might be an additional efficiency constraint such as using only $o(n^2)$ comparisons. In that case, a program which implements bubble sort would not be correct. Suppose you want to use quick sort. Then only the programs which implement quick sort would be correct.   
These are three distinct notions of correctness. In the first one, a program's correctness is specified in terms of its input/output-behavior. In the second one, constraints on a program's execution behavior are imposed. In the third one, the execution traces produced by a program must match a specific pattern that is recognized as execution of a certain algorithm. Both methods assume the first notion of correctness, which has a severe disadvantage: its generality.

Consider the program which gets a planar graph as input and always returns true. Prove that this program correctly solves the problem of deciding whether a planar graph is 4-colorable. Proving the program's correctness is equivalent to proving the four color theorem. 
More generally, proving an arbitrary mathematical statement can be reduced to verifying a program's correctness. Thus, any programming method which claims to produce correctness proofs w.r.t.~input/output-behavior must explain how to prove theorems in general. If it does not explain how to do this then how much does it really help in proving a program's correctness other than imposing some formalism? The underlying issue is the conflation of algorithm design and programming. 

Let $X$ be an algorithmic problem and let $C$ be a set of constraints such as time and space complexity requirements. We call the task of devising an algorithm $A$ which solves $X$ while satisfying $C$ and a proof thereof algorithm design. Implementing $A$ as a program $P$ in a given programming language is what we call programming, as stated in the beginning. We distinguish between the correctness of $A$ and $P$. We call $P$ correct if it faithfully implements $A$ (the third notion of correctness). In this sense, proving the four color theorem is \emph{not} part of programming but part of algorithm design. Any programming method which includes a correctness proof of the program's underlying algorithm mixes these two activities. 

For our method the third notion of correctness is assumed: a program is correct if it implements a given algorithm. But how is an algorithm given if not as program? We assume that the programmer possesses a mental representation of it. More specifically, we say somebody has a mental representation of an algorithm if they can execute it step by step for any concrete input. A recording of such an execution is called an execution trace. 
For example, to check whether a child possesses a mental representation of long multiplication, it can be asked to multiply various pairs of numbers on squared paper. 

We assume that the algorithm to be implemented is correct. This is not an assumption of our method in particular but of any programming method which considers algorithm design to be a distinct activity from programming. Proving an algorithm's correctness is part of algorithm design and therefore should be conducted before programming. 
As illustrated by the four-color theorem example, this can be a challenging mathematical task. 
However, many problems faced in practice are algorithmically trivial, i.e.~an algorithm which solves them is immediately obvious. For example, deciding whether a string contains another one as substring is algorithmically trivial. Formally proving the correctness of such an obvious algorithm is like cracking a nut with a sledgehammer. 
We suppose that often programmers do not struggle to find a correct algorithm for the problem at hand but to write a program which correctly implements this algorithm. Our method is aimed at overcoming this difficulty.

\subparagraph*{Outline of the method.}
The execution trace of an algorithm can be seen as a table where each column corresponds to a variable. The $i$-th row contains the values of the variables after $i$ execution steps of the algorithm for a certain input. The first step is to determine the set of variables used by the algorithm. Then an input is chosen and an execution trace for this input is written down. The next step is to generalize the literals in this table. More specifically, one has to determine how each value in row $i+1$ can be expressed in terms of the values in row $i$. After removing the literals from the table, the sequence of expressions that is left in each row corresponds to an operation executed by the algorithm. Certain rows correspond to the same operation. 
This information can be used to derive the control flow graph of the desired program. 
Assume the $i$-th row corresponds to the operation $\alpha(i)$. Then the graph has a vertex for each operation and for each $i$ there is an edge from $\alpha(i)$ to $\alpha(i+1)$. Stated differently, the sequence of operations in the table describes a path through the control flow graph. By repeating the previous steps for various inputs one will eventually arrive at the complete control flow graph of the program. Finally, it remains to determine the edge predicates, i.e.~under what circumstances does the program move from one program state to the next. 

The method works in such a way that the constructed program is consistent with all execution traces that were used to build it.
A one-to-one correspondence between program states and operations is assumed. Intuitively, this means each line of code in a program must be unique. 
While this is not necessarily the case, this can always be achieved by adding a state variable. 

\subparagraph*{Related work.}

Biermann and his colleagues explored the idea of deriving a program from execution traces in the 70's \cite{biertm,bierspeed,bier}. The goal was to simplify programming by delegating part of the intellectual burden to the computer: the programmer enters some execution traces and the computer finds the desired program. A similar approach is taken in \cite{ek} where only the control flow of the program is synthesized from traces.
In contrast, our method's purpose is to support the programmer's thought process and it does not require a computer to be effective.

\subparagraph*{Paper overview.}
In Section~\ref{sec:ff}, we describe a general formal framework for programming in terms of traces. 
We consider the synthesis problem (constructing a program from a set of traces) and show that the main obstacle is to determine what rows of the traces can be mapped to the same program state; this can be encoded as a coloring problem. We also explain how this framework can be related to imperative and functional programming languages. If one is primarily interested in the method then this section can be skipped.
In Section~\ref{sec:md}, we demonstrate our method for two examples and explain how it relates to the framework from Section~\ref{sec:ff}.  
The first example is a string matching problem where the algorithm is quite simple. The second example is to decide whether two rooted trees are isomorphic using only a logarithmic amount of space. In the first part we describe Lindell's algorithm \cite{lindell} which accomplishes this. Implementing this algorithm can be challenging, even for a seasoned programmer. Therefore it is a good example to showcase the strengths of our method.
%The method demonstration for the two examples can be understood without having read Section~\ref{sec:ff} before.
In Section~\ref{sec:con}, we recapitulate the advantages of our method compared to the ad hoc approach.
%and  mention what further steps are required to make it viable in software development.

\section{Trace-Based Programming}
\label{sec:ff}
This section is intended to provide a formal framework for programming methods where a program is built from traces. To present our method it would suffice to use the definition of a trace from the introduction: a description of how the values of the variables change after every execution step of a program for some input. However, a complete formalization would still require a definition of program. Instead, we define a notion of program first from which a more general definition of traces naturally follows. The advantage of this more general definition is that the considerations within this section are applicable beyond the scope of our method. In particular, we consider how to compute a small program from a set of traces and show that this problem is equivalent to finding a restricted coloring of a certain graph defined in terms of the traces. Consequently, any trace-based programming method must address how to obtain such a coloring. Additionally, we use this framework to relate our method to other ideas and areas such as automata theory.

A program is a directed graph which represents the control flow. Each vertex is associated with an operation and each edge is associated with a predicate (see Figure~\ref{fig:double}). When the program reaches a certain vertex the associated operation is executed and then the outgoing edge whose predicate is true dictates what vertex to visit next; we require that there can only be one such edge, i.e.~programs are deterministic. Execution ends when a vertex is reached such that no outgoing edge exists whose  predicate is true.	The possible operations and predicates are determined by what we call a model of computation. For example, the Turing machine is a model of computation and its operations are moving the tape head to the left or right and overwriting the current cell with a symbol from a given alphabet. Instruction set architectures such as x86 or MIPS are also models of computation. 

In the first subsection we describe the three elementary notions of our framework---models of computation, programs and traces---and give some examples. In the second subsection we characterize under what circumstances a set of traces admits a small program. Moreover, we show that any program can be represented as a finite set of traces. This implies that trace-based programming is a universal form of programming. 
In the last subsection we define the class of models of computation used by our method and show how programs in our sense can be translated into an imperative or functional programming language. 

\subsection{A Meta-Model of Computation}
\label{ss:mmoc}
Abstractly, a machine can be seen as a device which has a set of buttons and a set of indicator lamps. It resides in a particular state and reveals (partial) information about its current state through the indicator lamps, which are either on or off. An operator can affect the machine state by pushing a button. Its behavior is deterministic, i.e.~pushing a particular button in a particular state always has the same effect. The machine only changes its state when a button is pushed. The following definition formalizes this.

\begin{definition}
	A model of computation $\compmodel$ is a triple $(\cmspace,\opset,\predset)$ where 
	$\cmspace$ is a countable set, $\opset$ is a finite set of functions from $\cmspace$ to $\cmspace$ and $\predset$ is a finite set of functions from $\cmspace$ to $\{0,1\}$. The set $\cmspace$ is called state space of $\compmodel$ and an element $s$ of $\cmspace$ is a state of $\compmodel$; if $\compmodel$ is clear from the context we call $s$ a machine state. A function from $\opset$ is called an operation of $\compmodel$ and a function from $\predset$ is called a predicate of $\compmodel$. 
\end{definition}

We require the state space to be countable in order to guarantee that every machine state has a finite representation as string.
Operations and predicates need not be computable and thus certain models of computations might not be effective in the sense that they cannot be simulated by a Turing machine.
%there is no Turing machine which simulates them.

\begin{example}[$k$-counter machine]
	Let $k \in \N$. The model of computation $\cmcm$ is defined as follows. It has $\No^k$ as state space and for every $i \in [k]$ it has the operations:	
	\begin{align*}
	\text{`}\mathrm{R}i+1\text{'} &:= (x_1,\dots,x_k) \mapsto (x_1,\dots,x_{i-1},x_i+1,x_{i+1},\dots,x_k) \\
	\text{`}\mathrm{R}i-1\text{'} &:=  (x_1,\dots,x_k) \mapsto (x_1,\dots,x_{i-1},x_i-1,x_{i+1},\dots,x_k)
	\end{align*}
	where $x_i - 1$ is defined as $0$ if $x_i = 0$. For every $i \in [k]$ it has the predicate:
	$$  \text{`}\mathrm{R}i=0\text{'} := (x_1,\dots,x_k) \mapsto 1 \Leftrightarrow x_i = 0  $$		
\end{example}

\begin{definition}
	Let $\compmodel$ be a model of computation. 
	The state space graph $G(\compmodel)$ of $\compmodel$ is a directed graph defined as follows. It has the state space of $\compmodel$ as vertex set and there is an edge from $s$ to $s'$ if there exists an operation $g$ of $\compmodel$ such that $g(s) = s'$. A path $(s_1,\dots,s_n)$ in $G(\compmodel)$ with $n \geq 1$ is called a trace of $\compmodel$. 
    Let $\vec{s} = (s_0,s_1,\dots,s_n)$ be a sequence of machine states and let $\vec{g} = (g_1,\dots,g_n)$ be a sequence of operations of $\compmodel$. We say $(\vec{s},\vec{g})$ is an extended trace if $g_i(s_{i-1}) = s_i$ holds for all $i \in [n]$. The length of an extended trace $(\vec{s},\vec{g})$ is $|(\vec{s},\vec{g})| = |\vec{g}|$.
\end{definition}

\begin{example}
	The sequence $\vec{s} = ((1,2),(2,2),(3,2),(3,1))$ is a trace of  $\cmcm[2]$ whereas neither $((0,1),(2,1))$ nor $((3,4),(3,4))$ are since there is no operation which increments by two and there is no operation which leaves $(3,4)$ unchanged. The tuple $(\vec{s},\vec{g})$ is an extended trace where $\vec{g} = (\mathrm{R1}+1,\mathrm{R1}+1,\mathrm{R2}-1)$.
    The tuple $(((7,6)),())$ is also an extended trace. 
\end{example}

\begin{definition}
	Let $\compmodel$ be a model of computation. An $\compmodel$-program $P$ is a tuple $(G,v_0,\alpha,\beta)$ where $G$ is a directed graph, $v_0$ is a vertex of $G$ with in-degree 0, $\alpha$ is a function which maps every vertex of $G$ except $v_0$ to an operation of $\compmodel$ and $\beta$ is a function which maps every edge of $G$ to a $k$-ary boolean function where $k$ is the number of predicates of $\compmodel$. The graph $G$ is called control flow graph (CFG) of $P$ and the vertices of $G$ are called states of $P$ or program states. The vertex $v_0$ is called start state. For an edge $(u,v)$ of $G$ the boolean function $\beta(u,v)$ is called edge predicate.
\end{definition}

\begin{definition}
	Let $\compmodel$ be a model of computation. For an $\compmodel$-program $P$, a program state $v$ and a machine state $s$ we write $P(s,v)$ to denote the trace produced when executing $P$ starting from program state $v$ with machine state $s$. We write $P[s,v]$ to denote the sequence of operations that occurs during this execution. We write $P(s)$ and $P[s]$ to denote $P(s,v_0)$ and $P[s,v_0]$ where $v_0$ is the start state of $P$. If the execution is ambiguous, i.e.~at some point during execution it is not clear which state to visit next because more than one of the edge predicates is true, then these expressions are undefined.  
    We say $P$ terminates on $s$ if $P(s)$ is finite. 
\end{definition}

\begin{figure}
    \begin{center}
        \begin{tikzpicture}[shorten >=1pt,auto,node distance=1.2cm,
main node/.style={draw}]

%\path[use as bounding box] (-1.5,-4.8) rectangle (14.8,4.1);

\newcommand*{\xlen}{1cm}%
\newcommand*{\ylen}{0.6cm}%

\newcommand*{\xtlen}{1cm}%
\newcommand*{\ytlen}{0.4cm}%

\node[] (GHOST) at (-1.1,0) {};
\node[main node] (INIT) at (0,0) {Start};
%\node[main node,fill={black!10}] (INIT) at (0,0) {START};

\node[main node, above right = \ylen and 3*\xlen of INIT] (A1) {$\mathrm{R2}-1$};
\node[main node, below right = \ylen and 3*\xlen of INIT] (A3) {$\mathrm{R1}-1$};
\node[main node, right = \xlen of A3] (A4) {$\mathrm{R2}+1$};
\node[main node, right = \xlen of A4] (A5) {$\mathrm{R2}+1$};

\path[->]
(GHOST) edge (INIT)
(INIT) edge node[midway, above, sloped]{\footnotesize $\mathrm{R2}>0$} (A1)
(INIT) edge node[midway, above, sloped]{\footnotesize $\mathrm{R2}=0 \: \& \: \mathrm{R1}>0$} (A3)
(A1) edge node[midway]{\footnotesize $\mathrm{R2}=0 \: \& \: \mathrm{R1}>0$} (A3)
(A3) edge node[midway]{\footnotesize $\top$} (A4)
(A4) edge node[midway]{\footnotesize $\top$} (A5)

(A1) edge[loop] node[midway, above, sloped]{\footnotesize $\mathrm{R2}>0$} (A1)
(A5) edge[bend left=80] node[midway, above, sloped]{\footnotesize $\mathrm{R1} > 0$} (A3)
;

\end{tikzpicture}
    \end{center}
    \caption{2-counter machine program which doubles the value of register 1}
    \label{fig:double}	
\end{figure}
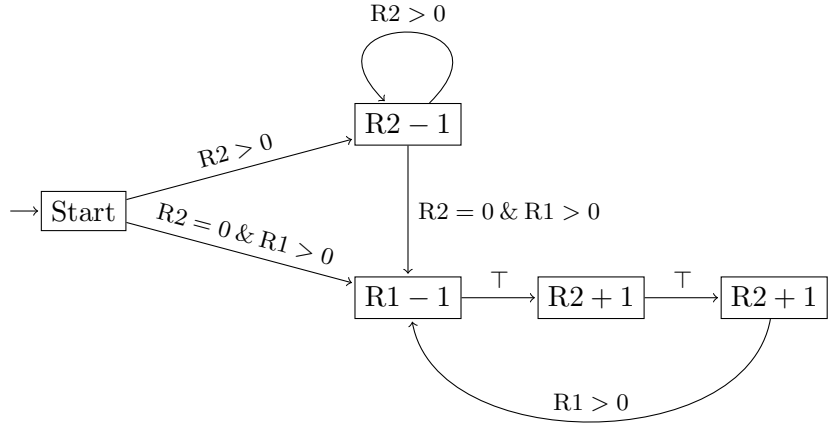

\begin{example}
	Figure~\ref{fig:double} shows a $\cmcm[2]$-program which clears register 2 and then doubles register 1 and writes the result to register 2. The node labeled `Start' corresponds to $v_0$; its ingoing edge is just decorative, i.e.~it does not appear in the CFG. The symbol $\top$ represents the boolean function which is always true.  For example, on input $(2,1)$ the program behaves as follows:
    \begin{align*}
     P((2,1)) &= ((2,1),(2,0),(1,0),(1,1),(1,2),(0,2),(0,3),(0,4)) \\
     P[(2,1)] &= (\mathrm{R2}-1,\mathrm{R1}-1,\mathrm{R2}+1,\mathrm{R2}+1,\mathrm{R1}-1,\mathrm{R2}+1,\mathrm{R2}+1) 
    % P((2,0),v) &= ((2,0),(1,0),(1,1),(1,2),(0,2),(0,3),(0,4)) \\
   	% P[(2,0),v] &= (\mathrm{R1}-1,\mathrm{R2}+1,\mathrm{R2}+1,\mathrm{R1}-1,\mathrm{R2}+1,\mathrm{R2}+1) 
    \end{align*}
    Let $v$ be the program state with operation $\mathrm{R2}-1$.
    The sequences $P((2,0),v)$ and $P[(2,0),v]$ are the same as $P((2,1))$ and $P[(2,1)]$ minus the first element. 
%    Let $v$ be the program state with operation $\mathrm{R2}-1$. The sequence $P((2,0),v) = ((2,0),(1,0),\dots)$ is the same as $P((2,1))$ without the first element and $P[(2,0),v] = (\mathrm{R1}-1,\mathrm{R2}+1,\dots)$ is the same as $P[(2,1)]$ without the first element.     
\end{example}

\begin{definition}
    Let $\compmodel$ be a model of computation and let $P,P'$ be $\compmodel$-programs. We say $P$ and $P'$ are equivalent if $P(s) = P'(s)$ holds for all machine states $s$. 
\end{definition}

To illustrate this notion of program equivalence, reconsider the example of a machine $\compmodel$ with buttons and indicator lamps. Suppose you can observe the inner workings of $\compmodel$ and thus also know in what (machine) state it resides in at any moment. Your task is to discern two operators $P,P'$. To accomplish this, you are allowed to set $\compmodel$ to an initial state and ask each operator to execute their routine. During execution you can record how the machine state changes after every action by the operator. This sequence is a trace. You can only distinguish the two operators if their traces for some initial machine state differs. Notice that the operators might push different buttons at a certain point, but if both buttons result in the same subsequent machine state you will not notice this.
% Replacing $P(s) = P(s')$ with $P[s] = P'[s]$ yields a finer equivalence relation.

\begin{example}
    Consider the program $P$ shown in Figure~\ref{fig:double}. Let $P'$ be the program which results from $P$ by adding an additional program state $x$ with operation R1-1 and reroute the outgoing edge from the rightmost program state to $x$. Additionally, add an edge from $x$ to the ``left R2+1'' with edge predicate $\top$. The programs $P$ and $P'$ produce the same traces and thus are equivalent.
    \label{ex:progeq}
\end{example}

\subparagraph*{Relation to automata theory and logic.}
In automata theory the models of computation that underlie finite state automata (FSA), pushdown automata (PDA) and Turing machines (TM) are defined implicitly. To illustrate what we mean by that let us consider Turing machines as an example. Informally, a TM can be described as a device that has 1) a two-way infinite tape which consists of cells and each cell can contain a symbol from the tape alphabet $\Gamma$, 2) a tape head used to read and modify the content of the cells, and 3) a finite state machine which guides its behaviour. In our terminology 1) and 2) describe the model of computation and 3) is a program. The separation of models of computation and programs is also suggested in \cite{dana}. The standard formalization of a TM as a tuple combines all three parts. 
A state in the sense of automata theory is a program state and a configuration is the combination of a machine and program state.

The model of computation for a Turing machine with tape alphabet $\Gamma$ could be defined as follows. Its state space is $\Gamma^* \times \Gamma \times \Gamma^*$ where the first (third) component describes the tape content to the left (right) of the head and the second component is the content of the current cell. The operations are moving the head to the left or right and setting the content of the current cell to $a$ for all $a \in \Gamma$. For all $a \in \Gamma$ there is a predicate which holds iff the current cell contains $a$. Observe that instead of $\Gamma^* \times \Gamma \times \Gamma^*$ we could choose $\Gamma^* \times \Gamma^* \times \Gamma$ as state space and modify operations and predicates accordingly to reflect this permutation. Even though the two models of computation are not identical, they are equivalent in the sense that both describe the same model of computation up to some immaterial difference in encoding of the state space. This notion of equivalence can be formally captured by regarding a model of computation as relational structure where the state space is the universe. Two models of computation are equivalent if they are isomorphic as relational structures. A formula $\varphi(s,t)$ over such a structure can be interpreted as a program. It describes how to move from one machine state to the next and terminates when for a given $s$ there exists no $t$ such that $\varphi(s,t)$ holds. For example, the following propositional formula over $\cmcm[2]$ adds the double of the first register to the second and sets the first register to zero in the process: 
$$ \Big( \neg \text{`R1=0'}(s) \rightarrow t =  \text{`R2+1'}(\text{`R2+1'}(\text{`R1-1'}(s))) \Big) \wedge   \text{`R1=0'}(s) \rightarrow \bot  $$
%The sequences of machine states induced by a formula are not necessarily traces because multiple operations can be executed at once.  We remark that every $\cmcm[k]$-program can be expressed as formula over $\cmcm[k+1]$ where the additional register represents the program state. Axiomatically defining a model of computation  provides a basis for deductively reasoning about programs described by formulas. 
A similar way of using formulas to represent programs is used in \cite{lamport}.
An algorithm can be represented as a model of computation and a program in this model.
A related way of representing algorithms are abstract state machines \cite{gurevich,gurevich2}.

\subsection{Consistency}

A set of traces $T$ is consistent if there exists a program $P$ which produces all traces in $T$. We also say $P$ witnesses that $T$ is consistent. 
We call $T$ $k$-consistent if $P$ has at most $k$ states. We show how to find a small program which produces a given set of traces and that every program can be specified by a finite subset of its traces.

A fixed model of computation $\compmodel$ is assumed in the following statements. The full proofs of this section can be found in the appendix.

\begin{definition}
    Let $k \in \N$. A set of traces $T$ is ($k$-)consistent if there exists a program $P$ (with at most $k$ states) such that $P(s_0) = \vec{s}$ 
    holds for all traces $\vec{s} = (s_0,\dots)$ in $T$.
\end{definition}

\begin{definition}
	Let $\{f_1,\dots,f_k\}$ be the set of predicates of $\compmodel$. The predicate word $\predseq(s)$ of a machine state $s$ is the binary string $f_1(s)\dots f_k(s)$. 
	For two machine states $s,s'$ we say they are indistinguishable, in symbols $s \sim s'$, if  $\predseq(s) = \predseq(s')$.
	Let $\vec{s} = (s_0,\dots), \vec{t} = (t_0,\dots)$ be traces. We say $\vec{s} \sim_0 \vec{t}$ if $s_0 \sim t_0$. 
	For $\vec{s} = (s_0,s_1,\dots,s_n)$ let $(\vec{s})' = (s_1,\dots,s_n)$.
\end{definition}

Stated differently, two machine states $s,s'$ are indistinguishable iff there exists no predicate $f$ such that $f(s) \neq f(s')$. The next action of a program solely depends on its current program state and the predicate word of its current machine state.   

\begin{proposition}
	A set of traces $T$ is consistent iff for all $T'$ in the quotient set $T / \sim_0$ it holds that
    \underline{either} every trace in $T'$ has one element \underline{or} there exists an operation $g$ such that for all traces $\vec{s} = (s_0,s_1,\dots)$ in $T'$ it holds that $g(s_0) = s_1$ and $|\vec{s}| > 1$  and $\left\{ (\vec{s})' \mid \vec{s} \in T' \right\}$ is consistent.
    \label{prop:ttopt}
\end{proposition}

Implicit in the ``$\Leftarrow$''-direction of the proof of Proposition~\ref{prop:ttopt} is a recursive algorithm that, given a consistent set of traces $T$, outputs a program $P_T$ which witnesses this and whose control flow graph is a tree. 
%There is a certain degree of freedom in the construction of $P_T$, i.e.~there might be multiple operations $g$ such that $g(s_0) = s_1$ holds for all traces $\vec{s} = (s_0,s_1,\dots,s_n)$ in $T'$.

\begin{definition}
	Let $k \in \N$.	A set of extended traces $T$ is ($k$-)consistent iff there exists a program $P$ (with at most $k$ states) such that $P(s_0) = \vec{s}$ and $P[s_0] = \vec{g}$ holds for all $(\vec{s},\vec{g}) \in T$ where $s_0$ is the first element of $\vec{s}$.
\end{definition}

The problem of determining whether a set of traces $T$ is $k$-consistent is equivalent to determining whether $T$ can be transformed into a set of extended traces $T'$ such that $T'$ is $k$-consistent. 

%\begin{proposition}
%	Let $k \in \N$. A set of traces $T$ is $k$-consistent iff there exists a function $\sigma$ which maps every trace of $T$ to a sequence of operations such that $(\vec{s},\sigma(\vec{s}))$ is an extended trace and $\left\{  (\vec{s},\sigma(\vec{s})) \mid \vec{s} \in T  \right\}$ is $k$-consistent. 
%	\label{prop:rttet}
%\end{proposition}
%
%Stated differently, the problem of determining whether a set of traces $T$ is $k$-consistent reduces to determining whether $T$ can be transformed into a set of extended traces $T'$ such that $T'$ is $k$-consistent. 

\begin{definition}
	A line is a triple $(i,s,g)$ where $i \in \No$, $s$ is a machine state and $g$ is either an operation or the special symbol $\emptyset$. Let $(\vec{s},\vec{g})$ be an extended trace with $\vec{s} = (s_0,s_1,\dots,s_n), \vec{g} = (g_1,\dots,g_n)$. The $0$-th line of $(\vec{s},\vec{g})$ is defined as $(0,s_0,\emptyset)$. For $i \in [n]$ the $i$-th line of $(\vec{s},\vec{g})$ is defined as $(i,s_i,g_i)$. %We define $|(\vec{s},\vec{g})|$ as $n$.
\end{definition}

\begin{definition}
	Let $X$ and $Y$ be extended traces and $n = |X|, m = |Y|$. 
	We define a binary matrix $U_{i,j}$ with $i \in \{0,\dots,n+1\}$ and $j \in \{0,\dots,m+1\}$ as follows. The matrix has the following fixed entries:
	%alternative to nicematrix
	\vspace{-0.2cm}
	\begin{center}
		\includegraphics[scale=0.3]{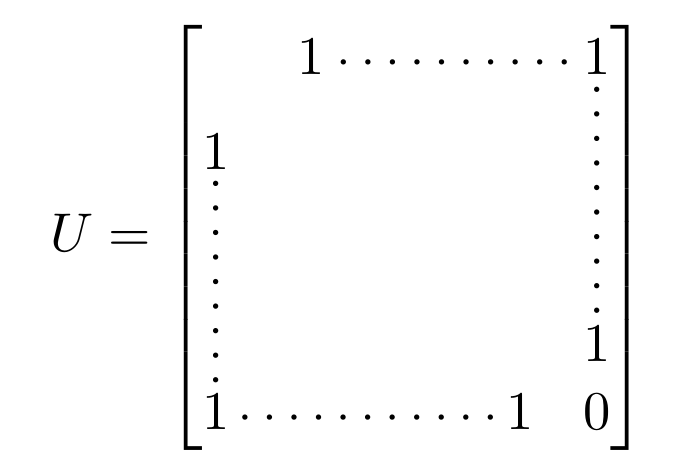}	
	\end{center}
	\vspace{-0.3cm}
%	$$ U = \begin{bNiceMatrix}
%	  & 1 & \Cdots &  & 1 \\
%  	 1 & & & & \Vdots \\ 	
% 	 \Vdots &  &  &  & \Vdots \\
% 	 \Vdots &  &  &  & 1 \\ 	 
% 	 1 & \Cdots & \Cdots & 1 & 0  	  	  
%	\end{bNiceMatrix} $$
	For $(i,j) \in ([n] \times [m]) \cup \{(0,0)\}$ the entry $U_{i,j}$ is defined as follows. Let $(i,s,g)$ be the $i$-th line of $X$ and let $(j,t,h)$ be the $j$-th line of $Y$. Let $U_{i,j} = 1$ iff $g \neq h$ or $s \sim t$ and $U_{i+1,j+1}=1$. For $i \in \{0,\dots,n\}$ and $j \in \{0,\dots,m\}$ we say the $i$-th line of $X$ and the $j$-th line of $Y$ are unmergeable if $U_{i,j} = 1$.
\end{definition}

The unmergeability relation intuitively means the following. Suppose you have two extended traces $X,Y$ with initial machines states $s,t$ respectively and a program $P$ which produces both. If the $i$-th line of $X$ and the $j$-th line of $Y$ are unmergeable then $P$ cannot reach the same program state when executing input $s$ for $i$ steps and $t$ for $j$ steps.  

\begin{definition}
	Let $T = \{ X_1,\dots,X_n \}$ be a set of extended traces. The undirected graph $G(T)$ is defined as follows. It has a vertex $v_{i,j}$ for every $i \in [n]$ and $j \in \{0,\dots,|X_i|\}$ and two vertices $v_{i,j}$ and $v_{i',j'}$ are adjacent if the $j$-th line of $X_i$ and the $j'$-th line of $X_{i'}$ are unmergeable.
%    
%    Let $k \in  \N$. We say a function $c \colon V(G(T)) \rightarrow [k]$ is a restricted $k$-coloring of $G(T)$ if 
%    no two adjacent vertices in $G(T)$ receive the same color
%    and, additionally, for all $i,i' \in [n]$ and $j \in \{0,\dots,|X_i|\}$, $j' \in \{0,\dots,|X_{i'}|\}$ it holds that if 
%    $$ c(v_{i,j}) = c(v_{i',j'}) \wedge s(v_{i,j}) \sim s(v_{i',j'}) $$
%    then
%    either 
%    $$ j < |X_i| \wedge j' < |X_{i'}| \wedge c(v_{i,j+1}) = c(v_{i',j'+1}) $$
%    or 
%    $$  j = |X_i| \wedge j' = |X_{i'}| $$
%    %$j < |X_i|$, $j' < |X_{i'}|$ and $c(v_{i,j+1}) = c(v_{i',j'+1})$ or $j = |X_i|$ and $j' = |X_{i'}|$,
%    where $s(\cdot)$ denotes the machine state that is associated with a vertex via its corresponding line. 
	\label{def:trace_graph}
\end{definition}

\begin{theorem}
	Let $k \in \N$. A set of extended traces $T$ is $k$-consistent iff $T$ is consistent and $G(T)$ has a restricted $k$-coloring.
	\label{thm:color}
\end{theorem}
\begin{proof}[Proof idea]
	A program $P$ which shows that $T$ is $k$-consistent corresponds to a (restricted) $k$-coloring of $G(T)$: associate each line of every trace in $T$ with the program state of $P$ reached during execution. This means every program state corresponds to a color. This correspondence can also be used to construct a program from a $k$-coloring. However, an arbitrary coloring might yield a non-deterministic program. Thus, a certain restriction on admissible colorings must be made to ensure that the resulting program is deterministic, see Definition~\ref{def:crestr}.
\end{proof}

In summary, finding a small program which produces a given set of traces $T$ is equivalent to finding an extended version $T'$ of $T$ and a small restricted coloring of $G(T')$. Any algorithm which computes a small program for an extended set of traces $T$ must do so by (implicitly) computing a restricted coloring of $G(T')$. 
%todo: subsume
%Implicit in the proof is an algorithm that, given a set of extended traces $T$ and $k \in \N$, either outputs a program which witnesses that $T$ is $k$-consistent or determines that $T$ is not $k$-consistent. 
%First, check whether $T$ is consistent. If this is the case, try to compute a restricted $k$-coloring $c$ of $G(T)$. If such a coloring does not exist then $T$ is not $k$-consistent. Otherwise, construct the desired program using $c$ (see the ``$\Leftarrow$''-direction of the previous proof). 
%In summary, finding a small program which produces a given set of traces $T$ is equivalent to finding an extended version $T'$ of $T$ and a small restricted coloring $c$ of $G(T')$. A small program $P$ can be constructed from $T'$ and $c$. Conversely, a program $P$ with $k$ states which produces $T$ can be used to convert $T$ into an extended version $T'$ and a $k$-coloring $c$ of $G(T')$. Stated differently, any algorithm which computes a small program from a given set of traces $T$, must do so by implicitly computing an extended version $T'$ of $T$ and a restricted coloring $c$ of $G(T')$. 

\begin{proposition}
	Any program can be represented as a finite set of traces up to equivalence. More formally, for every program $P$ there exists a finite set of traces $T$ such that $P$ is equivalent to $P'$ where $P'$ is a minimal (w.r.t.~program states) program which witnesses that $T$ is consistent. 
	\label{prop:ftrepr}    
\end{proposition}  

The size required to express a program solely in terms of its traces can be quite large. A simple way to decrease this size is to use extended traces. But even then the number of traces required to express a simple program might quickly become impractical. For example, consider a model of computation with $p$ predicates. Suppose you want to specify a program which goes from start state $v_0$ to $v_1$ if a certain predicate holds and to $v_2$ if that predicate does not hold. Specifying this behavior might already require $2^p$ (extended) traces.
% if for every string $q \in \{0,1\}^p$ there is a machine state whose predicate word equals $q$. Traces can be extended as follows to resolve this.
 %already requires $2^p$ (extended) traces if for every string $q \in \{0,1\}^p$ there is a machine state whose predicate word equals $q$. 
Traces can be extended as follows to resolve this.
Let $(\vec{s},\vec{g})$ be an extended trace with $\vec{s} = (s_0,s_1,\dots,s_n)$. Add $\vec{F} = (F_1,\dots,F_n)$ to the extended trace where $F_i$ is a subset of predicates with the following meaning. One goes from $s_i$ to $s_{i+1}$ solely based on the truth values of the predicates in $F_{i+1}$, i.e.~the predicates outside of $F_{i+1}$ are irrelevant. With this kind of information it only requires two traces to express the previously described behavior. This might yield acceptable specification sizes in practice and has been used in \cite{bier} as basis for specifying programs. 

Our method diverges from the concept of specifying a program completely in terms of its traces. Instead, we only use extended traces to obtain the control flow graph of the program. The edge predicates are added separately.

\subsection{Relation to Programming Languages}
A basic unit of abstraction in imperative and functional languages are functions.
We explain how a function can be regarded as a program in an implicit model of computation.
This makes it clear how trace-based programming can be used in conjunction with such languages. 

A function $\texttt{f}$ consists of a declaration and a body.
The declaration defines $\texttt{f}$'s return type and parameters, the body defines its behavior. 
In an imperative language $\texttt{f}$ can modify the value of variables that are within its scope using assignments.  
An assignment consists of the variable to be modified and an expression. An expression is a composition of built-in and self-defined functions. 
The body of $\texttt{f}$ consists of variable declarations, assignments, control structures and other statements with side-effects (print, open a socket, \dots). Let us assume $\texttt{f}$ is pure. This means its body does not contain statements of the last category.  
The model of computation $\compmodel(\texttt{f})$ implicitly described by  $\texttt{f}$ looks as follows. A machine state describes the value of each variable that is within the scope of $\texttt{f}$. Each assignment in  $\texttt{f}$ is an operation and each boolean expression is a predicate. The body of $\texttt{f}$ minus the variable declarations can be expressed as $\compmodel(\texttt{f})$-program. 
The purpose of our method is to systematically construct the body of $\texttt{f}$.

In a purely functional language there are no assignments, and control structures are modeled as functions instead of statements. For example, $\texttt{if}(x,y,z)$ is a function which evaluates to $y$ if the boolean expression $x$ holds and $z$ otherwise. The body of a function is just an expression. 
Let us consider how an $\compmodel$-program $P$ can be implemented in such a language. First, $\compmodel$ has to be implemented. This can be done by implementing each of its operations and predicates as function. 
To implement $P$, define a function $\texttt{f}_v \colon S \rightarrow S$ for each program state $v$ where $S$ is the state space. The function $\texttt{f}_v$ `executes' the operation associated with $v$ and then calls the function of the next program state. 
For example, consider the program in Figure~\ref{fig:double}. Let $v$ be the state with operation $\mathrm{R2-1}$. Then $\texttt{f}_v(s)$ is defined by the following expression:
\begin{center}
	\begin{tikzpicture}[sibling distance=20pt]
\Tree [.{\texttt{if}} [.{$\mathrm{R2} > 0$} {$s'$} ]
[.{$\texttt{f}_v$} {$s'$} ] 
[.{\texttt{if}} 
	[.{$\wedge$} 
		[.{$\mathrm{R2} = 0$} {$s'$} ]	
		[.{$\mathrm{R1} > 0$} {$s'$} ]
	]			
	[.{$\texttt{f}_w$} {$s'$} ]
	[.{$s'$} ]
]
]
\end{tikzpicture}

\end{center}
where $s' = \text{`}\mathrm{R2-1}\text{'}(s)$ and $w$ is the program state with operation $\mathrm{R1-1}$. In an imperative language a program can be implemented as a series of code fragments $\texttt{S}_v$ for each program state $v$:
\begin{align*}
	\texttt{S}_v : \: \: &  s \leftarrow \text{`}\mathrm{R2-1}\text{'}(s) \\
	& \text{if } \text{`}\mathrm{R2>0}\text{'}(s) \text{ then goto } \texttt{S}_v \\ 
	& \text{if } \text{`}\mathrm{R2=0}\text{'}(s) \wedge \text{`}\mathrm{R1>0}\text{'}(s) \text{ then goto } \texttt{S}_w \\
	& \text{goto } \texttt{END} \\
\end{align*}
Both representations of $P$ carry exactly the same information as $P$ itself. 
Consequently, no meaningful distinction between `imperative' and `functional' programs can be made in the context of pure functions.
More specifically, it is impossible to distinguish between the functional and imperative translation of $P$ without referencing representational artifacts. 
Thus, programming in a functional language is not significantly different from doing it in an imperative one.

\subparagraph*{Virtual Machines.} 
Our programming method uses the following class of models of computation.
Consider the following variable and function declarations:
$$\mathrm{int} \: x, \mathrm{int} \: y, \mathrm{float} \: z, \mathrm{bool} \: p, \mathrm{bool} \: q   \: \: \: ; \: \: \: \mathrm{DIV} \colon \mathrm{int} \times \mathrm{int} \rightarrow \mathrm{float} , \: 
\mathrm{XOR} \colon \mathrm{bool} \times \mathrm{bool} \rightarrow \mathrm{bool} $$
The assignments `$z \leftarrow \mathrm{DIV}(x,y)$' and `$p \leftarrow \mathrm{XOR}(p,q)$' are valid because the types of the variables match the function signatures. 
There are 12 valid assignments induced by the above set of variables and functions.
A parallel assignment is a subset of valid assignments such that every variable occurs at most once on the left-hand side. The empty set corresponds to the operation which leaves the state unchanged.

\begin{definition}(Virtual Machine)
	A type is a countable set. A variable consists of a unique name and a type. 
	Let $V$ be a finite set of variables and let $F$ be a finite set of functions.
	The model of computation $\cmvm$ is defined as follows. Its state space is the Cartesian product of the multiset of types that occur in $V$.
	It has the set of parallel assignments induced by $V$ and $F$ as operations.
	Every $k$-ary function from $F$  with codomain $\{0,1\}$ and matching sequence of $k$ variables from $V$ correspond to a predicate of $\cmvm$.
\end{definition}

\begin{example}
	Consider the set of variables $V = \{(x,\mathrm{int}),(y,\mathrm{int}),\dots\}$  and functions $F=\{\mathrm{DIV}, \mathrm{XOR}\}$ from before. 
	The types are defined as: $\mathrm{int} = \mathbb{Z}$, $\mathrm{float} = \mathbb{Q}$ and $\mathrm{bool} = \{0,1\}$. The state space of $\cmvm$ is $\mathbb{Z}^2 \times \mathbb{Q} \times \{0,1\}^2$. It has $5^3$ parallel assignments as operations and the two predicates $\mathrm{XOR}(p,q)$ and $\mathrm{XOR}(q,p)$.
\end{example}

In our method, an algorithm $A$ is implemented as $\cmvm$-program where $V$ is the set of variables used by $A$, and $F$ is the set of functions assumed to be available such as built-in functions  of the target programming language and compositions thereof.

\section{Method Demonstration}
\label{sec:md}
We demonstrate our method by applying it to implement two algorithms and explain how it relates to the formal framework from the previous section.
The first example is a simple algorithm on strings which is used to introduce the method.
The second one is a more involved algorithm used to decide isomorphism of trees. 
We suggest that the reader tries to implement the second algorithm themselves using their preferred method in order to determine the relative utility of our method. Our method's utility is also relative to the skill level of the prorgammer. While a skilled programmer might find it trivial to implement the first algorithm (and the application of our method useless in that case), a less skilled one might already profit from its use.   

\subsection{Confidential String Matching}

\textbf{Problem description}. You are given two arrays of strings $A$ and $B$ as input where $A$ has $n$ elements and $B$ has $m$ elements. Decide whether $A[1] \dots A[n] = B[1] \dots B[m]$. 
For example, if $A=[ab,ab,a]$ and $B=[aba,ba]$ the output should be `true' since the concatenation of the strings in $A$ and $B$ both equal $ababa$. The trivial solution would be to concatenate the strings in $A$ and $B$ and check for equality. However, there is a twist. The strings in the input arrays $A$ and $B$ are confidential. In order to protect their confidentiality you are only granted indirect access. Namely, your program can only access the following information:
\begin{itemize}
	\item the number of elements in $A$ and $B$
	\item the length of the strings in $A$ and $B$ 
	\item whether $\mathrm{substr}(A[i],x,z) = \mathrm{substr}(B[j],y,z)$ (you supply the integers $i,j,x,y,z$ and get the boolean result)
\end{itemize}
where $\mathrm{substr}(s,o,l)$ denotes the substring of $s$ which starts at offset $o$ and has length $l$. Since the third operation is quite costly you shall not use more than $n+m$ of those queries for any input.

\begin{figure}[!b]
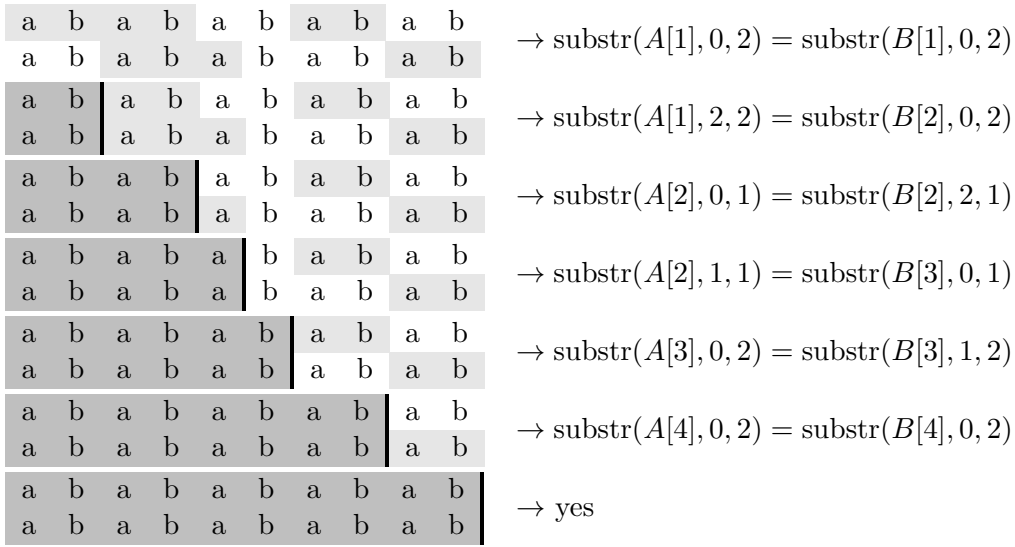

	\begin{center}
		\begin{tabular}{l l}

\begin{tabular}{ c c c c c c c c c c  }
    \cellcolor{black!10}a & \cellcolor{black!10}b & \cellcolor{black!10}a & \cellcolor{black!10}b &  a & b & \cellcolor{black!10}a & \cellcolor{black!10}b & a & b  \\
    a & b & \cellcolor{black!10}a & \cellcolor{black!10}b & \cellcolor{black!10}a & b & a & b & \cellcolor{black!10}a & \cellcolor{black!10}b  \\
\end{tabular} 
& $\rightarrow \mathrm{substr}(A[1],0,2) = \mathrm{substr}(B[1],0,2)$ \vspace{0.2em} \\

\begin{tabular}{ c c ?{0.5mm} c c c c c c c c  }
    \cellcolor{black!25}a & \cellcolor{black!25}b & \cellcolor{black!10}a & \cellcolor{black!10}b &  a & b & \cellcolor{black!10}a & \cellcolor{black!10}b & a & b  \\
    \cellcolor{black!25}a & \cellcolor{black!25}b & \cellcolor{black!10}a & \cellcolor{black!10}b & \cellcolor{black!10}a & b & a & b & \cellcolor{black!10}a & \cellcolor{black!10}b  \\
\end{tabular}
& $\rightarrow \mathrm{substr}(A[1],2,2) = \mathrm{substr}(B[2],0,2)$ \vspace{0.2em} \\

\begin{tabular}{ c c c c ?{0.5mm} c c c c c c  }
    \cellcolor{black!25}a & \cellcolor{black!25}b & \cellcolor{black!25}a & \cellcolor{black!25}b &  a & b & \cellcolor{black!10}a & \cellcolor{black!10}b & a & b  \\
    \cellcolor{black!25}a & \cellcolor{black!25}b & \cellcolor{black!25}a & \cellcolor{black!25}b & \cellcolor{black!10}a & b & a & b & \cellcolor{black!10}a & \cellcolor{black!10}b  \\
\end{tabular}
& $\rightarrow \mathrm{substr}(A[2],0,1) = \mathrm{substr}(B[2],2,1)$ \vspace{0.2em} \\

\begin{tabular}{ c c c c c ?{0.5mm} c c c c c  }
    \cellcolor{black!25}a & \cellcolor{black!25}b & \cellcolor{black!25}a & \cellcolor{black!25}b &  \cellcolor{black!25}a & b & \cellcolor{black!10}a & \cellcolor{black!10}b & a & b  \\
    \cellcolor{black!25}a & \cellcolor{black!25}b & \cellcolor{black!25}a & \cellcolor{black!25}b & \cellcolor{black!25}a & b & a & b & \cellcolor{black!10}a & \cellcolor{black!10}b  \\
\end{tabular}
& $\rightarrow \mathrm{substr}(A[2],1,1) = \mathrm{substr}(B[3],0,1)$ \vspace{0.2em} \\

\begin{tabular}{ c c c c c c ?{0.5mm} c c c c  }
    \cellcolor{black!25}a & \cellcolor{black!25}b & \cellcolor{black!25}a & \cellcolor{black!25}b &  \cellcolor{black!25}a & \cellcolor{black!25}b & \cellcolor{black!10}a & \cellcolor{black!10}b & a & b  \\
    \cellcolor{black!25}a & \cellcolor{black!25}b & \cellcolor{black!25}a & \cellcolor{black!25}b & \cellcolor{black!25}a & \cellcolor{black!25}b & a & b & \cellcolor{black!10}a & \cellcolor{black!10}b  \\
\end{tabular}
& $\rightarrow \mathrm{substr}(A[3],0,2) = \mathrm{substr}(B[3],1,2)$ \vspace{0.2em} \\

\begin{tabular}{ c c c c c c c c ?{0.5mm} c c  }
    \cellcolor{black!25}a & \cellcolor{black!25}b & \cellcolor{black!25}a & \cellcolor{black!25}b &  \cellcolor{black!25}a & \cellcolor{black!25}b & \cellcolor{black!25}a & \cellcolor{black!25}b & a & b  \\
    \cellcolor{black!25}a & \cellcolor{black!25}b & \cellcolor{black!25}a & \cellcolor{black!25}b & \cellcolor{black!25}a & \cellcolor{black!25}b & \cellcolor{black!25}a & \cellcolor{black!25}b & \cellcolor{black!10}a & \cellcolor{black!10}b  \\
\end{tabular}
& $\rightarrow \mathrm{substr}(A[4],0,2) = \mathrm{substr}(B[4],0,2)$ \vspace{0.2em} \\
\begin{tabular}{ c c c c c c c c c c ?{0.5mm} }
    \cellcolor{black!25}a & \cellcolor{black!25}b & \cellcolor{black!25}a & \cellcolor{black!25}b &  \cellcolor{black!25}a & \cellcolor{black!25}b & \cellcolor{black!25}a & \cellcolor{black!25}b & \cellcolor{black!25}a & \cellcolor{black!25}b  \\
    \cellcolor{black!25}a & \cellcolor{black!25}b & \cellcolor{black!25}a & \cellcolor{black!25}b & \cellcolor{black!25}a & \cellcolor{black!25}b & \cellcolor{black!25}a & \cellcolor{black!25}b & \cellcolor{black!25}a & \cellcolor{black!25}b  \\
\end{tabular}
& $\rightarrow$ yes   \\
\end{tabular}
\\
\vspace{2mm}
Input: $A=[\text{abab, ab, ab, ab}]$, $B=[\text{ab, aba, bab, ab}]$ 
	\end{center}
	\caption{Algorithm for the confidential string matching problem}
	\label{fig:csma}
\end{figure}

We show how to implement the obvious algorithm for this problem which is exemplified in Figure~\ref{fig:csma}. This figure can be seen as a graphical representation of a trace. Observe that along with the problem description this single trace already provides sufficient information to communicate the algorithm, i.e.~it should be clear how to create an analogous graphical trace for another input. 

\subparagraph*{Step 1: Determine the variables used by the algorithm.} 
For this algorithm we need to remember the current strings under consideration in array $A$ and $B$ ($ca,cb$), how much of them has been read (offsets $oa,ob$), the length for the next substring call ($l$) and a boolean variable for the result ($r$). This is also exactly the information encoded in Figure~\ref{fig:csma}.
For example, in the second step of Figure~\ref{fig:csma} the algorithm is at the first string of $A$ ($ca = 1$) with an offset of 2 ($oa=2$), at the second string of $B$ ($cb = 2$) with no offset ($ob = 0$) and the length for the next substring call is $l=2$. The input arrays $A$ and $B$ are also variables that the algorithm needs to access. However, since the algorithm does not modify them we will not explicitly mention them in the following. 

\subparagraph*{Step 2: Select inputs and produce traces.} 
A trace is a description of the contents of the variables after each execution step. It can be  represented as a table which has a column for each variable and the $i$-th row describes the contents of the variables after having executed $i$ steps of the algorithm. %The traces serve as basis for building the desired program. 

Let us consider the trace for the input in Figure~\ref{fig:csma}. 
The first step of the algorithm is to initialize the variables $ca$, $cb$, $oa$ and $ob$. After the first step it holds that $(ca,cb,oa,ob,l,r) = (1,1,0,0,\cdot,\cdot)$ since the algorithm starts with the first string in $A$ and $B$ and no offset (`$\cdot$' indicates the corresponding variable has not been assigned a value by the algorithm in the current step; $l$ and $r$ are undefined so far). The next step is to determine the length for the next substring call, i.e.~the number of characters until the end of one of the two current strings is reached. In this case this is two and thus after the second step it holds that $(ca,cb,oa,ob,l,r) = (\cdot,\cdot,\cdot,\cdot,2,\cdot)$. Since the two substrings match, i.e.~$\mathrm{substr}(A[1],0,2) = \mathrm{substr}(B[1],0,2)$, the algorithm increases the offset of the current string in $A$ and moves on to the next string in $B$ and resets the offset. This means after the third step it holds that $(ca,cb,oa,ob,l,r) = (\cdot,2,2,0,\cdot,\cdot)$. In the fourth step the algorithm determines the length for the next substring call which is two again and therefore after this step it holds that $(ca,cb,oa,ob,l,r) = (\cdot,\cdot,\cdot,\cdot,2,\cdot)$. After the fifth step it holds that $(ca,cb,oa,ob,l,r) = (2,\cdot,0,2,\cdot,\cdot)$ since the algorithm has moved on to the next string in $A$ and adjusted the offset for the current string in $B$. The full trace is shown in Table \ref{tab:ltrace1} (`$\cdot$' is represented by an empty cell). 

The set of inputs for which traces are produced should be chosen in such a way that most of the behavior of the algorithm is exhibited (high code coverage).  
In this case one should choose positive and negative inputs.  
An input is negative because either the length of the concatenated strings differ or they have identical length but are different. We produce additional traces for: $([\text{a, a, a}], [\text{aaa}])$, $([\text{ba, a}], [\text{b, ab}])$ and  $([\text{a}], [\text{aa}])$. Their  traces can be found in Table \ref{tab:ltrace2}, \ref{tab:ltrace3}, \ref{tab:ltrace4} respectively; they have an additional column `Name' which should be ignored at this point. Note, for the input $([\text{a}], [\text{aa}])$ the algorithm recognizes that the length of the concatenation of the strings in $A$ differs from the one in $B$ and therefore immediately sets $r$ to `$\bot$'.

\subparagraph*{Step 3: Generalize the literals in the traces.}
In the following we refer to the trace in Table~\ref{tab:ltrace1}. 
Let us write $x(i)$ to denote the value of variable $x$ after the $i$-th execution step of the algorithm and let $(x,i)$ denote the corresponding cell in the table. The goal of this step is to express the values in the $i$-th row of the table in terms of the values in row $i-1$.
The literals in the trace are instantiations of abstract quantities. For example, for $i=2$ the value $l(i)=2$ represents the abstract quantity $|B[1]|$.  This can be further generalized to $|B[cb]|$ since $cb(i-1) = 1$. We add this information to the table by writing `$2 := |B[cb]|$' in cell $(l,2)$. %We call expressions after `$:=$' atomic operations.  

To see how the values in the third row emerge consider the first step in Figure~\ref{fig:csma}. After two characters have been read the algorithm is still at the first string in $A$ and has moved on to the next string in $B$, therefore we write `$2 := cb + 1$' in cell $(cb,3)$. Since two characters have been read the offset $oa$ has to be increased by this value, we write `$2 := oa + l$' in cell $(oa,3)$. 
The offset for the current string in $B$ needs to be reset since nothing of it has been read yet, we write `$0 := 0$' in cell $(ob,3)$.
The value $l(4)$ represents the length of the remainder of the current string in $A$, i.e.~$A[1]$ starting from offset 2; write `$2 := |A[ca]| - oa$' in cell $(l,4)$. In row 5 the situation is symmetric to row 3; instead of moving to the next string in $B$ one moves to the next string in $A$. In analogy to $(l,4)$ write `$2 := |B[cb]| - ob$' in cell $(l,6)$. At this point one might notice that the contents of cell $l(2)$ can be further generalized from `$2 := |B[cb]|$' to `$2 := |B[cb]| - ob$', which is the same expression as in cell $(l,6)$. After completing this step we arrive at Table~\ref{tab:btrace1}. By checking whether each expression indeed evaluates to the literal on the LHS we can immediately spot a wrong generalization.
This step is conducted for every trace.

\subparagraph*{Step 4: Identify and gather operations.}
Remove the literals and `:=' from the generalized traces. We call a row of expressions an operation. Observe that certain rows have the same operation. For example, this is the case for the rows 2, 6 and 3, 7 in Table~\ref{tab:btrace1}. 
Collect the different operations from all traces in a table and name them. We call the operation $(\cdot,\cdot,\cdot,\cdot, |A[ca]| - oa ,\cdot)$ ALEN because it determines the length of the remaining current string in $A$. Analogously, we name $(\cdot,\cdot,\cdot,\cdot, |B[cb]| - ob ,\cdot)$ BLEN. The operation $(ca + 1, \cdot, 0, ob + l,\cdot,\cdot)$ is called ANBS (A Next, B Stay).  The operations in row 3 and 11 are named ASBN and ANBN, respectively. See Table~\ref{tab:stateops} for the list of operations. For the other three traces the additional column `Name' refers to the operation. This is a more compact representation of a generalized trace. In general, the previous step of generalizing traces can be interleaved with this one. For example, one can generalize the first trace, identify its operations and then generalize the other traces by simply attaching the operation's name to each row. One should, however, not forget to check whether the operations match the literals. 

During this step one might notice that certain operations are missing because they do not appear in the considered traces. In that case one should find inputs which lead to traces with the missing operations and go back to step 2. For example, in this case one might wonder whether an operation ASBS is missing. After trying to come up with an input whose trace contains this operation one should eventually realize that it can never occur. We recommend to prove the existence of every operation by an actual input rather than justifying it by analogy since the latter might lead to phantom operations, i.e.~those which are never executed by the program. 

\subparagraph*{Step 5: Synthesize the control flow graph.}
Instead of writing lines of code we build the desired program by constructing its control flow graph $G$ from the traces. At each node of $G$ an operation is executed and an edge from $u$ to  $v$ indicates that the program \emph{might} proceed with $v$ after $u$. A generalized trace corresponds to a path through $G$.
The information that we have collected so far can be used to partially synthesize $G$. More specifically, let the vertex set of $G$ be the set of operations and a special vertex `Start' and add an edge for every two consecutive operations in the traces. For instance, in Table~\ref{tab:btrace1} the operation of row 1 is INIT and  the operation of row 2 is BLEN, thus we add an edge from INIT to BLEN in $G$. Additionally, add an edge from `Start' to every operation that occurs in row 1 of a trace. At the vertex `Start' no operation is executed; it can be thought of as the program state where the input is passed to the program.
This leads to the graph shown in Figure~\ref{fig:pcfg}. The states YES and NO are marked to indicate that they are terminating states, i.e~whenever the program reaches one of these states it has properly terminated. 
The edges (INIT,ALEN), (BLEN,NO), (Start,NO) are from Tables~\ref{tab:ltrace2}, \ref{tab:ltrace3}, \ref{tab:ltrace4}, respectively.

\begin{figure}
	\begin{center}
		\begin{tikzpicture}[shorten >=1pt,auto,node distance=1.2cm,
main node/.style={draw}]

%\path[use as bounding box] (-1.5,-4.8) rectangle (14.8,4.1);

\newcommand*{\xlen}{4cm}%
\newcommand*{\ylen}{1cm}%

\newcommand*{\xtlen}{1cm}%
\newcommand*{\ytlen}{0.4cm}%

\node[] (GHOST) at (-3.3,0) {};

\node[main node] (START) at (-1.9-0.3,0) {Start};
\node[main node] (INIT) at (0,0) {INIT};

\node[main node, above right = \ylen and \xlen of INIT] (ALEN) {ALEN};
\node[main node, below right = \ylen and \xlen of INIT] (BLEN) {BLEN};

\node[main node,fill={black!10}, above left = \ytlen and \xtlen-0.1cm of ALEN] (YES) {YES};
\node[main node,fill={black!10}, below left = \ytlen+0.2cm and \xtlen of ALEN] (NO) {NO};

\node[main node, right = \xlen + 4.2cm of INIT] (ANBN) {ANBN};
\node[main node, above = 2cm of ANBN] (ANBS) {ANBS};
\node[main node, below = 1.5cm of ANBN] (ASBN) {ASBN};

%\node[main node, right = 3cm of ANBN] (NOOP) {NOOP};

\path[->]
(GHOST) edge (START)
(START) edge (INIT)
%(INIT) edge (NO)

(START) edge[bend left=20] (NO)

(INIT) edge[bend left=10] (ALEN)
(INIT) edge[bend right=10] (BLEN)

(ALEN) edge (YES)
%(ALEN) edge (NO)
(ALEN) edge[bend left=10] (ANBS)
(ALEN) edge[bend right=10] (ANBN)

(BLEN) edge[bend right=10] (ASBN)
(BLEN) edge (NO)

(ANBS.90) edge[bend right=40] (ALEN)
(ANBN) edge[bend right=10] (ALEN)
(ASBN) edge[bend right=10] (ALEN.-60)

(ANBS.240) edge[bend right=20] (BLEN)
%(ANBN) edge[bend left=10] (BLEN)
%(ASBN) edge[bend right=10] (BLEN.-7)
;

\end{tikzpicture}
	\end{center}
	\caption{Partial control flow graph}
	\label{fig:pcfg}	
\end{figure}
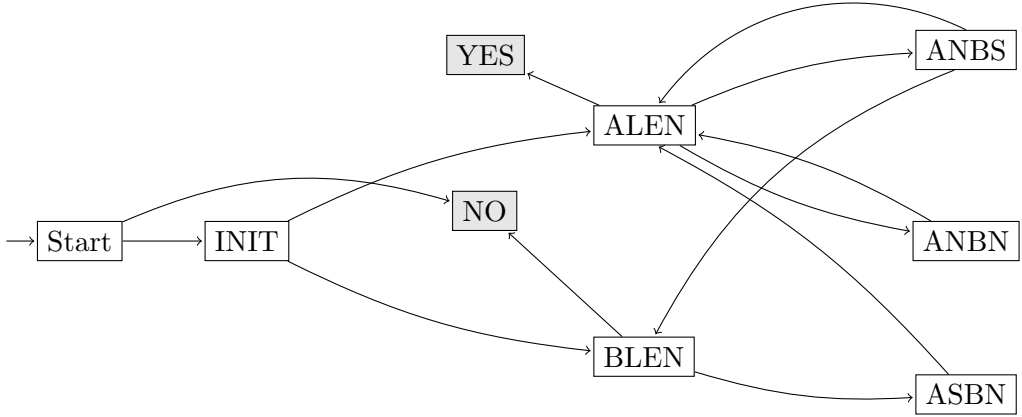

\subparagraph*{Step 6: Complete the control flow graph.} 
The objective of this step is to identify missing edges in the control flow graph and find inputs whose traces witness these missing edges. Let us start by considering the node Start. Should there be an edge from Start to YES? If we allow the input which consists of two empty arrays then there should be such an edge because we could immediately check whether the input arrays are empty, and if so, conclude that they match. However, we shall disallow empty arrays as input and therefore there is no such edge. If the program does not continue with NO after Start this implies that the concatenation of the strings in both arrays have the same length. Consequently, at least one substring check has to be performed and before that we have to initialize the variables. Therefore Start only goes to NO and INIT. 

There is no edge from INIT to YES or NO because at least one substring check has to be performed after this program state. There is no edge from INIT to ANBS, ANBN or ASBN because these operations are performed after a substring check which requires determining $l$ first. Since YES and NO are terminating states they have no outgoing edges.

Next, let us consider ALEN and BLEN. There is an edge missing from ALEN to NO, which is witnessed by the input $([\text{a}],[\text{b}])$ (we make the arbitrary choice that ALEN is executed whenever the remaining current strings in $A$ and $B$ are of equal length). There is no edge from ALEN to ASBN because after ALEN the rest of the current string in $A$ is read which implies that one has to move to the next string in $A$. There is no edge from BLEN to YES because the substring call after BLEN cannot be the last one. This holds because if BLEN is visited than the length of the remaining current string in $B$ is strictly shorter than the length of the one in $A$. This means there is still a part of the current string in $A$ to be checked. There is no edge from BLEN to ANBS for the same reason that there is no edge from ALEN to ASBN. Additionally, there is also no edge from BLEN to ANBN due to the arbitrary choice that we made above.  

Lastly, we consider ANBS, ANBN and ASBN. The program should only terminate right after ALEN or BLEN and INIT is never visited a second time. Thus there can only be edges from these three nodes to ALEN or BLEN. All these six edges should be present which means the edges (ANBN,BLEN) and (ASBN,BLEN) are missing. These two edges are witnessed by the input $([\text{a, aaa}],[\text{a, a, a, a}])$.

\subparagraph*{Step 7: Add edge predicates to the control flow graph.}
In this step the edges of the control flow graph are labeled with boolean expressions which specify what path the program takes through the control flow graph during execution. 

As previously stated the program goes from Start to NO if the concatenated strings of $A$ and  $B$ have different lengths. Let EQLEN denote the predicate which is true iff the concatenated strings are of equal length, i.e.~$\sum_{i=1}^{|A|} |A[i]| = \sum_{i=1}^{|B|} |B[i]|$. This means we label the edge (Start, NO) with `$\neg$EQLEN'. After labeling an edge with a predicate we have to verify that it is consistent with the traces. In this case it has to hold that one goes from Start to NO in a trace iff the predicate `$\neg$EQLEN' is satisfied, which is the case. Remember that we omitted columns for the input variables $A,B$ in the traces to save space since they are not changed by the program. However, in general all variables should appear in the traces. The program state Start can be thought of as a 0-th row in the traces where the input variables are set. The predicate EQLEN is interpreted w.r.t.~this 0-th row. Also, we label (Start, INIT) with `EQLEN'.

The program goes from INIT to ALEN if the length of the remaining current string in $A$ is not larger than the one in $B$, i.e.~$|A[ca]| - oa \leq |B[cb]| - ob$. Let ALEQ denote this predicate. We label (INIT,ALEN) with `ALEQ' and (INIT,BLEN) with `$\neg$ALEQ'. Again, we have to verify that the newly added edge predicates are consistent with the traces. In Table~\ref{tab:btrace1} one goes from INIT (row 1) to BLEN (row 2). It holds that ALEQ is false ($|A[ca(1)]| - oa(1) = 4  \not\leq 2 = |B[cb(1)]| - ob(1) $). In Table~\ref{tab:ltrace2} one goes from INIT to ALEN. In that case ALEQ is true because $|A[ca(1)]| - oa(1) = 1 \leq 3 = |B[cb(1)]| - ob(1)$. Similarly, in Table~\ref{tab:ltrace3} one goes from INIT to BLEN and ALEQ is false. Therefore the newly added predicates are consistent with all traces. 

For $X \in \{ \text{ANBS, ANBN, ASBN}\}$ label $(X,\text{ALEN})$ with `ALEQ' and $(X,\text{BLEN})$ with `$\neg$ALEQ'. It holds that INIT, ANBS, ANBN and ASBN go to ALEN if ALEQ holds and to BLEN otherwise. Since these four program states behave identically in the sense that they have the same set of outgoing edges with the same edge predicates, the control flow can be simplified as follows.
We restructure the control flow graph by adding a new node called NOOP (no operation) which has incoming edges from INIT, ANBS, ANBN and ASBN and outgoing edges to ALEN and BLEN. The program does not modify any variable at this program state and thus we can assume that it does not add a row to the traces.  
We label ($X$,NOOP) with `$\top$' (constant true) for $X \in  \{ \text{INIT, ANBS, ANBN, ASBN}\}$, (NOOP,ALEN) with `ALEQ' and (NOOP,BLEN) with `$\neg$ALEQ'. Consistency is maintained after these modifications.

It remains to specify the edge predicates for the outgoing edges of ALEN and BLEN. Let us start with BLEN which either goes to NO or ASBN. The only reason the program should go to NO if the concatenated strings have equal length is if one of the substring checks fails. Let SS be true if the substring check succeeds, i.e.~$\mathrm{substr}(A[ca],oa,l) = \mathrm{substr}(B[cb],ob,l)$. Label (BLEN,NO) with `$\neg$SS' and (BLEN,ASBN) with `SS'. Let us consider ALEN. For the same reason as before we can label (ALEN,NO) with `$\neg$SS'. The program should go to YES if every substring check succeeded. In particular, the program should go from ALEN to YES if the next substring check succeeds and it is the final one. Let FC (final call) be the predicate which denotes that the next substring check is the last one. Label (ALEN,YES) with `SS $\wedge$ FC'. Whether one goes from ALEN to ANBS or ANBN depends on whether the end of the current string in $B$ has been reached. Label (ALEN,ANBS) with `$\neg$EOB $\wedge$ SS $\wedge$ $\neg$FC' and  (ALEN,ANBN) with `EOB $\wedge$ SS $\wedge$ $\neg$FC' where EOB is the predicate which expresses that the end of the current string in $B$ has been reached. Notice, that the part `$\wedge \neg$FC' for the edge predicate of (ALEN,ANBS) can be omitted since $\neg$EOB implies that the next substring call cannot be the final one ($\neg$FC). The definition of all predicates is given in Table~\ref{tab:predicates}. 

The final control flow graph with edge predicates is shown in Figure~\ref{fig:ccfg}. Together with the definitions of the operations and predicates this is a complete program. 

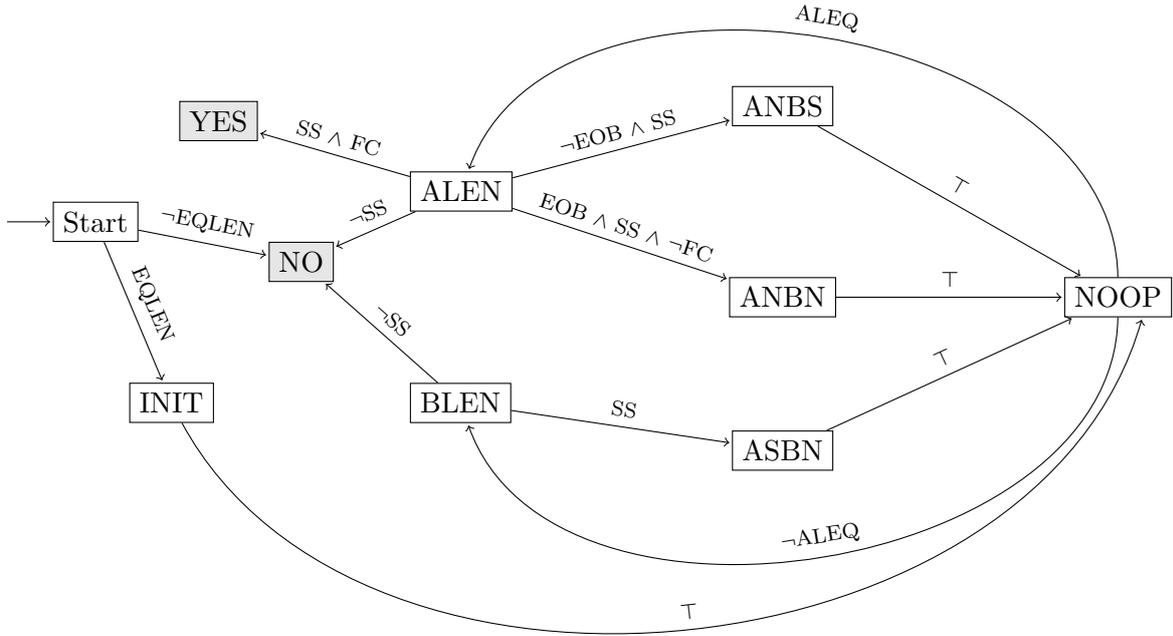
\begin{figure}
	\begin{center}
		\begin{tikzpicture}[shorten >=1pt,auto,node distance=1.2cm,
main node/.style={draw}]

\path[use as bounding box] (-1.5,-4.8) rectangle (14.8,4.1);

\newcommand*{\xlen}{4cm}%
\newcommand*{\ylen}{1cm}%

\newcommand*{\xtlen}{1cm}%
\newcommand*{\ytlen}{0.4cm}%

\node[] (GHOST) at (-1.3,1) {};
\node[main node] (START) at (0,1) {Start};

%INIT Phantom
\node (INIT) at (0,0) {};

\node[main node] (INITR) at (1,-1.4) {INIT};

\node[main node, above right = \ylen and \xlen of INIT] (ALEN) {ALEN};
\node[main node, below right = \ylen and \xlen of INIT] (BLEN) {BLEN};

\node[main node,fill={black!10}, above left = \ytlen and \xtlen+1cm of ALEN] (YES) {YES};
\node[main node,fill={black!10}, below left = \ytlen and \xtlen of ALEN] (NO) {NO};

\node[main node, right = \xlen + 4.2cm of INIT] (ANBN) {ANBN};
\node[main node, above = 2cm of ANBN] (ANBS) {ANBS};
\node[main node, below = 1.5cm of ANBN] (ASBN) {ASBN};

\node[main node, right = 3cm of ANBN] (NOOP) {NOOP};

\path[->]
(GHOST) edge (START)
(START) edge node[midway, above, sloped]{\scriptsize EQLEN} (INITR)
(START) edge node[midway, above, sloped]{\scriptsize $\neg$EQLEN} (NO)

(ALEN) edge node[midway, above, sloped]{\scriptsize SS $\wedge$ FC} (YES)
(ALEN) edge node[midway, above, sloped]{\scriptsize $\neg$SS}  (NO)
(ALEN) edge node[midway, above, sloped]{\scriptsize $\neg$EOB $\wedge$ SS}  (ANBS)
(ALEN) edge node[midway, above, sloped]{\scriptsize EOB $\wedge$ SS $\wedge$ $\neg$FC}  (ANBN)

(BLEN) edge node[midway, above, sloped]{\scriptsize SS} (ASBN)
(BLEN) edge node[midway, above, sloped]{\scriptsize $\neg$SS} (NO)

(ANBS) edge node[midway, above, sloped]{\scriptsize $\top$} (NOOP)
(ASBN) edge node[midway, above, sloped]{\scriptsize $\top$} (NOOP)
(ANBN) edge node[midway, above, sloped]{\scriptsize $\top$} (NOOP)

(NOOP) edge[bend right=80] node[midway, above, sloped]{\scriptsize ALEQ} (ALEN)
(NOOP) edge[bend left=80] node[midway, above, sloped]{\scriptsize $\neg$ALEQ} (BLEN)

(INITR) edge[bend right=68] node[midway, above, sloped]{\scriptsize $\top$} (NOOP.-40)
;
%(INIT) edge node[midway, above, sloped]{\scriptsize $\top$} (LEN)
%(INIT) edge node[midway, above, sloped]{\scriptsize $\neg$EQLEN} (NO)
%(LEN) edge node[midway, above, sloped]{\scriptsize $\neg$SS} (NO)
%(LEN) edge node[midway, above, sloped]{\scriptsize SS $\wedge$ FC} (YES)
%(LEN) edge node[midway, above, sloped]{\scriptsize $\neg$EOA $\wedge$ EOB $\wedge$ SS} (ASBN)
%(LEN) edge node[midway, above, sloped]{\scriptsize EOA $\wedge$ $\neg$EOB $\wedge$ SS} (ANBS)
%(LEN) edge node[midway, above, sloped]{\scriptsize EOA $\wedge$ EOB  $\wedge$ SS $\wedge$ $\neg$FC} (ANBN)
%
%(ASBN) edge node[midway, above, sloped]{\scriptsize $\top$} (NOOP)
%(ANBS) edge node[midway, above, sloped]{\scriptsize $\top$} (NOOP)
%(ANBN) edge node[midway, above, sloped]{\scriptsize $\top$} (NOOP)
%
%(NOOP.90) edge[bend right=75] node[midway, above, sloped]{\scriptsize $\top$} (LEN)
%;

%\path[-,shorten >=0pt,shorten >=0pt]
%(ASBN) edge  (NOOP)
%(ANBS) edge (NOOP)
%(ANBN) edge (NOOP)
%;

\end{tikzpicture}
	\end{center}
	\caption{Complete control flow graph with edge predicates}
	\label{fig:ccfg}	
\end{figure}

\subparagraph*{Further remarks.} 
Notice that the variable $l$ is functionally dependent on the remaining variables. More specifically, $l = \min ( |A[ca]|-oa, |B[cb]|-ob )$ holds whenever it is used. This can be used to merge the operations ALEN and BLEN since the previous expression generalizes both operations. Alternatively, one can also completely omit the variable $l$ and the operations where it is computed. In that case its computation is delegated to the edge predicates and the number of program states can be reduced. Demanding that there is no functional dependency among the variables might be a useful convention as it leads to more canonical programs, which facilitates comprehending someone else's code (reminiscent of normal forms in database theory). %It also makes traces less redundant.
%Demanding that might be a useful convention as it leads to more canonical programs and thus simplifies  akin to normalization in database theory.

%For a given program state, the disjunction of the predicates of its ingoing edges can be seen as precondition. 
%Similarly, the conjunction of the predicates of its outgoing edges are a postcondition. 
%Observe that when going from ALEN to ANBS it is logically implied that FC does not hold because FC requires EOB to hold. Therefore adding `$\wedge$ $\neg$FC' to this predicate does not affect the program semantics. However, it increases readability in the sense that it is immediately clear that the program continues with either ANBS or ANBN after ALEN solely depending on EOB iff SS $\wedge$ $\neg$FC holds without having to know the underlying logical dependencies of the atomic predicates. The additional conjunct $\neg$FC can be seen as assertion.

\subsection{Lindell's Tree Isomorphism Algorithm}

In \cite{lindell} it is shown how to decide whether two rooted trees are isomorphic 
using $\mathcal{O}(\log n)$ space where $n$ is the number of nodes of both trees.
The algorithm has read-only access to the two input trees $S, T$. 
It can query the root nodes and, furthermore, the parent, the lexicographically first child and the lex.~next sibling of a given node. 
It uses four types of variables. The first type stores a pointer to a node. 
The second type can hold a constant number of different values, i.e.~an enum. 
The third type holds a non-negative integer. 
The fourth type is a stack whose elements are one of the previous three types. Variables of the first and second type only require $\mathcal{O}(\log n)$ space by definition. For variables of the other two types we argue later on why they do not exceed the logarithmic space bound. Since the algorithm always pushes and pops from all stacks at the same time, we see them as a single stack whose elements are sequences of values. 

In addition to the description of the algorithm in the following subsection, we recommend reading the first two parts of the implementation (step 1 to 4) and the associated traces. This makes it easier to comprehend how the algorithm works. 

In the following we omit the qualifier rooted. Let $T$ be a tree and let $v$ be a node of $T$. We write $|T|$ to denote the number of nodes of $T$, $T_v$ to denote the subtree of $T$ induced by $v$ ($T_v$ has root node $v$) and $\#v$ to denote the number of children of $v$. 

\subsubsection{Algorithm}
The rough idea behind the algorithm is as follows. A binary relation $\prec$ on trees is defined such that two trees $S,T$ are isomorphic ($S \cong T$) iff neither $S \prec T$ nor $T \prec S$ holds. The relation $\prec$ is a strict order on the isomorphism classes of trees and it is defined recursively in terms of the subtrees induced by the children of the root node.
Assume two trees $S,T$ are given as input. Let $\textsc{cmp}(s,t)$ be a recursive procedure where $s,t$ are nodes of $S,T$ (respectively) and $\textsc{cmp}(s,t)$ outputs whether $S_s \prec T_t$ or $T_t \prec S_s$ or $S_s \cong T_t$. Thus, calling $\textsc{cmp}(s,t)$ with $s,t$ being the root nodes of $S,T$ yields the desired result. However, a naive implementation of $\textsc{cmp}(s,t)$ would require too much space since the recursion depth $d$ is determined by the depth of the trees $S$ and $T$ and thus can be linear w.r.t.~the input size $n$ (each recursive call requires to store the caller environment on a stack and thus the stack contains at least $d$ bits at some point). The algorithm computes $\textsc{cmp}(s,t)$ in such a way that under certain circumstances the caller environment can be restored without any information at all and otherwise only very little information needs to be pushed on the call stack.  

Before we define $\prec$, let us first show how a strict order on a given domain can be extended to a strict order on multisets over the same domain. This simplifies the definition of $\prec$. The multisets that we want to order have the same cardinality. For example, consider $\N$ as domain along with its natural order $<$. Suppose you want to order the multisets $A=\multiset{ 6,1,3,1,1,2,3}$ and $B=\multiset{  8,1,4,1,1,2,2 }$. You can do this by converting them into sorted tuples (in ascending order) and then comparing the tuples lexicographically. The sorted tuple of $A$ is $\vec{a} = (1,1,1,2,3,3,6)$ and for $B$ it is $\vec{b} = (1,1,1,2,2,4,8)$. Then $B < A$ because the fifth element of $\vec{b}$ is smaller than that of $\vec{a}$.   

\begin{definition}
    Let $S,T$ be two trees with root nodes $s,t$. Let $S_1,\dots,S_{\#s}$ and $T_1,\dots,T_{\#t}$ be the subtrees induced by the children of $s$ and $t$.
    The binary relations $\prec_1$ and $\prec$ are defined as follows.
    The relation $S \prec_1 T$ holds if: 
    \begin{enumerate}
        \item $|S| < |T|$, or
        \item $|S| = |T|$ and $\#s < \#t$, or
        \item $|S| = |T|$, $\#s = \#t = r$ and $\multiset{ |S_1|,\dots,|S_r| }  < \multiset{ |T_1|,\dots,|T_r| }$
    \end{enumerate}
	The relation $S \prec T$ holds if:
    \begin{enumerate}
	\item $S \prec_1 T$, or
	\item $S \not\prec_1 T$, $T \not\prec_1 S$ and $\multiset{S_1,\dots,S_r} \prec \multiset{T_1,\dots,T_r}$
%	$|S| = |T|$, $\#s = \#t = r$, $\multiset{ |S_1|,\dots,|S_r| } = \multiset{ |T_1|,\dots,|T_r| }$ and $\multiset{S_1,\dots,S_r} \prec \multiset{T_1,\dots,T_r}$ 
	%\item $|S| = |T|$, $\#s = \#t = r$, $\multiset{ |S_1|,\dots,|S_r| } = \multiset{ |T_1|,\dots,|T_r| }$ and $\multiset{S_1,\dots,S_r} \prec \multiset{T_1,\dots,T_r}$ 
\end{enumerate}	
\label{def:lto}
\end{definition}

The definition of $\prec$ given here slightly differs from the original one given by Lindell (if the third condition is removed from $\prec_1$ then they are equivalent) to simplify the algorithm. 

Now, we can explain how the algorithm computes $\textsc{cmp}(s,t)$ 
when given two trees $S,T$ as input. All variables used by the algorithm are global, i.e.~they can be accessed from anywhere during execution. 
The procedure $\textsc{cmp}$ is implemented without any parameters using two global variables $s,t$ instead.
Calling $\textsc{cmp}(u,v)$ corresponds to setting $s$ to $u$ and $t$ to $v$ and then calling $\textsc{cmp}$; since no local variables are used this can be realized with gotos and a pointer to remember where the call came from. 

\subparagraph{Recursion termination (non-recursive checks).} 
Suppose $s,t$ are set and $\textsc{cmp}$ is called for the last time, i.e.~the tree of recursive calls has reached a leaf and thus the recursion starts to return. We have to determine whether $S_s \prec T_t$ (`$\prec$') or $T_t \prec S_s$ (`$\succ$') or $S_s \cong T_t$ (`$\cong$'), store the result in a variable $res$ and return. Let $\textsc{cmp}_1$ be a procedure which outputs whether $S_s \prec_1 T_t$ (`$\prec_1$') or $T_t \prec_1 S_s$ (`$\succ_1$')  or neither (`$\cong_1$') . Execute $\textsc{cmp}_1$. Suppose $\textsc{cmp}_1$ returns `$\prec_1$'. This implies $S_s \prec T_t$. Therefore we can set $res$ to `$\prec$' and return to the caller or terminate if there exists no caller. If $s,t$ are the root nodes of $S,T$ then
the computation of the recursive procedure $\textsc{cmp}$ is finished and $S$ and $T$ are non-isomorphic. Assume $s,t$ are not the root nodes of $S,T$. Let $s_0,t_0$ be the parent nodes of $s,t$. The call to the current $\textsc{cmp}$ came from the parent nodes since only parents are interested in comparing their children. Therefore we set $s$ to $s_0$ and $t$ to $t_0$ and return. The case `$\succ_1$' can be handled analogously. Assume $\textsc{cmp}_1$ returns `$\cong_1$'. From the definition of $\prec$ it follows that the children of $s,t$ have to be compared w.r.t.~$\prec$. This implies that further recursive calls have to be made. Since we are at a leaf in the tree of recursive calls this cannot be the case.

In order to compute $\textsc{cmp}_1$ we have to compute $\prec_1$. This reduces to solving the following three problems in logspace. Let $T$ be a tree and let $v$ be a node of $T$. Problem 1 is to compute the size of the induced subtree $|T_v|$. Problem 2 is to count the number of children of $v$. Problem 3 is to compare two multisets $A,B$ over $\N$ with the same cardinality. The latter two problems are straightforward to solve. The first problem can be solved via a depth-first traversal through the tree starting from $v$ and counting the nodes during this process (see \cite{lindell}). For the sake of brevity we assume that programs which solve these three problems and a program which computes $\textsc{cmp}_1$ have been already implemented and can be used by our implementation.

\subparagraph{Equicardinality blocks.} 
   Let $T$ be a tree and let $v$ be a node of $T$. We call a non-empty, inclusion-maximal subset $B = \{v_1,\dots,v_l\}$ of children of $v$ a block if the size of the induced subtree of every element in $B$ has the same size $k$, i.e.~$|T_{v_1}| = \dots = |T_{v_l}| = k$. We say $B$ has block size $k$ and cardinality $l$. As $k$ uniquely determines $B$ we also write $B(k)$ to refer to $B$. For example, let $T$ be a tree with a node $v$ which has children $v_1,\dots,v_6$ with
    \[ (|T_1|,\dots,|T_6|) = (2,5,2,4,4,2) \]
    where $T_1,\dots,T_6$ denote the subtrees induced by the children of $v$. Then $v$ has blocks $B(2),B(4),B(5)$ with respective cardinalities $3,2,1$. We call $\bigl(\begin{smallmatrix}
    2 & 4 & 5 \\
    3 & 2 & 1 
    \end{smallmatrix}\bigr)$ the block signature of $s$; this is just a different representation of the multiset $\multiset{ |T_1|,\dots,|T_6| }$. Notice that the third condition in the definition of $\prec_1$ 
    requires us to compare the block signatures of the root nodes. Therefore if we reach the recursive part of the algorithm, we know that the two nodes $s,t$ under consideration have the same block signature. Suppose that both have the block signature $\bigl(\begin{smallmatrix}
    k_1 & \dots & k_j \\
    l_1 & \dots & l_j 
    \end{smallmatrix}\bigr)$ for some $j \geq 1$. In the recursive part of the algorithm we iterate from $i=1$ to $j$ and compare the children in block $B(k_i)$ of $s$ with the children in block $B(k_i)$ of $t$. We add a subscript to denote from what tree a block is. 
    We say two blocks $B_S,B_T$ match if there exists a bijection $\pi \colon B_S \rightarrow B_T$ such that $S_v \cong T_{\pi(v)}$ holds for all $v \in B_S$. The subtrees induced by $s$ and $t$ are isomorphic iff all of their blocks match.

\subparagraph{Memoryless recursion ($l=1$).} 
Consider the two isomorphic trees $S,T$ shown in Figure~\ref{fig:ex_trees1}. 
When $\textsc{cmp}(1',1'')$ is called the algorithm first computes $\textsc{cmp}_1(1',1'')$ to check whether $S_{1'} \prec_1 T_{1''}$ or $T_{1''} \prec_1 S_{1'}$  holds. Neither holds because both trees have the same size and the root nodes have the same block signature $\bigl(\begin{smallmatrix}
2 & 3 & 4 \\
1 & 1 & 1 
\end{smallmatrix}\bigr)$. Therefore the algorithm proceeds with the recursive part where each pair of blocks is compared. First, it checks whether the blocks $B_S(2), B_T(2)$ match. They match iff $S_{2'} \cong T_{4''}$. Therefore a recursive call $\textsc{cmp}(2',4'')$ is made. In our example this call returns `$\cong$' and the algorithm can proceed with the comparison of the blocks $B_S(3), B_T(3)$. The environment before the recursive call consists of $(s,t,k) = (1',1'',2)$ where $k$ is the current block size under consideration. Observe that this information can be restored after returning from the call $\textsc{cmp}(2',4'')$ without remembering anything. It is clear that the call was issued from the parent nodes of $2', 4''$ and the block size must have been $|S_{2'}|$. After comparing the blocks with size 3 and the ones with size 4, the algorithm concludes that the trees are isomorphic since no mismatch was found. 

Suppose that the blocks $B_S(2), B_T(2)$ do not match, e.g.~$S_{2'} \prec T_{4''}$. Then after returning from the call $\textsc{cmp}(2',4'')$ with the result `$\prec$', the algorithm  immediately concludes that $S_{1'} \prec T_{1''}$ holds without comparing the remaining blocks. This is correct  due to the lexicographic nature of $\prec$.

\begin{figure}
    \begin{center}
        \begin{tikzpicture}[shorten >=1pt,auto,node distance=1.2cm,
main node/.style={draw}]

\node at (-2.4,1.6) {$S$};
\node at (2.65,1.6) {$T$};

\node at (-2,0) {
\Tree [.$1'$ [.$2'$ [.$5'$ ] ] [.$3'$ [.$6'$ ] [.$7'$ ] ] [.$4'$ [.$8'$ ] [.$9'$ ] [.$10'$ ] ] ]
};

\node at (2.6,0) {
\Tree [.$1''$ [.$2''$ [.$5''$ ] [.$6''$ ] ] [.$3''$ [.$7''$ ] [.$8''$ ] [.$9''$ ] ] [.$4''$ [.$10''$ ] ] ]
};

\end{tikzpicture}
    \end{center}
    \caption{Example trees ($l=1$)}
    \label{fig:ex_trees1}
\end{figure}
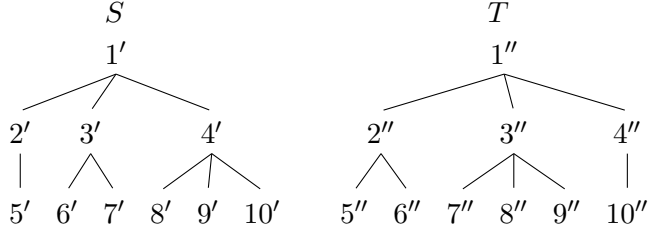

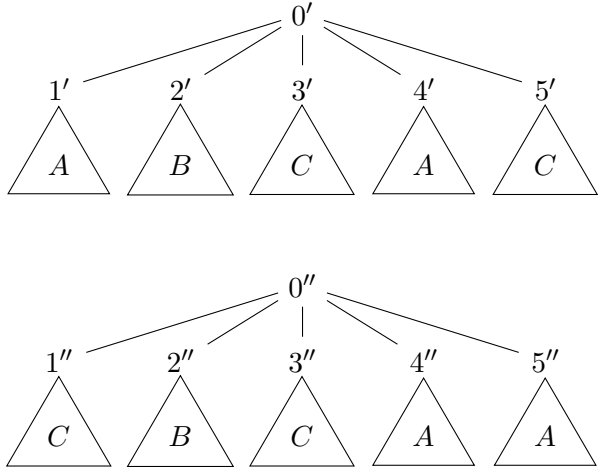
\begin{figure}
    \begin{center}
        \begin{tikzpicture}[shorten >=1pt,auto,node distance=1.2cm,
main node/.style={draw},triangle/.style = {draw, regular polygon, regular polygon sides=3 }]

\newcommand*{\xa}{4cm}%
\newcommand*{\ya}{1cm}%

\newcommand*{\xb}{4cm}%
\newcommand*{\yb}{-2.6cm}%

\newcommand*{\xlen}{1.6cm}%
\newcommand*{\ylen}{1cm}%

\node (srn) at ({\xa},{\ya}) {$0'$};

\node (s3) at ({\xa},{\ya-1*\ylen}) {$3'$};
\node (s2) at ({\xa-\xlen},{\ya-1*\ylen}) {$2'$};
\node (s1) at ({\xa-2*\xlen},{\ya-1*\ylen}) {$1'$};
\node (s4) at ({\xa+\xlen},{\ya-1*\ylen}) {$4'$};
\node (s5) at ({\xa+2*\xlen},{\ya-1*\ylen}) {$5'$};

\node[triangle] (s1t) at ({\xa-2*\xlen-0.2},{\ya-2*\ylen}) {$A$};
\node[triangle] (s1t) at ({\xa-1*\xlen-0.2},{\ya-2*\ylen}) {$B$};
\node[triangle] (s1t) at ({\xa-0*\xlen-0.2},{\ya-2*\ylen}) {$C$};
\node[triangle] (s1t) at ({\xa+1*\xlen-0.2},{\ya-2*\ylen}) {$A$};
\node[triangle] (s1t) at ({\xa+2*\xlen-0.2},{\ya-2*\ylen}) {$C$};

\path
(srn) edge (s3)
(srn) edge (s2)
(srn) edge (s1)
(srn) edge (s4)
(srn) edge (s5)
;

\node (srn) at ({\xb},{\yb}) {$0''$};

\node (s3) at ({\xb},{\yb-1*\ylen}) {$3''$};
\node (s2) at ({\xb-\xlen},{\yb-1*\ylen}) {$2''$};
\node (s1) at ({\xb-2*\xlen},{\yb-1*\ylen}) {$1''$};
\node (s4) at ({\xb+\xlen},{\yb-1*\ylen}) {$4''$};
\node (s5) at ({\xb+2*\xlen},{\yb-1*\ylen}) {$5''$};

\node[triangle] (s1t) at ({\xb-2*\xlen-0.2},{\yb-2*\ylen}) {$C$};
\node[triangle] (s1t) at ({\xb-1*\xlen-0.2},{\yb-2*\ylen}) {$B$};
\node[triangle] (s1t) at ({\xb-0*\xlen-0.2},{\yb-2*\ylen}) {$C$};
\node[triangle] (s1t) at ({\xb+1*\xlen-0.2},{\yb-2*\ylen}) {$A$};
\node[triangle] (s1t) at ({\xb+2*\xlen-0.2},{\yb-2*\ylen}) {$A$};

\path
(srn) edge (s3)
(srn) edge (s2)
(srn) edge (s1)
(srn) edge (s4)
(srn) edge (s5)
;
%\node at (2.65,1.6) {$T$};

\end{tikzpicture}
    \end{center}
    \caption{Example trees ($l>1$);  $A \prec B \prec C$ and $|A|=|B|=|C|=x$}
    \label{fig:ex_trees}
\end{figure}

\subparagraph{Recursion with order profiles / cross comparison ($l>1$).} 
It remains to explain how the algorithm deals with blocks that have cardinality greater than 1. Consider the trees $S,T$ shown in Figure~\ref{fig:ex_trees}. Assume that the subtrees $A,B,C$ have the same size $x=|A|=|B|=|C|$ and $A \prec B \prec C$. Thus, both root nodes have only one block $B(x)$ with cardinality $l=5$. We call $A$ the isomorphism type of $1',4',4'',5''$, $B$ the isomorphism type of $2',2''$ and so on. The blocks $B_S(x), B_T(x)$ match iff their multisets of isomorphism types coincides, which is the case in this example. 

Let $v$ be a node in $B_S(x)$.
The order profile of $v$ is $({gt}_v,{eq}_v)$ where ${gt}_v$ (${eq}_v$) is the number of nodes $w$ in $B_T(x)$ such that $T_w \prec S_v$ ($T_w \cong S_v$).
For example, the nodes $1',4',4'',5''$ have order profile $(0,2)$, the nodes $2',2''$ have order profile $(2,1)$ and the remaining nodes have order profile $(3,2)$. The order profile of a node only depends on its isomorphism type. In order to compute the order profile of a node from $S$ we only have to compare it to nodes from $T$ and vice versa. 

%idea
To check whether the two blocks match, we proceed as follows. First, we compute the order profile of each node in $B_S(x)$ until a node $v$ is found such that ${gt}_v = 0$. The node $v=1'$ is the first such node. Then we try to find a node $w$ in $B_T(x)$ such that ${gt}_w = 0$. The node $w=4''$ is the first such node. Since ${gt}_v = 0$ it follows that there is no isomorphism type in $B_T(x)$ that is smaller than that of $v$. An analogous statement holds for $w$. Therefore we can conclude that $v$ and $w$ have the same isomorphism type. Since ${eq}_v = 2 = {eq}_w$ it follows that the isomorphism type $A$ occurs twice in both blocks. If either such a node $v$ or $w$ would not exist or ${eq}_v \neq {eq}_w$ then we could conclude that the two blocks do not match. For example, if ${eq}_w < {eq}_v$ or no node $v$ with ${gt}_v = 0$ exists then $T_t \prec S_s$.   

Since both blocks have two nodes with the smallest isomorphism type, we can now check whether both blocks also have the same number of the next smallest isomorphism type $B$. To do this, compute the order profile of each node in $B_S(x)$ until a node $v$ is found such that ${gt}_v = 2$. The node $v=2'$ is the first such node.  Next, compute the order profile of every node in $B_T(x)$ until a node $w$ is found such that ${gt}_w = 2$. The node $w=2''$ is the first such node. Since ${eq}_v = 1 = {eq}_w$ it follows that the isomorphism type $B$ occurs once in both blocks. Next, we have to look for nodes $v$ and $w$ such that ${gt}_v, {gt}_w = x$ where $x = 2+1$; this is the sum of the previous order profile. We continue with this comparison until the sum of the previous order profile equals $l$. In this example, this holds after the last isomorphism type $C$ for which the order profile is $(3,2)$ since $3+2=l$.

It remains to explain how to restore the caller environment after a recursive call.  The environment consists of $(s,t,k) = (0',0'',x)$, the node under consideration in each block ($s'$,$t'$), whether we are looking for a node in the block of $S$ or $T$ ($f$ like flag), the two (incomplete) order profiles computed so far ($sgt,seq,tgt,teq$) and an offset for the sum of the previous order profile ($h$). Suppose we return from a recursive call $\textsc{cmp}(u,v)$. Then the values of $s,t,k$  can be restored just as in the memoryless case, i.e.~$s,t$ are the parents of $u,v$ and $k=|S_u|$. The values of $s',t'$ are $u,v$. Notice that the values of the remaining six variables range between $0$ and $l$. Before the recursive call $\textsc{cmp}(u,v)$ is made, we push their values onto a stack. After returning, we restore their values by popping them from the stack (just as one would implement a normal recursion). 

Let us briefly sketch why the size of the stack does not exceed the logarithmic space bound. Suppose you have two trees $S,T$ with $n$ nodes and you currently consider the block $B(k_1)$ which has cardinality $l_1$.  Clearly, it holds that $k_1l_1 \leq n$ and therefore $l_1 \leq n/k_1$. Next, you consider the block $B(k_2)$ with cardinality $l_2$ within $B(k_1)$. It holds that $k_2l_2 \leq k_1$ and thus $l_2 \leq k_1/k_2$. After further recursive calls, you consider the block $B(k_r)$ within $B(k_{r-1})$ with cardinality $l_r$. It holds that $l_r \leq k_{r-1}/k_r$. It follows that $l_1 l_2 \dots l_r \leq n$. At this point, the contents of the stack can be represented using $c \Sigma_{i=1}^r \log l_i$ bits for some constant $c \in \N$. It holds that this is in $\mathcal{O}(\log n)$ because $ \Sigma_{i=1}^r \log l_i = \log (l_1 l_2 \dots l_r)$ and $l_1 l_2 \dots l_r \leq n$.

\subsubsection{Implementation}
%In contrast to the confidential string matching example, we describe the application of our method more succinctly by grouping certain steps.
%The description provided here can be seen as a justification of a possible implementation of the previous algorithm. 

\subparagraph*{Step 1: Determine the variables used by the algorithm.} 
The algorithm has two input variables $S,T$ which contain the trees that we want to compare. Since these variables are never modified by the algorithm, they can be omitted in the traces. Moreover, the algorithm uses the variables $s,t$ to remember the current nodes whose induced subtrees are being compared, $k$ for the current block size under consideration and $res$ to pass the result of a comparison to the caller. 
At certain points the algorithm also needs to know the cardinality $l$ of the current block $B(k)$. Since the value of $l$ is functionally dependent on $k$ and $s$, i.e.~$l$ is the number of children $v$ of $s$ such that $|S_v| = k$, we do not use a variable for $l$. If both trees have no block with cardinality $l>1$ then these are the only variables used by the algorithm. Otherwise, the variables $s',t'$ are used to store the current nodes that are being compared in the cross comparison, $sgt,seq,tgt,teq$ contain the (partial) order profiles, $h$ contains the sum of the previous order profile, $f$ is a flag that indicates whether we look for a node in $S$ or $T$ and $stk$ is a stack that stores the values of the previous six variables during recursion.

\subparagraph*{Step 2, 3 \& 4: Select inputs, produce traces and identify operations.} 
Consider the trace for $S_1, T_1$ in Table~\ref{tab:ti_trace1}. We refer to $S_1$ and $T_1$ as $S$ and $T$ in this paragraph,.
Since the isomorphic trees $S,T$ have no block with cardinality $l>1$ it suffices to consider the variables $s,t,k,res$. In row 1 we initialize $s$ and $t$ with the root nodes of $S$ and $T$ and set $k$ to 0. 
We use the operation NB (next block) to iterate over the block sizes. It computes the smallest block size of $s$ in $S$ larger than $k$. By setting $k$ to 0 and calling NB we obtain the first (smallest) block size. Since the block signatures of the root nodes match, we have to compare their blocks and thus continue with the operation NB in row 2. Since $B(2)$ has cardinality 1, it follows that we can compare the unique children in $B(2)$ by setting $s$ and $t$ appropriately. The operation GC (get child) sets $s,t$ to the lexicographically first child in $B(k)$ and sets $k$ to 0 to restart the block iteration. In row 5 we compare the nodes 3 and 14. Since the nodes 3 and 14 have no children (and therefore no block size larger than 0) it follows that all blocks match and we can set $res$ to $\cong$ in row 6. 
The operation RET is used to return from a memoryless recursion, which is the case whenever the cardinality of the block which contains $s$ is 1. It sets $s$ and $t$ to the parent nodes of the current nodes and $k$ to $|S_s|$ (the previous block size). In row 7 we return from the recursive call where $3$ and $14$ were compared. In row 8 we would continue with the next block size. But since $2$ has no block size larger than $1$ it follows that we have compared all blocks and found no mismatch. Therefore we can conclude that $S_2$ and $T_{13}$ are isomorphic in row 8. In row 9 we return and in row 10 we continue with the next block size 4. In row 29 we conclude that the subtrees $S_8$ and $T_6$ are isomorphic and in row 30 we return to their parents. Since there is no more block left to compare we conclude in row 31 that $S_1$ and $T_1$ are isomorphic. The program terminates after row 31 since $s$ and $t$ are the root nodes.

Consider the trace for $S_2,T_2$ in Table~\ref{tab:ti_trace2}. We refer to $S_2$ and $T_2$ as $S$ and $T$. In row 3 we start comparing the subtrees $S_2$ and $T_2$. Since $T_2$ has less children than $S_2$ (this means their block signatures differ) we can conclude that $S_2 \succ T_2$ in row 4. After returning to their parents in row 5, we can immediately conclude that $S_1 \succ T_1$ in row 6 because the blocks $B(3)$ did not match. 

Consider the trace for $S_3,T_3$ in Table~\ref{tab:ti_trace3}. We refer to $S_3$ and $T_3$ as $S$ and $T$. The labeled triangles represent the isomorphism types of the induced subtrees. For example, $S_2$ has isomorphism type $A$ and $T_2$ has isomorphism type $C$ and $|S_2|=|T_2|=x$. The isomorphism type of $S_4$ and $T_4$ is $B$. In row 1 $s,t,k$ are initialized as in the previous two traces. Additionally, $stk$ is initialized with the empty stack. In row 2 we get the the smallest block size of node $0$ in $S$ which is $x+1$. Since $B(x+1)$ has cardinality 5 we have to execute the cross comparison part of the algorithm. During the cross comparison the value of $h$ is increased until either $h=l$ (meaning all blocks match) or a mismatch is found. In the following, we describe the operations used for the cross comparison.

The operation SETH sets $h$ to 0. It is called whenever a cross comparison begins. It occurs in row 3 (c.c.~of the nodes 0 and 0) and 72 (c.c.~of the nodes 4 and 4). 
The operation INCH increases $h$ by $seq$. It is called whenever two nodes with prescribed order profile $h=sgt=tgt$ and $seq =teq$ have been found. For example, in row 129 the nodes 1 of $S$ and 11 of $T$ with order profile $(0,2)$ have been found and $h$ is increased by 2. 
Other than that, it occurs in the rows 101, 146 and 163.

The operation FINDS sets $s'$ and $t'$ to the lexicographically minimal children of the current block, it initializes the order profile variables with 0 and sets $f$ to S to indicate that a node in $S$ is being searched. It is called whenever we start looking for a node in $S$ with a certain order profile (this means right after the value of $h$ is updated). The operation FINDT also sets $s'$ and $t'$ to the lexicographically minimal children of the current block and $f$ to $T$. It is called after a node in $S$ with the prescribed order profile has been found. The operation FINDS occurs in rows 4, 73, 130 and 147. The operation FINDT occurs in rows 37, 91, 138 and 159.

The operation NXTS sets $t'$ to the lexicographically next larger node in the current block.
It is used during the computation of the order profile of $s'$. For example, in row 22 $t$ is set to 9 because this is the next larger node after $3$ in $B(x+1)$ of $T$. The operation NXTT is defined analogously.

The operation NCS (next candidate S) sets $s'$ to the lexicographically next larger node in the current block and resets $t'$ to the lexicographically minimal child. It also resets $seq$ and $sgt$ to 0. It is called whenever the order profile of the previous node in $s'$ did not match the prescribed order profile. It occurs in rows 134, 151 and 155. 
The operation NCT is defined analogously. It occurs in rows 62, 121, 125 and 142.

The operation PUSH sets $s$ to $s'$, $t$ to $t'$, $k$ to 0 and pushes the values of $h,sgt,seq,tgt,$ $teq,f$ onto $stk$. It can be seen as a recursive call to $\textsc{cmp}$ during a cross comparison. The operation RET2 occurs when such a recursive call returns. It restores the caller environment and updates the order profile of $s'$ or $t'$ (depending on the restored value of $f$) accordingly. It sets $s,t$ to the parents of $s,t$ (the old values), $s',t'$ to $s,t$, $k$ to $|S_s|$ and restores the values of $h,sgt,\dots,f$ by popping the from $stk$. If the restored value of $f$ is S and $res$ equals $\cong$ then $seq$ is incremented by one. If $res$ equals $\succ$ then $sgt$ is incremented by one. If the restored value of $f$ is $T$ then $teq$ or $tgt$ (or neither) is incremented depending on $res$. For example, in row 12 the node 1 of $S$ has been compared with 1 of $T$. Since $S_1 \prec T_1$ neither $seq$ nor $sgt$ is incremented. In row 31 $seq$ is incremented because $S_1 \cong T_{11}$.  

The following is a high-level overview of what happens in the trace in Table~\ref{tab:ti_trace3}. In row 4 the computation of the order profile of node 1 of $S$ begins. In rows 5, 14, 23, 28 and 33 it is compared with the nodes 1, 3, 9, 11 and 13 of $T$. In row 36 its order profile is complete. Since $sgt=h$ the program continues with searching a node in $T$ in row 37. It tries node $1$ of $T$ as first candidate. In row 61 its order profile $(3,2)$ is complete. Since $tgt = 3 \neq 0 = h$, the program tries node 3 of $T$ as next candidate in row 62. In rows 63, 68, 107, 112 and 117 it is compared with the nodes of $S$. In row 120 its order profile $(2,1)$ is complete. Since $tgt = 2 \neq 0 = h$ the program tries node 9 as next candidate for $T$. In row 124 its order profile $(3,2)$ is complete. The program tries node 11 as next candidate for $T$ in row 125. In row 128 its order profile is complete and since $tgt=0=h$ and $seq=teq$  the offset $h$ is increased by 2 in row 129. In row 130 the program starts searching for nodes in $S$ and $T$ with $sgt=tgt=2$. In row 137 and 145 it finds the nodes 3 and 3 with matching order profile $(2,1)$. In row 147 it starts searching for nodes with $sgt=tgt=3$. In row 158 and 162 it finds the nodes 9 and 1 with matching order profile $(3,2)$. Since $h$ equals the cardinality of $B(x+1)$ after row 163 it follows that the blocks $B(x+1)$ of $S$ and $T$ match. Therefore the program concludes that $S_1 \cong T_1$ in row 164. During the cross comparison of the nodes 1 and 1 a cross comparison of the nodes 4 and 4 begins in row 72 and ends in row 102 where the program concludes that $S_4 \cong T_4$.

\subparagraph*{Step 5 \& 6: Synthesize and complete the control flow graph.} 
The edges induced by the three traces are described in Table~\ref{tab:ti_adj}. An entry $t$:$r$ in a cell means that the $r$-th row in Table~$t$ witnesses the respective edge. A star `$\star$' indicates an edge which is not backed up by any of the traces.

There is an edge from GC and RET to $\prec$ for the same reason that there are such edges to $\succ$. The trace for the input $S=T_2, T=S_2$ witnesses these edges. 
There is an edge from INIT to $\cong$, $\prec$ and $\succ$. The edge to $\cong$ is witnessed by the input where $S,T$ are both the single node tree. In this case the non-recursive check after initialization succeeds and there is no next block to check. Therefore the program can conclude that the trees are isomorphic. The edge to $\prec$ is witnessed by the input where $S$ is the single node tree and $T$ is the tree with two nodes. In this case the non-recursive check returns $\prec_1$. The edge to $\succ$ is witnessed by swapping $S$ and $T$. An edge from PUSH to $\cong$ is witnessed by the input where $S$ and $T$ are the 3-node tree where the root node has two children. The following tree witnesses the edges from PUSH to $\prec$ and $\succ$: $(  \: ( (  ( ) )  ) \: ( ()() ) \: )$.  

Notice that INIT, GC and PUSH share the same set of outgoing edges. From a high-level perspective, after the execution of INIT, GC or PUSH a recursive call to $\textsc{cmp}$ is made since all of them assign new values to $s$ and $t$ and set $k$ to 0. For a clearer version of the control flow graph we introduce an additional program state CMP which is a noop and has INIT, GC and PUSH as in-neighbors. This leads to the control flow shown in Figure~\ref{fig:iso_cfg}. The nodes for $\cong,\prec,\succ$ and RET, RET2 occur multiple times for better readability.

\subparagraph*{Step 7: Add edge predicates.} 
All program states with a single outgoing edge have the edge predicate $\top$. Namely, these  are START, INIT, RET, GC, PUSH and SETH, FINDS, FINDT, NXTS, NXTT, NCS, NCT. 

%CMP
The edge from CMP to $\prec$ ($\succ$) is taken iff $\textsc{cmp}_1(s,t) =$  `$\prec_1$' (`$\succ_1$'). The edge from CMP to NB is taken iff    $\textsc{cmp}_1(s,t) =$ `$\cong_1$' and $s$ has a block larger than $k$. Otherwise, the edge from CMP to $\cong$ is taken.

%RET
The edge from RET to $\prec$ ($\succ$) is taken iff $res =$ `$\prec$' (`$\succ$'). The edge from RET to NB is taken iff $res = $ `$\cong$' and $s$ has a block larger than $k$. Otherwise, the edge to $\cong$ is taken.

%NB
After the execution of NB the next block size $k$ for comparison has been computed. If the  cardinality $l$ of $B(k)$ is 1 then continue with GC. Otherwise, continue with SETH.

%RET
Whether to go from $\cong,\prec,\succ$ to RET, RET2 or terminate depends on the cardinality of the block that is currently being compared and whether $s,t$ are the root nodes. If $s,t$ are the root nodes then terminate. Suppose this is not the case. Then, continue with RET iff the cardinality of $B(|S_s|)$ w.r.t.~the parent node of $s$ is 1. Otherwise, continue with RET2. 

%INCH
The edge from INCH to FINDS is taken iff $h$ is smaller than the cardinality $l$ of the current block. Otherwise, it holds that $h=l$ and one continues with $\cong$. 
 
%RET2
When trying to write down the edge predicates for the outgoing edges of RET2 by thinking about each one separately, it becomes apparent that this is a rather tedious and possibly redundant task. A natural alternative is to write a binary decision tree which determines what program state to visit next. Such a diagram implicitly describes the edge predicates of all outgoing edges of a given program state. A diagram for RET2 is shown in Figure~\ref{fig:bdd_ret2}. 
A dashed edge (left subtree) indicates that the predicate of the parent node does not hold. At the root node, the predicate $f=S$ describes whether the program is currently trying to find a node in $S$ (right subtree) or one in $T$ (left subtree). Let us consider the right subtree of the root node first. The predicate `$t'$ is last' means that $t'$ is the lexicographically last node in its block. If this does not hold then the order profile of $s'$ is still being computed and we have to continue with the next node for $t'$ (NXTS). Otherwise, the order profile of $s'$ is complete and we can check whether $h = sgt$. If this is the case then $s'$ is the node that we were looking for and we can continue with searching for its pendant in $T$ (FINDT). If $h \neq sgt$ then we should consider the lexicographically next sibling of $s'$ as candidate (NCS). However, this is only possible if such a next sibling exists, i.e.~the predicate `$s'$ is last' must be false. Otherwise, it holds that there is no node $v$ in the current block of $S$ with $gt_v = h$. It follows that $T_t \prec S_s$ ($\succ$). Now, suppose $f \neq S$, i.e.~we are looking for a node in $T$. If $s'$ is not the lexicographically last node in its block then the computation of the order profile of $t'$ must continue (NXTT). Otherwise, the order profile of $t'$ is complete. If $h \neq tgt$ then $t'$ is not the desired node and we have to consider the next candidate (NCT). Analogously, this is only possible if $t'$ is not the last node in its block. Otherwise, it follows that $S_s \prec T_t$ ($\prec$). If $h=tgt$ then we have to compare whether $seq=teq$ to check whether the isomorphism type of $s'$ and $t'$ occurs equally often in both subtrees. If this is the case we can set $h$ to the sum of the current order profile of either $s'$ or $t'$ since they are identical (INCH). Otherwise,  $seq < teq$ implies $S_s \prec T_t$.

\begin{figure}
	\begin{center}
		\begin{tikzpicture}[shorten >=1pt,auto,node distance=1.2cm,
main node/.style={draw},level 2/.style={sibling distance=8mm}
,level 3/.style={sibling distance=8mm}
,level 4/.style={sibling distance=8mm}
,level 5/.style={sibling distance=8mm}
,level distance=15mm
 ]

%\node at (-2.4,1.6) {$S$};
%\node at (2.65,1.6) {$T$};

%first one
\node at (0,0) {
	\Tree [.$f=S$
	\edge[dashed];
	[.{$s'$ is last} \edge[dashed]; [.{NXTT}  ]  [.{$h=tgt$}  \edge[dashed]; [.{$t'$ is last} \edge[dashed]; [.{NCT} ] [.{$\prec$}  ]  ] [.{$seq = teq$} \edge[dashed]; [.{$seq < teq$} \edge[dashed]; [.{$\prec$}  ] [.{$\succ$}  ]  ] [.{INCH}  ]  ]  ]  ]
	[.{$t'$ is last} \edge[dashed]; [.{NXTS} ] [.{$h = sgt$} \edge[dashed]; [.{$s'$ is last} \edge[dashed]; [.{NCS}  ] [.{$\succ$}  ]  ] [.{FINDT} ]  ]  ]
	]
};

\end{tikzpicture}
	\end{center}
	\caption{Decision tree for RET2; dashed edges indicate that the predicate does not hold}
	\label{fig:bdd_ret2}
\end{figure}
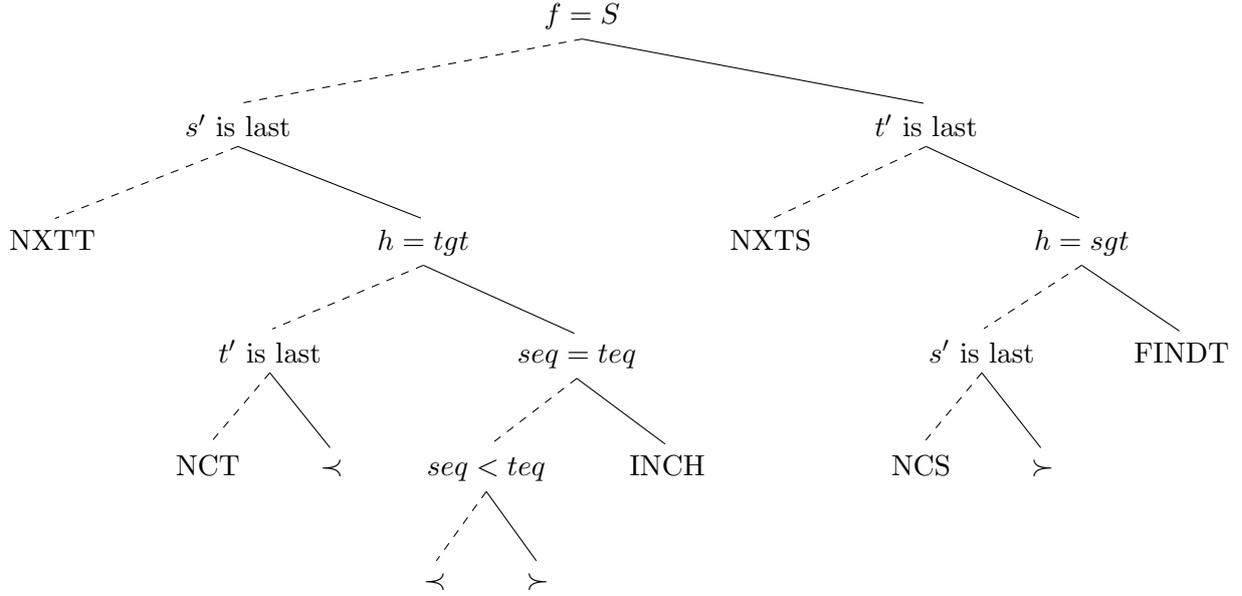

\begin{figure}
	\begin{center}
		\begin{tikzpicture}[shorten >=1pt,auto,node distance=1.2cm,
main node/.style={draw}]

\newcommand*{\xm}{1cm}%
\newcommand*{\ym}{1cm}%
\newcommand*{\xs}{0.5}%

\node[] (GHOST) at (-1.3,0) {};

\node[main node] (CMP) at (0,0) {CMP};
\node[main node, left = 0.7*\xm of CMP] (INIT) {INIT};
\node[main node, above = 1.2*\ym of INIT] (START) {START};
\node[left = 0.4*\xm of START] (PHANTOM) {};

\node[main node,fill={black!10}, below left = 3.8*\ym and -0.1*\xm of CMP] (ISOC) {$\cong$};
\node[main node,fill={black!10}, below = 0.7*\ym of ISOC] (NISO1C) {$\prec$};
\node[main node,fill={black!10}, below = 0.7*\ym of NISO1C] (NISO2C) {$\succ$};

\node[main node, left = 0.9*\xm of NISO1C] (RETA) {RET};
\node[main node, right = 0.9*\xm of NISO1C] (RET2A) {RET2};

\node[main node,fill={black!10}, below left = 1*\ym and 0.2*\xm of CMP] (NISO1A) {$\prec$};
\node[main node,fill={black!10}, left = 0.4*\xm of NISO1A] (ISOA) {$\cong$};
\node[main node,fill={black!10}, right = 0.4*\xm of NISO1A] (NISO2A) {$\succ$};

\node[main node, below = 0.7*\ym of NISO1A] (RET) {RET};

\node[main node, right = 2*\xm of CMP] (NB) {NB};
\node[main node, below left = 1*\ym and 0.5*\xm of NB] (GC) {GC};
\node[main node, below right = 1*\ym and 0.5*\xm of NB] (SETH) {SETH};

\node[main node, below left = 1*\ym and 0*\xm of SETH] (FINDS) {FINDS};
\node[main node, right = 4*\xm of FINDS] (PUSH) {PUSH};

\node[main node, below = 1.5*\ym of PUSH] (NXTT) {NXTT};
\node[main node, left = 0.4*\xm of NXTT] (NXTS) {NXTS};
\node[main node, left = 0.4*\xm of NXTS] (FINDT) {FINDT};
\node[main node, right = 0.4*\xm of NXTT] (NCS) {NCS};
\node[main node, right = 0.4*\xm of NCS] (NCT) {NCT};

\node[main node, below = 1.2*\ym of NXTT] (RET2) {RET2};
\node[main node, left = 4.2*\xm of RET2] (INCH) {INCH};

\node[main node,fill={black!10}, below = 0.6*\ym of INCH] (ISOB) {$\cong$};

\node[main node,fill={black!10}, below left = 0.6*\ym and -0.1*\xm of RET2] (NISO1B) {$\prec$};
\node[main node,fill={black!10}, below right = 0.6*\ym and  -0.1*\ym of RET2] (NISO2B) {$\succ$};

%
%%\node[main node,fill={black!10}] (YES) at (2.3*\xm+\xs,-2.5) {YES};
%%\node[main node,fill={black!10}] (NO) at (1.7*\xm+\xs,-2.5) {NO};
%
%\node[main node, right = 2*\xm of INIT] (LEN) {LEN};
%
%\node[main node,fill={black!10}, below left = 1.8cm and 0.7cm of LEN] (NO) {NO};
%\node[main node,fill={black!10}, below right = 1.8cm and 0.7cm of LEN] (YES) {YES};
%
%\node[main node, right = 3.2*\xm of LEN] (ANBS) {ANBS};
%\node[main node, above = \ym of ANBS] (ASBN) {ASBN};
%\node[main node, below = \ym of ANBS] (ANBN) {ANBN};
%
%
%\node[main node, right = 1.6*\xm of ANBS] (NOOP) {NOOP};
%%\node[fill=black, circle, inner sep=0cm, minimum size=1.6mm, right = 2.1*\xm of ANBS] (NOOP) {};
%
%
\path[->]
(CMP) edge (NB)
%(RET) edge (CMP)
(CMP) edge (ISOA)
(CMP) edge (NISO1A)
(CMP) edge (NISO2A)

(RET) edge (ISOA)
(RET) edge (NISO1A)
(RET) edge (NISO2A)
(RET) edge[bend right=50] (NB)

(NB) edge (GC)
(NB) edge (SETH)
(GC) edge (CMP)

(SETH) edge (FINDS)
(FINDS) edge (PUSH)

(FINDT) edge (PUSH)
(NXTS) edge (PUSH)
(NXTT) edge (PUSH)
(NCS) edge (PUSH)
(NCT) edge (PUSH)

(RET2) edge (FINDT)
(RET2) edge (NXTS)
(RET2) edge (NXTT)
(RET2) edge (NCS) 
(RET2) edge (NCT) 

(RET2) edge (NISO1B) 
(RET2) edge (NISO2B) 

(PUSH) edge[bend right=50] (CMP)

(INCH) edge (FINDS)
(RET2) edge (INCH)
(INCH) edge (ISOB)

(INIT) edge (CMP)
(START) edge (INIT)
(PHANTOM) edge (START)

(ISOC) edge (RETA)
(ISOC) edge (RET2A)
(NISO1C) edge (RETA)
(NISO1C) edge (RET2A)
(NISO2C) edge (RETA)
(NISO2C) edge (RET2A)
;
%(GHOST) edge (INIT)
%(INIT) edge node[midway, above, sloped]{\scriptsize EQLEN} (LEN)
%(INIT) edge node[midway, above, sloped]{\scriptsize $\neg$EQLEN} (NO)
%(LEN) edge node[midway, above, sloped]{\scriptsize $\neg$SS} (NO)
%(LEN) edge node[midway, above, sloped]{\scriptsize SS $\wedge$ FC} (YES)
%(LEN) edge node[midway, above, sloped]{\scriptsize $\neg$EOA $\wedge$ EOB $\wedge$ SS} (ASBN)
%(LEN) edge node[midway, above, sloped]{\scriptsize EOA $\wedge$ $\neg$EOB $\wedge$ SS} (ANBS)
%(LEN) edge node[midway, above, sloped]{\scriptsize EOA $\wedge$ EOB  $\wedge$ SS $\wedge$ $\neg$FC} (ANBN)
%
%(ASBN) edge node[midway, above, sloped]{\scriptsize $\top$} (NOOP)
%(ANBS) edge node[midway, above, sloped]{\scriptsize $\top$} (NOOP)
%(ANBN) edge node[midway, above, sloped]{\scriptsize $\top$} (NOOP)
%
%(NOOP.90) edge[bend right=75] node[midway, above, sloped]{\scriptsize $\top$} (LEN)
%;
%
%%\path[-,shorten >=0pt,shorten >=0pt]
%%(ASBN) edge  (NOOP)
%%(ANBS) edge (NOOP)
%%(ANBN) edge (NOOP)
%%;
\end{tikzpicture}
	\end{center}
	\caption{Control flow graph for Lindell's tree isomorphism algorithm}
	\label{fig:iso_cfg}
\end{figure}
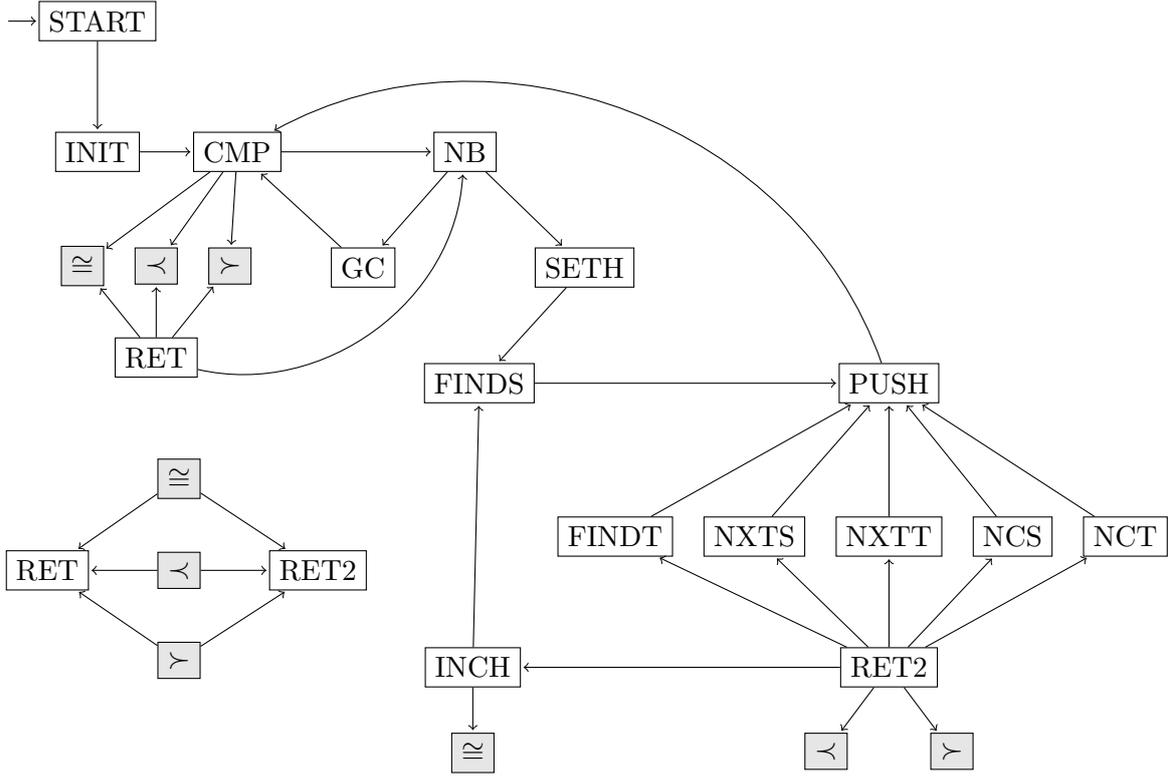

\subsection{Relation to Formal Foundation}

A generalized trace is an extended trace.
% in our formalism. 
%The virtual machines that were used as model of computations in the two examples have not been made explicit simply because there is no need. Making a model of computation explicit beforehand is only relevant when operations and edge predicates should be automatically inferred from the traces; in our case this information is supplied manually.
%However, for the sake of completeness let us outline the virtual machines' parameters. 
%The set of variables is determined by the algorithm. For example, in the case of the confidential string matching algorithm the variables $V$ of the virtual machine are $A,B$ (arrays of strings), $ca,cb,oa,ob,l$ (integers) and $r$ (boolean). The set of functions $F$ can be defined as the set of functions that can be written as composition of $c$ functions from a predefined set of available functions, where $c$ is a sufficiently large constant. In practice, the set of available functions consists of built-in functions provided by the programming language and functions to which we want to reduce the problem at hand. For example, the function $\textsc{cmp}_1$ from the tree isomorphism example is of the latter kind.
%Notice that the space of operations and predicates of the virtual machines can be very large thus making automatic synthesis difficult.  
Recall that finding a program with $k$ program states for a set of extended traces $T$ is equivalent to finding a restricted $k$-coloring of the graph $G(T)$; see Definition~\ref{def:trace_graph} and Theorem~\ref{thm:color}.
A vertex of $G(T)$ corresponds to a row of a generalized trace.
A simple observation is that two vertices of $G(T)$ must receive different colors if they are associated with different operations. The synthesis procedure of our method is based on the heuristic that the converse is also true: two vertices receive the same color iff they are associated with the same operation. 
This heuristic yields a correct coloring if every operation occurs at most once in the program, i.e. there are no two program states with the same operation. It is not difficult to come up with instances where this assumption is violated. 
%For example, it is impossible to synthesize the $\cmcm[2]$-program in Figure~\ref{fig:double} (or an equivalent program) with this procedure. However, our method is only intended to work with virtual machines as underlying model of computation. 
%Nonetheless, it is possible that one wants to implement an algorithm where this heuristic fails. 

For example, suppose you want to implement an algorithm which increases a counter $i \leftarrow i + 1$ in two different situations A and B. This means in a program which implements this algorithm there should be two distinct program states which both are associated with the operation $i \leftarrow i + 1$. An additional variable $s$ can be introduced and for every row in a trace with the operation $i \leftarrow i + 1$ one has to add either the assignment $s \leftarrow \text{A}$ or $s \leftarrow \text{B}$. The variable $s$ is essentially used to name program states in certain cases where the operation itself is not sufficient to distinguish between program states. In a sense, our synthesis procedure can be seen as a default coloring of $G(T)$ which can be refined using $s$. An advantage of this approach is that $G(T)$ never has to be explicitly constructed; on the other hand, this implies that the programmer must be able to intuitively sense when the heuristic fails. Due to the parallel assignments the space of operations is rather rich and if one additionally refrains from reusing a variable for different purposes then the fraction of cases where this heuristic fails seems negligible.

\section{Conclusion}
\label{sec:con}
The programming method presented here has multiple advantages over ad hoc programming. 
Most importantly, it provides increased confidence in the correctness of the constructed program at the cost of writing traces.

The traces can not only serve as a form of correctness proof but also as means of documentation. For example, in the case of the confidential string matching problem one could write a program which converts the contents of the variables into a graphic similar to Figure~\ref{fig:csma}. To comprehend the program one could just watch an animated version of the traces. 
One can also query for justification of certain program parts. For example, why is there a certain edge in the control flow graph? The answer would be two consecutive rows in a trace which cause this edge (or multiple such results). Similarly, for an operation the set of examples  which it generalizes can be provided. Conversely, parts of the program which are not backed up by traces can be marked and the programmer can provide a comment as justification.  

Producing correct programs is more likely due to the fact that the programmer is not overwhelmed with a myriad of details at once---as is common in the unstructured ad hoc approach---but can separately deal with each component of the program one step at a time. Additionally, the mean time between introducing and spotting an error is reduced since every newly introduced part of the program is immediately verified against the traces. This facilitates debugging and spares the programmer tracking down error sources, which can be very time-consuming.  

Another benefit of this method is that it is language-agnostic. This means it is compatible with any high-level functional or imperative programming language. Taking this a step further, assuming a standardized set of functions is available across different programming languages, it is possible to construct a program using this method and automatically translate it into any of these languages with minimal effort. Moreover, our method requires no strong mathematical background thus making it more accessible to the average programmer.
%Virtual machines can serve as theoretical basis for teaching programming. In contrast to a programming language, they abstract away unnecessary features which potentially obscure the crucial aspects of programming. Ideally, this should make it feasible to learn programming using only pen and paper. For example, standard algorithms such as breadth-first search could be taught by example and then the task would be to program them, thereby separating algorithm design from programming.  

%Unlike more formal programming methods, our method requires no strong mathematical background. This makes it more accessible to the average programmer. We suspect that the ability to select a set of inputs which exhibits a program's complete behavior is related to the ability to prove its correctness.  Proving that a program $P$ is correct w.r.t.~its input/output-behavior can be reduced to (1) finding an equivalence relation $\sim$ with finite index on the set of inputs with the property ``if $P$ yields the correct output on input $x$ then $P$ does so for every input $y$ with $x \sim y$'' and (2) showing that $P$ yields the correct output for a representative input of every equivalence class of $\sim$. The set of representative inputs should be the ones for which the traces are produced. 
%Stated differently, thinking about why a program is correct and for what inputs to produce traces appear to be related tasks. Consequently, learning to select inputs in the context of our method could be regarded as an intermediate step in learning to prove a program's correctness.

In order to be useful for software development in practice, there needs to be software support which takes care of the menial tasks such as verifying operations and edge predicates against the traces and synthesizing the control flow. 
%We recommend to visualize programs as graphs because this seems to be the clearest and most natural form of representation. After all, a program has no inherent linear structure.
An important aspect that we have not considered here is that many programs interact with their environment in  more complex ways than just getting the input and then returning the output at the end of the computation. Therefore a way of integrating such interaction in this programming method without compromising its advantages needs to be developed.

%\printbibliography[heading=bibintoc]
\bibliographystyle{alpha}
\bibliography{lit.bib}

\begin{thebibliography}{DDH72}

\bibitem[Bac09]{back}
Ralph{-}Johan Back.
\newblock Invariant based programming: basic approach and teaching experiences.
\newblock {\em Formal Asp. Comput.}, 21(3):227--244, 2009.

\bibitem[BBP75]{bierspeed}
Alan~W. Biermann, Richard~I. Baum, and Frederick~E. Petry.
\newblock Speeding up the synthesis of programs from traces.
\newblock {\em {IEEE} Trans. Computers}, 24(2):122--136, 1975.

\bibitem[Bie72]{biertm}
Alan~W. Biermann.
\newblock On the inference of turing machines from sample computations.
\newblock {\em Artif. Intell.}, 3(1-3):181--198, 1972.

\bibitem[BK76]{bier}
A.~W. Biermann and R.~Krishnaswamy.
\newblock Constructing programs from example computations.
\newblock {\em IEEE Transactions on Software Engineering}, SE-2(3):141--153,
  Sep. 1976.

\bibitem[DDH72]{dahl}
O.~J. Dahl, E.~W. Dijkstra, and C.~A.~R. Hoare, editors.
\newblock {\em Structured Programming}.
\newblock Academic Press Ltd., London, UK, UK, 1972.

\bibitem[Dij70]{dijk1}
E.~W. Dijkstra.
\newblock Concern for correctness as a guiding principle for program
  composition.
\newblock 1970.

\bibitem[Dij76]{dijk2}
Edsger~W. Dijkstra.
\newblock {\em A Discipline of Programming}.
\newblock Prentice-Hall, 1976.

\bibitem[Gri87]{gries}
David Gries.
\newblock {\em The Science of Programming}.
\newblock Springer-Verlag, Berlin, Heidelberg, 1st edition, 1987.

\bibitem[Gur00]{gurevich}
Yuri Gurevich.
\newblock Sequential abstract-state machines capture sequential algorithms.
\newblock {\em ACM Trans. Comput. Logic}, 1(1):77--111, July 2000.

\bibitem[Gur04]{gurevich2}
Yuri Gurevich.
\newblock Abstract state machines: An overview of the project.
\newblock In Dietmar Seipel and Jos{\'e}~Mar{\'i}a Turull-Torres, editors, {\em
  Foundations of Information and Knowledge Systems}, pages 6--13, Berlin,
  Heidelberg, 2004. Springer Berlin Heidelberg.

\bibitem[Lam18]{lamport}
Leslie Lamport.
\newblock If you're not writing a program, don't use a programming language.
\newblock {\em Bulletin of the {EATCS}}, 125, 2018.

\bibitem[Lin92]{lindell}
Steven Lindell.
\newblock A logspace algorithm for tree canonization (extended abstract).
\newblock In {\em Proceedings of the Twenty-fourth Annual ACM Symposium on
  Theory of Computing}, STOC '92, pages 400--404, New York, NY, USA, 1992. ACM.

\bibitem[Sco67]{dana}
Dana Scott.
\newblock Some definitional suggestions for automata theory.
\newblock {\em Journal of Computer and System Sciences}, 1(2):187 -- 212, 1967.

\bibitem[SE99]{ek}
Stefan Schr{\"{o}}dl and Stefan Edelkamp.
\newblock Inferring flow of control in program synthesis by example.
\newblock In {\em {KI-99:} Advances in Artificial Intelligence, 23rd Annual
  German Conference on Artificial Intelligence, Bonn, Germany, September 13-15,
  1999, Proceedings}, pages 171--182, 1999.

\bibitem[vE14]{emdenmc}
M.~H. van Emden.
\newblock Matrix code.
\newblock {\em Sci. Comput. Program.}, 84:3--21, 2014.

\bibitem[vE18]{emdencbc}
M.~H. van Emden.
\newblock Correct by construction.
\newblock {\em CoRR}, abs/1812.09411, 2018.

\bibitem[Wir73]{wirth2}
Niklaus Wirth.
\newblock {\em Systematic Programming: An Introduction}.
\newblock Prentice Hall PTR, Upper Saddle River, NJ, USA, 1973.

\end{thebibliography}

\pdfbookmark[section]{Appendix}{appendix}
\section*{Appendix}

\begin{proof}[Proof of Proposition~\ref{prop:ttopt}]
	``$\Rightarrow$'': Let $T$ be consistent via a program $P=(G,v_0,\alpha,\beta)$. 
	Assume for the sake of contradiction that the RHS of the `iff' does not hold. Let $T'$ be a set of traces in the quotient set $T / \sim_0$ which witnesses this.
	Assume there exist traces $\vec{s},\vec{t}$ in $T'$ such that $|\vec{s}| = 1$ and $|\vec{t}| > 1$. 
	Let $s_0$ and $t_0$ denote the first elements of $\vec{s}$ and $\vec{t}$ respectively. This contradicts that $T$ is consistent via $P$ because $s_0 \sim t_0$. This implies $|\vec{s}| > 1$ holds for all traces $\vec{s}$ in $T'$. For a trace $\vec{s}$ in $T'$ let $g_{\vec{s}}$ be the first operation executed by $P$ on input $s_0$ where $s_0$ is the first element of $\vec{s}$. Since $T'$ is consistent via $P$ it follows that $g_{\vec{s}} = g_{\vec{t}}$ holds for all $\vec{s},\vec{t}$ in $T'$. Therefore there exists an operation $g$ (any $g_{\vec{s}}$) such that for all traces $\vec{s} = (s_0,s_1,\dots,s_n)$ in $T'$ it holds that $g(s_0) = s_1$. It remains to argue that $\left\{ (\vec{s})' \mid \vec{s} \in T' \right\}$ is consistent.
	Let $v$ be the program state which is reached after executing $P$ on input $s_0$ for one step for any trace $\vec{s}$ in $T'$ and $s_0$ is the first element of $\vec{s}$; since $T'$ is a $\sim_0$-equivalence class $v$ must be the same program state irregardless of the choice of $\vec{s}$. Let $P'$ be the program that is obtained by modifying $P$ as follows. Remove all edges from the start state $v_0$ of $P$. For every edge $(v,w)$ in $G$ add an edge $(v_0,w)$ and let $\beta(v_0,w) = \beta(v,w)$. It follows that $\left\{ (\vec{s})' \mid \vec{s} \in T' \right\}$ is consistent via $P'$ since $P'(s_1) = P(s_1,v) = (\vec{s})'$ holds for all traces $\vec{s} = (s_0,s_1,\dots,s_n)$ in $T'$. 
	
	``$\Leftarrow$'': Let the RHS of the `iff' hold. We construct a program $P=(G,v_0,\alpha,\beta)$ which shows that $T$ is consistent. For every $T'$ in the quotient set $T / \sim_0$ such that every trace in $T'$ has more than one element we modify $P$ as follows. Let $g$ be an operation such that for all traces $\vec{s} = (s_0,s_1,\dots,s_n)$ in $T'$ it holds that $g(s_0) = s_1$ and let $P' = (G',v_0',\alpha',\beta')$ be a program which shows that $\left\{ (\vec{s})' \mid \vec{s} \in T' \right\}$ is consistent. We assume w.l.o.g. that the vertex set of $G'$ and $G$ are disjoint. Add a new vertex $v_{T'}$ to $G$ and an edge from the start state $v_0$ to $v_{T'}$. Let $\alpha(v_{T'}) = g$ and the edge predicate $[\beta(v_0,v_{T'})](x) = 1 \Leftrightarrow x = \predseq(s_0)$ for any $\vec{s} = (s_0,s_1,\dots,s_n)$ in $T'$.
	Additionally, add $P'$ to $P$ such that the start state of $P'$ is $v_{T'}$, i.e.~$P$ calls $P'$. 
\end{proof}
%
%\begin{proof}[Proof of Proposition~\ref{prop:rttet}]
%	``$\Rightarrow$'': Let $T$ be $k$-consistent via a program $P$.
%	For every trace $\vec{s} = (s_0,s_1,\dots,s_n)$ in $T$ let $\sigma(\vec{s}) := P[s_0]$. It holds that $\left\{  (\vec{s},\sigma(\vec{s})) \mid \vec{s} \in T  \right\}$ is $k$-consistent via $P$ by definition. 
%	
%	``$\Leftarrow$'': Follows from the fact that $k$-consistency of a set of extended traces \\
%	$\{ (\vec{s_1},\vec{g}_1) ,\dots, (\vec{s_n},\vec{g}_n) \}$ implies $k$-consistency of $\{ \vec{s_1} ,\dots, \vec{s_n} \}$.	
%\end{proof}

\begin{definition}	
	Let $k \in \N$ and let $T = \{ X_1,\dots,X_n \}$ be a set of extended traces. We say a function $c \colon V(G(T)) \rightarrow [k]$ is a restricted $k$-coloring of $G(T)$ if 
	no two adjacent vertices in $G(T)$ receive the same color
	and, additionally, for all $i,i' \in [n]$ and $j \in \{0,\dots,|X_i|\}$, $j' \in \{0,\dots,|X_{i'}|\}$ it holds that if 
	$$ c(v_{i,j}) = c(v_{i',j'}) \wedge s(v_{i,j}) \sim s(v_{i',j'}) $$
	then
	\underline{either} the successors of $j$ and $j'$ must receive the same color:
	$$ j < |X_i| \wedge j' < |X_{i'}| \wedge c(v_{i,j+1}) = c(v_{i',j'+1}) $$	
	\underline{or} they have no successors:
	$$  j = |X_i| \wedge j' = |X_{i'}| $$
	%$j < |X_i|$, $j' < |X_{i'}|$ and $c(v_{i,j+1}) = c(v_{i',j'+1})$ or $j = |X_i|$ and $j' = |X_{i'}|$,
	where $s(\cdot)$ denotes the machine state that is associated with a vertex via its corresponding line. 
	\label{def:crestr}
\end{definition}

%The additional restriction for the coloring can be intuitively interpreted as follows. 
%A color corresponds to a program state. If two lines receive the same color this means they occur at the same program state. If two lines receive the same color and their corresponding machine states are indistinguishable then either both lines are the last or their successor lines must receive the same color (program state) as well because the program must be deterministic. 

\begin{proof}[Proof of Theorem~\ref{thm:color}]
	``$\Rightarrow$'': Assume $T=\{X_1,\dots,X_n\}$ is $k$-consistent via a program $P$ whose set of program states is $[k]$.     
	Let $c(v_{i,j})$ be the program state which is reached when executing $P$ on $s_0$ for $j$ steps where $s_0$ is the first machine state of $X_i$, for $i \in [n]$ and $j \in \{0,\dots,|X_i|\}$. We claim that $c$ is a restricted $k$-coloring of $G(T)$.
	
	First, we argue that $c$ is a coloring of $G(T)$, i.e.~no two adjacent vertices receive the same color.    
	To show this let $X,Y$ be extended traces in $T$.
	Let $n = |X|, m = |Y|$	and let $(i,s,g)$ be the $i$-th line of $X$ and $(j,t,h)$ be the $j$-th line of $Y$. Assume that both these lines are associated with the program state $v$. Assume for the sake of contradiction that these lines are unmergeable. We distinguish between the following cases depending on $(i,j)$. If $i=0$ and $j \neq 0$ then $(i,s,g)$ and $(j,t,h)$ are unmergeable. However, this cannot be the case as it contradicts that $(i,s,g)$ and $(j,t,h)$ are both associated with the same program state $v$. The $0$-th line of $X$ must be associated with the start state of $P$ whereas the $j$-th line of $Y$ must be associated with a program state which is not the start state since $j \neq 0$. For the same reason the case $i\neq 0, j=0$ cannot occur. Therefore we can assume that $(i,j) \in ([n] \times [m]) \cup \{(0,0)\}$. Since $(i,s,g)$ and $(j,t,h)$ are unmergeable it must hold that $U_{i,j} = 1$. This is the case if $g \neq h$ or $s \sim t$ and $U_{i+1,j+1} = 1$. Since $(i,s,g)$ and $(j,t,h)$ are both associated with the same program state it must hold that $g = h$. Therefore it must be the case that $s \sim t$ and $U_{i+1,j+1} = 1$. Since $U_{n+1,m+1} = 0$ by definition it cannot be the case that $i=n,j=m$. Assume that $i+1 = n+1$ and $j+1 \leq m$. Since $i+1 = n+1$ this means $(i,s,g)$ is the last line of $X$ whereas $(j,t,h)$ is not the last line of $Y$ since $j+1 \leq m$. Since $s \sim t$ and both lines lead to the same program state this means $P$ has to behave identically in both cases. However, in order to be consistent with $X$ the program $P$ would have to terminate whereas in the case of $Y$ the program $P$ would have to execute at least one more operation. Therefore these values for $(i,j)$ are not possible. For the same reason it cannot hold that $i+1 \leq n$ and $j+1 = m +1$. It remains to consider the case $i+1 \leq n, j+1 \leq m$. In this case the $i$-th line of $X$ and the $j$-th line of $Y$ are unmergeable iff the $(i+1)$-th line of $X$ and the $(j+1)$-th line of $Y$ are unmergeable. Moreover, the $(i+1)$-th line of $X$ and the $(j+1)$-th line of $Y$ must be associated with the same program state since $s \sim t$. This means we can apply the same argument for the $(i+1)$-th line of $X$ and the $(j+1)$-th line of $Y$. Following this recursive chain of equivalences eventually leads to one of the three contradictory cases discussed before ($i+1=n+1, j+1\leq m$; $i+1\leq n, j+1=m+1$; $i+1=n+1,j+1=m+1$). 
	
	It remains to argue that $c$ has the restricted property. Consider two vertices $v_{i,j}, v_{i',j'}$ with the same color and their associated machine states are indistinguishable. In both cases $P$ behaves identically because they are associated with the same program state and they are indistinguishable. This means $P$ either terminates (both lines are the last) or it continues with the same program state (the successor lines have the same color).    
	
	``$\Leftarrow$'': Assume $T=\{X_1,\dots,X_n \}$ is consistent and 
	$c \colon V(G(T)) \rightarrow [k]$ is a restricted $k$-coloring of $G(T)$.    
	%$G(T)$ is $k$-colorable via a function $c \colon V(G(T)) \rightarrow [k]$. 
	We claim that $V_0 := \set{v_{i,0}}{ i \in [n]}$ is an independent set in $G(T)$ and for all $u \in V_0$ and $v \notin V_0$ it holds that $\{u,v\}$ is an edge in $G(T)$. For the first part suppose there exist $i,j \in [n]$ such that $v_{i,0}$ and $v_{j,0}$ are adjacent. This means the $0$-th line of $X_i$ and the $0$-th line of $X_j$ are unmergeable. However, this contradicts that $T$ is consistent. The second part holds because the $0$-th line and $j$-th line of some extended traces with $j > 0$ are unmergeable by definition. Therefore we can assume that the $0$-th lines all receive the same color which is distinct from all other lines, i.e.~$c(u) = c(v)$ for all $u,v \in V_0$ and $c(u) \neq c(v)$ for all $u \in V_0$ and $v \notin V_0$. Moreover, we assume w.l.o.g.~that every color is used at least once, i.e.~the image of $c$ is $[k]$.  
	
	\textbf{Construction.}
	We construct a program $P=(G,v_0,\alpha,\beta)$ with $k$ states which shows that $T$ is $k$-consistent. The control flow graph $G$ has $[k]$ as vertex set and the start state is $c(v_{1,0})$. 
	
	Before we define the other components of $P$ we introduce some auxiliary sets $L(\cdot)$ and $L(\cdot,\cdot)$.
	For a vertex $u$ of $G$ let $L(u)$ be the set of tuples $(i,j,s,g)$ such that $(j,s,g)$ is the $j$-th line of $X_i$ and $c(v_{i,j}) = u$, i.e.~the set of lines associated with the color $u$. For $u,v \in V(G)$ and $(i,j,s,g) \in L(u), (i',j',s',g') \in L(v)$ we say $(i,j,s,g)$ and $(i',j',s',g')$ are consecutive if $i=i'$ (the lines are from the same trace) and $j+1 = j'$. For $u,v \in V(G)$ let $L(u,v)$ be the set of tuples $(\ell_u,\ell_v)$ with $\ell_u \in L(u), \ell_v \in L(v)$ and $\ell_u$ and $\ell_v$ are consecutive. 
	
	The function $\alpha$ is defined as follows. For a non-start state $u \in V(G)$ choose some $(i,j,s,g) \in L(u)$; since every color is assumed to be used at least once such a tuple must exist. Let $\alpha(u) = g$. This is well-defined because the choice of $(i,j,s,g)$ is irrelevant. Stated differently, for all $(i,j,s,g),(i',j',s',g') \in L(u)$ it holds that $g = g'$. Suppose this would not be the case, i.e.~$g \neq g'$. This would imply that the $j$-th line of $X_i$ and the $j'$-th line of $X_{i'}$ are unmergeable and therefore $v_{i,j}$ and $v_{i',j'}$ are adjacent in $G(T)$. This contradicts that these two vertices have the same color. 
	There is an edge $(u,v)$ in $G$ if $L(u,v)$ is non-empty. 
	Let $(u,v)$ be an edge in $G$. The edge predicate of $(u,v)$ is defined as follows. Let $[\beta(u,v)](x) = 1$ if there exists $(\ell_u,\ell_v) \in L(u,v)$ with $\ell_u = (\cdot,\cdot,s,\cdot)$ such that $x = \predseq(s)$. The restricted property of $c$ guarantees that $P$ is deterministic.    
\end{proof}

\begin{proof}[Proof of Proposition~\ref{prop:ftrepr}]
	We assume w.l.o.g.~that $P$ is minimal and has $k$ states; otherwise replace $P$ by an equivalent program which is minimal. Let $\mathbb{P}_k$ be the the set of programs with $k$ states up to equivalence. Let $I_k$ be the set of machine states such that for every $P \neq P' \in \mathbb{P}_k$ there exists $s \in I_k$ with $P(s) \neq P'(s)$. Stated differently, two programs $P,P'$ with $k$ program states are equivalent iff $P(s) = P'(s)$ holds for all $s \in I_k$. Let $T$ be the set of traces $P(s)$ for every $s \in I_k$. Let $P'$ be a program which witnesses that $T$ is consistent and $P'$ is minimal. It holds that $P'$ has $k$ states. Assume that $P'$ and $P$ are not equivalent. This implies that there exists an $s \in I_k$ such that $P'(s) \neq P(s)$. This contradicts that $P(s)$ is in $T$.
\end{proof}

%The restricted coloring encodes a program with $k$ program states which witnesses the consistency of $T$. 

\begin{table}[h!]
	\caption{CSM: Trace for input from Figure \ref{fig:csma}}    
	\begin{center}
		%\footnotesize
		\rowcolors{2}{gray!15}{white}
\begin{tabular}{ r | l l l l l l }
	& $ca$ & $cb$ & $oa$ & $ob$ & $l$ & $r$ \\
	\hline 
	1  & 1	& 1	 & 0  & 0  &   &  \Tstrut \Bstrut \\ 
	2  & 	& 	 &    &    & 2 &   \Tstrut \Bstrut \\  
	3  & 	& 2	 & 2  & 0  &   & \Tstrut \Bstrut \\ 
	4  & 	& 	 &    &    & 2 &  \Tstrut \Bstrut \\   
	5  & 2	& 	 & 0  & 2  &   & \Tstrut \Bstrut \\  
	6  & 	& 	 &    &    & 1 &  \Tstrut \Bstrut \\  
	7  & 	& 3	 & 1  & 0  &   & \Tstrut \Bstrut \\    
	8  & 	& 	 &    &    & 1 &  \Tstrut \Bstrut \\  
	9  & 3	& 	 & 0  & 1  &   & \Tstrut \Bstrut \\     
	10 & 	& 	 &    &    & 2 &  \Tstrut \Bstrut \\   
	11 & 4	& 4	 & 0  & 0  &   & \Tstrut \Bstrut \\      
	12 &  	& 	 &    &    & 2 & \Tstrut \Bstrut \\   
	13 &  	& 	 &    &    &   & $\top$ \Tstrut \Bstrut \\    
\end{tabular}
	\end{center}            
	\label{tab:ltrace1}
\end{table}

\begin{table}[h!]
	\caption{CSM: Generalized trace for input from Figure \ref{fig:csma}}    
	\begin{center}
		%\footnotesize
		\rowcolors{2}{gray!15}{white}
\begin{tabular}{ r | l l l l l l }
    & $ca$ & $cb$ & $oa$ & $ob$ & $l$ & $r$ \\
    \hline 
    1  & 1 := 1	& 1 := 1	 & 0 := 0  & 0 := 0  & &   \Tstrut \Bstrut \\ 
    2  & 	& 	 &    &    & 2 := $|B[cb]| - ob$ &  \Tstrut \Bstrut \\  
    3  & 	& 2 := $cb+1$	 & 2 := $oa+l$  & 0 := 0  & &   \Tstrut \Bstrut \\ 
    4  & 	& 	 &    &    & 2 := $|A[ca]| - oa$ &  \Tstrut \Bstrut \\   
    5  & 2 := $ca+1$	& 	 & 0 := 0  & 2 := $ob+l$  & &   \Tstrut \Bstrut \\  
    6  & 	& 	 &    &    & 1 := $|B[cb]| - ob$ & \Tstrut \Bstrut \\  
    7  & 	& 3 := $cb+1$	 & 1 := $oa+l$  & 0 := 0  & &   \Tstrut \Bstrut \\    
    8  & 	& 	 &    &    & 1 := $|A[ca]| - oa$ & \Tstrut \Bstrut \\  
    9  & 3 := $ca+1$	& 	 & 0 := 0  & 1 := $ob+l$  & &   \Tstrut \Bstrut \\     
    10 & 	& 	 &    &    & 2 := $|A[ca]| - oa$&   \Tstrut \Bstrut \\   
    11 & 4 := $ca+1$	& 4 := $cb+1$	 & 0 := 0  & 0 := 0  & &   \Tstrut \Bstrut \\      
    12 &  	& 	 &    &    & 2 := $|A[ca]| - oa$&   \Tstrut \Bstrut \\   
    13 &  	& 	 &    &    & &  $\top$ := $\top$ \Tstrut \Bstrut \\   
\end{tabular}
	\end{center}            
	\label{tab:btrace1}
\end{table}

\begin{table}[h!]
	\caption{CSM: Trace for $A = [\text{a, a, a}], B = [\text{aaa}] $}     
	\begin{center}
		%\footnotesize
		%[a,a,a], [aaa]
\rowcolors{2}{gray!15}{white}
\begin{tabular}{ r | l | l l l l l l }
    & Name & $ca$ & $cb$ & $oa$ & $ob$ & $l$ & $r$ \\
    \hline 
 1  & INIT & 1	& 1	 & 0  & 0  &   &  \Tstrut \Bstrut \\ 
 2  & ALEN & 	& 	 &    &    & 1 &   \Tstrut \Bstrut \\  
 3  & ANBS & 2	& 	 & 0  & 1  &   & \Tstrut \Bstrut \\ 
 4  & ALEN & 	& 	 &    &    & 1 &  \Tstrut \Bstrut \\   
 5  & ANBS & 3	& 	 & 0  & 2  &   & \Tstrut \Bstrut \\  
 6  & ALEN & 	& 	 &    &    & 1 &  \Tstrut \Bstrut \\  
 7  & YES &  	& 	 &    &    &   & $\top$ \Tstrut \Bstrut \\     
\end{tabular}
	\end{center}            
	\label{tab:ltrace2}
\end{table}

\begin{table}[h!]
	\caption{CSM: Trace for $A = [\text{ba, a}], B = [\text{b, ab}] $}     
	\begin{center}
		%\footnotesize
		%[ba,a], [b,ab]
\rowcolors{2}{gray!15}{white}
\begin{tabular}{ r | l | l l l l l l }
    & Name & $ca$ & $cb$ & $oa$ & $ob$ & $l$ & $r$ \\
    \hline 
    1 & INIT & 1	& 1	 & 0  & 0  &   &  \Tstrut \Bstrut \\ 
    2 & BLEN  & 	& 	 &    &    & 1 &   \Tstrut \Bstrut \\  
    3 & ASBN & 	& 2	 & 1   & 0  &  &   \Tstrut \Bstrut \\    
    4 & ALEN & 	& 	 &    &    & 1 &   \Tstrut \Bstrut \\  
    5 & ANBS & 2	& 	 & 0   & 1  &  &   \Tstrut \Bstrut \\    
    6 & BLEN & 	& 	 &    &    & 1 &   \Tstrut \Bstrut \\        
    7 & NO &  	& 	 &    &    &   & $\bot$ \Tstrut \Bstrut \\    
\end{tabular}
	\end{center}            
	\label{tab:ltrace3}
\end{table}

\begin{table}[h!]
	\caption{CSM: Trace for $A = [\text{a}], B = [\text{aa}] $}    
	\begin{center}
		%\footnotesize
		%[a], [aa]
\rowcolors{2}{gray!15}{white}
\begin{tabular}{ r | l | l l l l l l }
    & Name & $ca$ & $cb$ & $oa$ & $ob$ & $l$ & $r$ \\
    \hline 
% 1  & 1	& 1	 & 0  & 0  &   &  \\
 1 & NO & 	& 	 &    &    &  & $\bot$  \Tstrut \Bstrut \\   
\end{tabular}
	\end{center}            
	\label{tab:ltrace4}
\end{table}

\begin{table}[h!]
	\caption{CSM: Operations}    
	\begin{center}
		\rowcolors{2}{gray!15}{white}
\begin{tabular}{ l | l l l l l l }
    Name & $ca$ & $cb$ & $oa$ & $ob$ & $l$ & $r$ \\
    \hline 
    INIT  & 1	& 1	 & 0  & 0  & &  \Tstrut \Bstrut \\
    ALEN  & 	& 	 &    &    & $|A[ca]| - oa$ &  \Tstrut \Bstrut \\ 
    BLEN  & 	& 	 &    &    & $|B[cb]| - ob$ &  \Tstrut \Bstrut \\     
    ASBN  & 	& $cb+1$ & $oa+l$  & 0  & &  \Tstrut \Bstrut \\
    ANBS  & $ca+1$	&  & 0  & $ob+l$  & &   \Tstrut \Bstrut \\    
    ANBN & $ca+1$	& $cb+1$	 & 0  &  0  & &   \Tstrut \Bstrut \\     
    YES &  	& 	 &    &    & & $\top$ \Tstrut \Bstrut \\  
    NO &  	& 	 &    &    & & $\bot$   \Tstrut \Bstrut \\
\end{tabular}
	\end{center}            
	\label{tab:stateops}
\end{table}

\begin{table}[h!]
	\caption{CSM: Predicates}    
	\begin{center}
		\rowcolors{2}{gray!15}{white}
\begin{tabular}{ l | l }
	Name & Expression  \\
	\hline 
	EQLEN &  $\sum_{i=1}^{|A|} |A[i]| = \sum_{i=1}^{|B|} |B[i]|$ \TTstrut \Bstrut
	\\[5pt]
	SS &  $\mathrm{substr}(A[ca],oa,l) = \mathrm{substr}(B[cb],ob,l)$ \Tstrut \Bstrut \\
	ALEQ & $ |A[ca]| - oa  \leq |B[cb]| - ob  $ \Tstrut \Bstrut \\
	EOA & $ |A[ca]| - oa = l$  \Tstrut \Bstrut \\
	EOB & $ |B[cb]| - ob = l$    \Tstrut \Bstrut \\   
	LASTA & $ ca = |A|$    \Tstrut \Bstrut \\    
	LASTB & $ cb = |B|$    \Tstrut \Bstrut \\ 
	FC & EOA $\wedge$ EOB $\wedge$ LASTA $\wedge$ LASTB \Tstrut \Bstrut \\
\end{tabular}
	\end{center}            
	\label{tab:predicates}
\end{table}

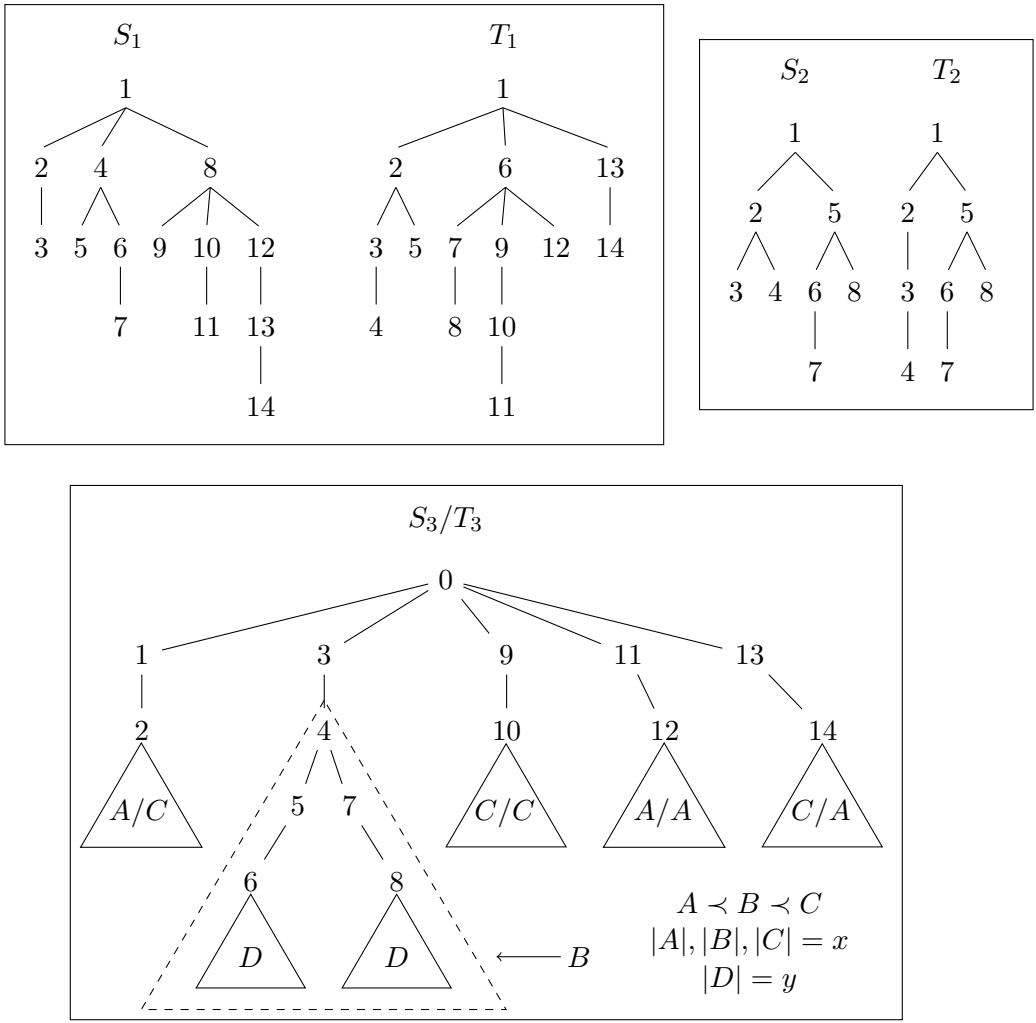
\begin{figure}[h!]
	\begin{center}
		\begin{tikzpicture}[shorten >=1pt,auto,node distance=1.2cm,
main node/.style={draw},triangle/.style = {draw, regular polygon, regular polygon sides=3 }]
%\node at (-2.4,1.6) {$S$};
%\node at (2.65,1.6) {$T$};

%first one

\node[main node] at (0,0) {
\begin{tikzpicture}[shorten >=1pt,auto,node distance=1.2cm,
main node/.style={draw},triangle/.style = {draw, regular polygon, regular polygon sides=3 }]
\usetikzlibrary{shapes.geometric}
%\node at (-2.4,1.6) {$S$};
%\node at (2.65,1.6) {$T$};

%first one
\node at (-2.35,2.8) {$S_1$};
\node at (-2,0) {
\Tree [.$1$
[.$2$ [.$3$ ] ]
[.$4$ [.$5$ ] [.$6$  [.$7$ ] ]  ]
[.$8$ [.$9$ ] [.$10$ [.$11$ ] ] [.$12$ [.$13$ [.$14$ ] ] ] ]
]
};

\node at (2.6,2.8) {$T_1$};
\node at (2.5,0) {
    \Tree [.$1$
    [.$2$ [.$3$  [.$4$ ] ] [.$5$ ]   ]
    [.$6$ [.$7$ [.$8$ ] ] [.$9$ [.$10$ [.$11$ ] ] ] [.$12$ ] ]    
    [.$13$ [.$14$ ] ]
    ]
};

\end{tikzpicture}
};

\node[main node] at (7,0) {
    \begin{tikzpicture}[shorten >=1pt,auto,node distance=1.2cm,
main node/.style={draw},triangle/.style = {draw, regular polygon, regular polygon sides=3 }]
%\node at (-2.4,1.6) {$S$};
%\node at (2.65,1.6) {$T$};

%first one
\node at (-1,2.4) {$S_2$};
\node at (-1,0) {
\Tree [.$1$ [.$2$ [.$3$ ] [.$4$ ] ] [.$5$ [.$6$ [.$7$ ] ] [.$8$ ]  ] ]
};

\node at (1,2.4) {$T_2$};
\node at (1,0) {
    \Tree [.$1$ [.$2$  [.$3$  [.$4$ ] ] ] [.$5$ [.$6$ [.$7$ ] ] [.$8$ ]  ]  ]
};
\end{tikzpicture}
};

\node[main node] at (2,-7) {
    \begin{tikzpicture}[shorten >=1pt,auto,node distance=1.2cm,
main node/.style={draw},triangle/.style = {draw, regular polygon, regular polygon sides=3 }]

\newcommand*{\xa}{4cm}%
\newcommand*{\ya}{1cm}%

\newcommand*{\xb}{4cm}%
\newcommand*{\yb}{-2.6cm}%

\newcommand*{\xlen}{1.6cm}%
\newcommand*{\ylen}{1cm}%

%\draw[dashed] ({\xa-1*\xlen},{\ya-2.15*\ylen}) --
% ({\xa-2.5*\xlen-7},{\ya-5.7*\ylen}) -- 
% ({\xa+0.5*\xlen+7},{\ya-5.7*\ylen}) -- 
% cycle;
 
\draw[dashed] ({\xa-1*\xlen},{\ya-1.6*\ylen}) --
 ({\xa-2.5*\xlen-0.2},{\ya-5.7*\ylen}) -- 
 ({\xa+0.5*\xlen-0.2},{\ya-5.7*\ylen}) -- 
 cycle;

\node (tlabel) at ({\xa},1.8) {$S_3$/$T_3$};

\node (srn) at ({\xa},{\ya}) {$0$};

\node (s1) at ({\xa-2.5*\xlen},{\ya-1*\ylen}) {$1$};
\node (s2) at ({\xa-1*\xlen},{\ya-1*\ylen}) {$3$};
\node (s3) at ({\xa+0.5*\xlen},{\ya-1*\ylen}) {$9$};
\node (s4) at ({\xa+1.5*\xlen},{\ya-1*\ylen}) {$11$};
\node (s5) at ({\xa+2.5*\xlen},{\ya-1*\ylen}) {$13$};

\node (sb1) at ({\xa-2.5*\xlen},{\ya-2*\ylen}) {$2$};
\node (sb2) at ({\xa-1*\xlen},{\ya-2*\ylen}) {$4$};
\node (sb3) at ({\xa+0.5*\xlen},{\ya-2*\ylen}) {$10$};
\node (sb4) at ({\xa+(1.5+0.3)*\xlen},{\ya-2*\ylen}) {$12$};
\node (sb5) at ({\xa+(2.5+0.6)*\xlen},{\ya-2*\ylen}) {$14$};

\node[triangle] (s1ta) at ({\xa-2.5*\xlen-0.2},{\ya-(3*\ylen)-2.5}) {\phantom{$A.$}};
\node[triangle] (s1tb) at ({\xa+0.5*\xlen-0.2},{\ya-3*\ylen-2.5}) {\phantom{$C.$}};
\node[triangle] (s1tc) at ({\xa+(1.5+0.3)*\xlen-0.2},{\ya-3*\ylen-2.5}) {\phantom{$A.$}};
\node[triangle] (s1td) at ({\xa+(2.5+0.6)*\xlen-0.2},{\ya-3*\ylen-2.5}) {\phantom{$C.$}};

\node (s1tal) at ({\xa-2.5*\xlen-0.2-1},{\ya-3*\ylen-2.5}) {$A$/$C$};
\node (s1tbl) at ({\xa+0.5*\xlen-0.2-0.5},{\ya-3*\ylen-2.5}) {$C$/$C$};
\node (s1tcl) at ({\xa+(1.5+0.3)*\xlen-0.2-0.5},{\ya-3*\ylen-2.5}) {$A$/$A$};
\node (s1tdl) at ({\xa+(2.5+0.6)*\xlen-0.2-0.5},{\ya-3*\ylen-2.5}) {$C$/$A$};

\node (sc5) at ({\xa-1.3*\xlen-0.2+4},{\ya-3*\ylen}) {$5$};
\node (sc7) at ({\xa-0.7*\xlen-0.2-4},{\ya-3*\ylen}) {$7$};

\node (sc6) at ({\xa-1.6*\xlen-0.2},{\ya-4*\ylen}) {$6$};
\node (sc8) at ({\xa-0.4*\xlen-0.2},{\ya-4*\ylen}) {$8$};

\node[triangle] (sc6t) at ({\xa-1.6*\xlen-0.2},{\ya-5*\ylen}) {$D$};
\node[triangle] (sc6t) at ({\xa-0.4*\xlen-0.2},{\ya-5*\ylen}) {$D$};

\node[] (sc6t) at ({\xa+1.1*\xlen-0.2},{\ya-5*\ylen}) {$B$};

\draw[->] ({\xa+0.95*\xlen-0.2},{\ya-5*\ylen}) -- ({\xa+0.4*\xlen-0.2},{\ya-5*\ylen});

\node[] (sc6t) at ({\xa+2.5*\xlen-0.2},{\ya-4.3*\ylen}) {$A \prec B \prec C$};
\node[] (sc6t) at ({\xa+2.5*\xlen-0.2},{\ya-(4.3+0.5)*\ylen}) {$|A|, |B|, |C| = x$};
\node[] (sc6t) at ({\xa+2.5*\xlen-0.2},{\ya-(4.3+1)*\ylen}) {$|D| = y$};

\path
(srn) edge (s3)
(srn) edge (s2)
(srn) edge (s1)
(srn) edge (s4)
(srn) edge (s5)

(s1) edge (sb1)
(s2) edge (sb2)
(s3) edge (sb3)
(s4) edge (sb4)
(s5) edge (sb5)

(sb2) edge (sc5)
(sb2) edge (sc7)

(sc5) edge (sc6)
(sc7) edge (sc8)
;

\end{tikzpicture}
};

\end{tikzpicture}
	\end{center}
	\caption{Tree Isomorphism: Input trees for the traces}
	\label{fig:inp_trees}
\end{figure}

\begin{table}[h!]
	\parbox[t]{.45\linewidth}{
		\centering 
		\caption{Tree Isomorphism: Trace for $S_1,T_1$ from Figure \ref{fig:inp_trees}}   \vspace{0.35cm}
		\rowcolors{2}{gray!15}{white}
\begin{tabular}{ r | l | l l l l l l }
     & Name & $s$ & $t$ & $k$ & $res$ \\
    \hline 
1 & INIT & 1 & 1 & 0 & \Tstrut \Bstrut \\ 
2 & NB &  &  & 2 & \Tstrut \Bstrut \\ 
3 & GC & 2 & 13 & 0 & \Tstrut \Bstrut \\ 
4 & NB &  &  & 1 & \Tstrut \Bstrut \\ 
5 & GC & 3 & 14 & 0 & \Tstrut \Bstrut \\ 
6 & $\cong$ &  &  &  & $\cong$ \Tstrut \Bstrut \\ 
7 & RET & 2 & 13 & 1 & \Tstrut \Bstrut \\ 
8 & $\cong$ &  &  &  & $\cong$ \Tstrut \Bstrut \\ 
9 & RET & 1 & 1 & 2 & \Tstrut \Bstrut \\ 
10 & NB &  &  & 4 & \Tstrut \Bstrut \\ 
11 & GC & 4 & 2 & 0 & \Tstrut \Bstrut \\ 
12 & NB &  &  & 1 & \Tstrut \Bstrut \\ 
13 & GC & 5 & 5 & 0 & \Tstrut \Bstrut \\ 
14 & $\cong$ &  &  &  & $\cong$ \Tstrut \Bstrut \\ 
15 & RET & 4 & 2 & 1 & \Tstrut \Bstrut \\ 
16 & NB &  &  & 2 & \Tstrut \Bstrut \\ 
17 & GC & 6 & 3 & 0 & \Tstrut \Bstrut \\ 
18 & NB &  &  & 1 & \Tstrut \Bstrut \\ 
19 & GC & 7 & 4 & 0 & \Tstrut \Bstrut \\ 
20 & $\cong$ &  &  &  & $\cong$ \Tstrut \Bstrut \\ 
21 & RET & 6 & 3 & 1 & \Tstrut \Bstrut \\ 
22 & $\cong$ &  &  &  & $\cong$ \Tstrut \Bstrut \\ 
23 & RET & 4 & 2 & 2 & \Tstrut \Bstrut \\ 
24 & $\cong$ &  &  &  & $\cong$ \Tstrut \Bstrut \\ 
25 & RET & 1 & 1 & 4 & \Tstrut \Bstrut \\ 
26 & NB &  &  & 7 & \Tstrut \Bstrut \\ 
27 & GC & 8 & 6 & 0 & \Tstrut \Bstrut \\ 
… &  &  &  &  & \Tstrut \Bstrut \\ 
29 & $\cong$ &  &  &  & $\cong$ \Tstrut \Bstrut \\ 
30 & RET & 1 & 1 & 7 & \Tstrut \Bstrut \\ 
31 & $\cong$ &  &  &  & $\cong$ \Tstrut \Bstrut \\ 
\end{tabular}	
		\label{tab:ti_trace1}
	}
	\hfill 
	\parbox[t]{.45\linewidth}{		
		\centering
		\caption{Tree Isomorphism: Trace for $S_2,T_2$ from Figure \ref{fig:inp_trees}} \vspace{0.35cm}    
		\rowcolors{2}{gray!15}{white}
\begin{tabular}{ r | l | l l l l l l }
     & Name & $s$ & $t$ & $k$ & $res$ \\
    \hline 
1 & INIT & 1 & 1 & 0 & \Tstrut \Bstrut \\ 
2 & NB &  &  & 3 & \Tstrut \Bstrut \\ 
3 & GC & 2 & 2 & 0 & \Tstrut \Bstrut \\ 
4 & $\succ$ &  &  &  & $\succ$ \Tstrut \Bstrut \\ 
5 & RET & 1 & 1 & 3 & \Tstrut \Bstrut \\ 
6 & $\succ$ &  &  &  & $\succ$ \Tstrut \Bstrut \\ 
\end{tabular}
		\label{tab:ti_trace2}
	}
\end{table}

\clearpage
\begin{center}	
	%\begin{tabular}{ r | l | c c c c c c c c c c c c c c c }
\rowcolors{2}{gray!15}{white}
\begin{longtable}{ r | l | l l l l l l l l l l l l l l l }
\caption{Tree Isomorphism: Trace for $S_3,T_3$ from Figure~\ref{fig:inp_trees}\label{tab:ti_trace3}} \\[5ex]
 & Name & $s$ & $t$ & $k$ & $res$ & $s'$ & $t'$ & $h$ & $sgt$ & $seq$ & $tgt$ & $teq$ & $f$ & $stk$  \\
 \hline 
 \endhead
1 & INIT & 0 & 0 & 0 &  &  &  &  &  &  &  &  &  & []  \Tstrut \Bstrut \\ 
2 & NB &  &  & $x+1$ &  &  &  &  &  &  &  &  &  &   \Tstrut \Bstrut \\ 
3 & SETH &  &  &  &  &  &  & 0 &  &  &  &  &  &   \Tstrut \Bstrut \\ 
4 & FINDS &  &  &  &  & 1 & 1 &  & 0 & 0 & 0 & 0 & S &   \Tstrut \Bstrut \\ 
5 & PUSH & 1 & 1 & 0 &  &  &  &  &  &  &  &  &  & [(0,0,0,0,0,S)]  \Tstrut \Bstrut \\ 
6 & NB &  &  & $x$ &  &  &  &  &  &  &  &  &  &   \Tstrut \Bstrut \\ 
7 & GC & 2 & 2 & 0 &  &  &  &  &  &  &  &  &  &   \Tstrut \Bstrut \\ 
$\dots$ &  &  &  &  &  &  &  &  &  &  &  &  &  &   \Tstrut \Bstrut \\ 
9 & $\prec$ &  &  &  & $\prec$ &  &  &  &  &  &  &  &  &   \Tstrut \Bstrut \\ 
10 & RET & 1 & 1 & $x$ &  &  &  &  &  &  &  &  &  &   \Tstrut \Bstrut \\ 
11 & $\prec$ &  &  &  & $\prec$ &  &  &  &  &  &  &  &  &   \Tstrut \Bstrut \\ 
12 & RET2 & 0 & 0 & $x+1$ &  & 1 & 1 & 0 & 0 & 0 & 0 & 0 & S & []  \Tstrut \Bstrut \\ 
13 & NXTS &  &  &  &  & 1 & 3 &  &  &  &  &  &  &   \Tstrut \Bstrut \\ 
14 & PUSH & 1 & 3 & 0 &  &  &  &  &  &  &  &  &  & [(0,0,0,0,0,S)]  \Tstrut \Bstrut \\ 
15 & NB &  &  & $x$ &  &  &  &  &  &  &  &  &  &   \Tstrut \Bstrut \\ 
16 & GC & 2 & 4 & 0 &  &  &  &  &  &  &  &  &  &   \Tstrut \Bstrut \\ 
$\dots$ &  &  &  &  &  &  &  &  &  &  &  &  &  &   \Tstrut \Bstrut \\ 
18 & $\prec$ &  &  &  & $\prec$ &  &  &  &  &  &  &  &  &   \Tstrut \Bstrut \\ 
19 & RET & 1 & 3 & $x$ &  &  &  &  &  &  &  &  &  &   \Tstrut \Bstrut \\ 
20 & $\prec$ &  &  &  & $\prec$ &  &  &  &  &  &  &  &  &   \Tstrut \Bstrut \\ 
21 & RET2 & 0 & 0 & $x+1$ &  & 1 & 3 & 0 & 0 & 0 & 0 & 0 & S & []  \Tstrut \Bstrut \\ 
22 & NXTS &  &  &  &  & 1 & 9 &  &  &  &  &  &  &   \Tstrut \Bstrut \\ 
23 & PUSH & 1 & 9 & 0 &  &  &  &  &  &  &  &  &  & [(0,0,0,0,0,S)]  \Tstrut \Bstrut \\ 
$\dots$ &  &  &  &  &  &  &  &  &  &  &  &  &  &   \Tstrut \Bstrut \\ 
25 & $\prec$ &  &  &  & $\prec$ &  &  &  &  &  &  &  &  &   \Tstrut \Bstrut \\ 
26 & RET2 & 0 & 0 & $x+1$ &  & 1 & 9 & 0 & 0 & 0 & 0 & 0 & S & []  \Tstrut \Bstrut \\ 
27 & NXTS &  &  &  &  & 1 & 11 &  &  &  &  &  &  &   \Tstrut \Bstrut \\ 
28 & PUSH & 1 & 11 & 0 &  &  &  &  &  &  &  &  &  & [(0,0,0,0,0,S)]  \Tstrut \Bstrut \\ 
$\dots$ &  &  &  &  &  &  &  &  &  &  &  &  &  &   \Tstrut \Bstrut \\ 
30 & $\cong$ &  &  &  & $\cong$ &  &  &  &  &  &  &  &  &   \Tstrut \Bstrut \\ 
31 & RET2 & 0 & 0 & $x+1$ &  & 1 & 11 & 0 & 0 & 1 & 0 & 0 & S &   \Tstrut \Bstrut \\ 
32 & NXTS &  &  &  &  & 1 & 13 &  &  &  &  &  &  &   \Tstrut \Bstrut \\ 
33 & PUSH & 1 & 13 & 0 &  &  &  &  &  &  &  &  &  & [(0,0,1,0,0,S)]  \Tstrut \Bstrut \\ 
$\dots$ &  &  &  &  &  &  &  &  &  &  &  &  &  &   \Tstrut \Bstrut \\ 
35 & $\cong$ &  &  &  & $\cong$ &  &  &  &  &  &  &  &  &   \Tstrut \Bstrut \\ 
36 & RET2 & 0 & 0 & $x+1$ &  & 1 & 13 & 0 & 0 & 2 & 0 & 0 & S &   \Tstrut \Bstrut \\ 
37 & FINDT &  &  &  &  & 1 & 1 &  &  &  &  &  & T &   \Tstrut \Bstrut \\ 
38 & PUSH & 1 & 1 & 0 &  &  &  &  &  &  &  &  &  & [(0,0,2,0,0,T)]  \Tstrut \Bstrut \\ 
$\dots$ &  &  &  &  &  &  &  &  &  &  &  &  &  &   \Tstrut \Bstrut \\ 
40 & $\prec$ &  &  &  & $\prec$ &  &  &  &  &  &  &  &  &   \Tstrut \Bstrut \\ 
41 & RET2 & 0 & 0 & $x+1$ &  & 1 & 1 & 0 & 0 & 2 & 1 & 0 & T & []  \Tstrut \Bstrut \\ 
42 & NXTT &  &  &  &  & 3 & 1 &  &  &  &  &  &  &   \Tstrut \Bstrut \\ 
43 & PUSH & 3 & 1 & 0 &  &  &  &  &  &  &  &  &  & [(0,0,2,1,0,T)]  \Tstrut \Bstrut \\ 
$\dots$ &  &  &  &  &  &  &  &  &  &  &  &  &  &   \Tstrut \Bstrut \\ 
45 & $\prec$ &  &  &  & $\prec$ &  &  &  &  &  &  &  &  &   \Tstrut \Bstrut \\ 
46 & RET2 & 0 & 0 & $x+1$ &  & 3 & 1 & 0 & 0 & 2 & 2 & 0 & T & []  \Tstrut \Bstrut \\ 
47 & NXTT &  &  &  &  & 9 & 1 &  &  &  &  &  &  &   \Tstrut \Bstrut \\ 
48 & PUSH & 9 & 1 & 0 &  &  &  &  &  &  &  &  &  & [(0,0,2,2,0,T)]  \Tstrut \Bstrut \\ 
$\dots$ &  &  &  &  &  &  &  &  &  &  &  &  &  &   \Tstrut \Bstrut \\ 
50 & $\cong$ &  &  &  & $\cong$ &  &  &  &  &  &  &  &  &   \Tstrut \Bstrut \\ 
51 & RET2 & 0 & 0 & $x+1$ &  & 9 & 1 & 0 & 0 & 2 & 2 & 1 & T &   \Tstrut \Bstrut \\ 
52 & NXTT &  &  &  &  & 11 & 1 &  &  &  &  &  &  &   \Tstrut \Bstrut \\ 
53 & PUSH & 11 & 1 & 0 &  &  &  &  &  &  &  &  &  & [(0,0,2,2,1,T)]  \Tstrut \Bstrut \\ 
$\dots$ &  &  &  &  &  &  &  &  &  &  &  &  &  &   \Tstrut \Bstrut \\ 
55 & $\prec$ &  &  &  & $\prec$ &  &  &  &  &  &  &  &  &   \Tstrut \Bstrut \\ 
56 & RET2 & 0 & 0 & $x+1$ &  & 11 & 1 & 0 & 0 & 2 & 3 & 1 & T & []  \Tstrut \Bstrut \\ 
57 & NXTT &  &  &  &  & 13 & 1 &  &  &  &  &  &  &   \Tstrut \Bstrut \\ 
58 & PUSH & 13 & 1 & 0 &  &  &  &  &  &  &  &  &  & [(0,0,2,3,1,T)]  \Tstrut \Bstrut \\ 
$\dots$ &  &  &  &  &  &  &  &  &  &  &  &  &  &   \Tstrut \Bstrut \\ 
60 & $\cong$ &  &  &  & $\cong$ &  &  &  &  &  &  &  &  &   \Tstrut \Bstrut \\ 
61 & RET2 & 0 & 0 & $x+1$ &  & 13 & 1 & 0 & 0 & 2 & 3 & 2 & T & []  \Tstrut \Bstrut \\ 
62 & NCT &  &  &  &  & 1 & 3 &  &  &  & 0 & 0 &  &   \Tstrut \Bstrut \\ 
63 & PUSH & 1 & 3 & 0 &  &  &  &  &  &  &  &  &  & [(0,0,2,0,0,T)]  \Tstrut \Bstrut \\ 
$\dots$ &  &  &  &  &  &  &  &  &  &  &  &  &  &   \Tstrut \Bstrut \\ 
65 & $\prec$ &  &  &  & $\prec$ &  &  &  &  &  &  &  &  &   \Tstrut \Bstrut \\ 
66 & RET2 & 0 & 0 & $x+1$ &  & 1 & 3 & 0 & 0 & 2 & 1 & 0 & T & []  \Tstrut \Bstrut \\ 
67 & NXTT &  &  &  &  & 3 & 3 &  &  &  &  &  &  &   \Tstrut \Bstrut \\ 
68 & PUSH & 3 & 3 & 0 &  &  &  &  &  &  &  &  &  & [(0,0,2,1,0,T)]  \Tstrut \Bstrut \\ 
69 & NB &  &  & $x$ &  &  &  &  &  &  &  &  &  &   \Tstrut \Bstrut \\ 
70 & GC & 4 & 4 & 0 &  &  &  &  &  &  &  &  &  &   \Tstrut \Bstrut \\ 
71 & NB &  &  & $y+1$ &  &  &  &  &  &  &  &  &  &   \Tstrut \Bstrut \\ 
72 & SETH &  &  &  &  &  &  & 0 &  &  &  &  &  &   \Tstrut \Bstrut \\ 
73 & FINDS &  &  &  &  & 5 & 5 &  & 0 & 0 & 0 & 0 & S &   \Tstrut \Bstrut \\ 
74 & PUSH & 5 & 5 & 0 &  &  &  &  &  &  &  &  &  & [\dots,(0,0,0,0,0,S)]  \Tstrut \Bstrut \\ 
75 & NB &  &  & $y$ &  &  &  &  &  &  &  &  &  &   \Tstrut \Bstrut \\ 
76 & GC & 6 & 6 & 0 &  &  &  &  &  &  &  &  &  &   \Tstrut \Bstrut \\ 
$\dots$ &  &  &  &  &  &  &  &  &  &  &  &  &  &   \Tstrut \Bstrut \\ 
78 & $\cong$ &  &  &  & $\cong$ &  &  &  &  &  &  &  &  &   \Tstrut \Bstrut \\ 
79 & RET & 5 & 5 & $y$ &  &  &  &  &  &  &  &  &  &   \Tstrut \Bstrut \\ 
80 & $\cong$ &  &  &  & $\cong$ &  &  &  &  &  &  &  &  &   \Tstrut \Bstrut \\ 
81 & RET2 & 4 & 4 & $y+1$ &  & 5 & 5 & 0 & 0 & 1 & 0 & 0 & S & [(0,0,2,1,0,T)]  \Tstrut \Bstrut \\ 
82 & NXTS &  &  &  &  & 5 & 7 &  &  &  &  &  &  & [\dots,(0,0,1,0,0,S)]  \Tstrut \Bstrut \\ 
83 & PUSH & 5 & 7 & 0 &  &  &  &  &  &  &  &  &  &   \Tstrut \Bstrut \\ 
84 & NB &  &  & $y$ &  &  &  &  &  &  &  &  &  &   \Tstrut \Bstrut \\ 
85 & GC & 6 & 8 & 0 &  &  &  &  &  &  &  &  &  &   \Tstrut \Bstrut \\ 
$\dots$ &  &  &  &  &  &  &  &  &  &  &  &  &  &   \Tstrut \Bstrut \\ 
87 & $\cong$ &  &  &  & $\cong$ &  &  &  &  &  &  &  &  &   \Tstrut \Bstrut \\ 
88 & RET & 5 & 7 & $y$ &  &  &  &  &  &  &  &  &  &   \Tstrut \Bstrut \\ 
89 & $\cong$ &  &  &  & $\cong$ &  &  &  &  &  &  &  &  &   \Tstrut \Bstrut \\ 
90 & RET2 & 4 & 4 & $y+1$ &  & 5 & 7 & 0 & 0 & 2 & 0 & 0 & S & [(0,0,2,1,0,T)]  \Tstrut \Bstrut \\ 
91 & FINDT &  &  &  &  & 5 & 5 &  &  &  &  &  & T &   \Tstrut \Bstrut \\ 
92 & PUSH & 5 & 5 & 0 &  &  &  &  &  &  &  &  &  & [\dots,(0,0,2,0,0,T)]  \Tstrut \Bstrut \\ 
$\dots$ &  &  &  &  &  &  &  &  &  &  &  &  &  &   \Tstrut \Bstrut \\ 
94 & $\cong$ &  &  &  & $\cong$ &  &  &  &  &  &  &  &  &   \Tstrut \Bstrut \\ 
95 & RET2 & 4 & 4 & $y+1$ &  & 5 & 5 & 0 & 0 & 2 & 0 & 1 & T & [(0,0,2,1,0,T)]  \Tstrut \Bstrut \\ 
96 & NXTT &  &  &  &  & 7 & 5 &  &  &  &  &  &  &   \Tstrut \Bstrut \\ 
97 & PUSH & 7 & 5 & 0 &  &  &  &  &  &  &  &  &  & [\dots,(0,0,2,0,1,T)]  \Tstrut \Bstrut \\ 
$\dots$ &  &  &  &  &  &  &  &  &  &  &  &  &  &   \Tstrut \Bstrut \\ 
99 & $\cong$ &  &  &  & $\cong$ &  &  &  &  &  &  &  &  &   \Tstrut \Bstrut \\ 
100 & RET2 & 4 & 4 & $y+1$ &  & 7 & 5 & 0 & 0 & 2 & 0 & 2 & T & [(0,0,2,1,0,T)]  \Tstrut \Bstrut \\ 
101 & INCH &  &  &  &  &  &  & 2 &  &  &  &  &  &   \Tstrut \Bstrut \\ 
102 & $\cong$ &  &  &  & $\cong$ &  &  &  &  &  &  &  &  &   \Tstrut \Bstrut \\ 
103 & RET & 3 & 3 & $x$ &  &  &  &  &  &  &  &  &  &   \Tstrut \Bstrut \\ 
104 & $\cong$ &  &  &  & $\cong$ &  &  &  &  &  &  &  &  &   \Tstrut \Bstrut \\ 
105 & RET2 & 0 & 0 & $x+1$ &  & 3 & 3 & 0 & 0 & 2 & 1 & 1 & T & []  \Tstrut \Bstrut \\ 
106 & NXTT &  &  &  &  & 9 & 3 &  &  &  &  &  &  &   \Tstrut \Bstrut \\ 
107 & PUSH & 9 & 3 & 0 &  &  &  &  &  &  &  &  &  & [(0,0,2,1,1,T)]  \Tstrut \Bstrut \\ 
$\dots$ &  &  &  &  &  &  &  &  &  &  &  &  &  &   \Tstrut \Bstrut \\ 
109 & $\succ$ &  &  &  & $\succ$ &  &  &  &  &  &  &  &  &   \Tstrut \Bstrut \\ 
110 & RET2 & 0 & 0 & $x+1$ &  & 9 & 3 & 0 & 0 & 2 & 1 & 1 & T & []  \Tstrut \Bstrut \\ 
111 & NXTT &  &  &  &  & 11 & 3 &  &  &  &  &  &  &   \Tstrut \Bstrut \\ 
112 & PUSH & 11 & 3 & 0 &  &  &  &  &  &  &  &  &  & [(0,0,2,1,1,T)]  \Tstrut \Bstrut \\ 
$\dots$ &  &  &  &  &  &  &  &  &  &  &  &  &  &   \Tstrut \Bstrut \\ 
114 & $\prec$ &  &  &  & $\prec$ &  &  &  &  &  &  &  &  &   \Tstrut \Bstrut \\ 
115 & RET2 & 0 & 0 & $x+1$ &  & 11 & 3 & 0 & 0 & 2 & 2 & 1 & T & []  \Tstrut \Bstrut \\ 
116 & NXTT &  &  &  &  & 13 & 3 &  &  &  &  &  &  &   \Tstrut \Bstrut \\ 
117 & PUSH & 13 & 3 & 0 &  &  &  &  &  &  &  &  &  & [(0,0,2,2,1,T)]  \Tstrut \Bstrut \\ 
$\dots$ &  &  &  &  &  &  &  &  &  &  &  &  &  &   \Tstrut \Bstrut \\ 
119 & $\succ$ &  &  &  & $\succ$ &  &  &  &  &  &  &  &  &   \Tstrut \Bstrut \\ 
120 & RET2 & 0 & 0 & $x+1$ &  & 13 & 3 & 0 & 0 & 2 & 2 & 1 & T & []  \Tstrut \Bstrut \\ 
121 & NCT &  &  &  &  & 1 & 9 &  &  &  & 0 & 0 &  &   \Tstrut \Bstrut \\ 
122 & PUSH & 1 & 9 & 0 &  &  &  &  &  &  &  &  &  & [(0,0,2,0,0,T)]  \Tstrut \Bstrut \\ 
$\dots$ &  &  &  &  &  &  &  &  &  &  &  &  &  &   \Tstrut \Bstrut \\ 
124 & RET2 & 0 & 0 & $x+1$ &  & 13 & 9 & 0 & 0 & 2 & 3 & 2 & T & []  \Tstrut \Bstrut \\ 
125 & NCT &  &  &  &  & 1 & 11 &  &  &  & 0 & 0 &  &   \Tstrut \Bstrut \\ 
126 & PUSH & 1 & 11 & 0 &  &  &  &  &  &  &  &  &  & [(0,0,2,0,0,T)]  \Tstrut \Bstrut \\ 
$\dots$ &  &  &  &  &  &  &  &  &  &  &  &  &  &   \Tstrut \Bstrut \\ 
128 & RET2 & 0 & 0 & $x+1$ &  & 13 & 11 & 0 & 0 & 2 & 0 & 2 & T & []  \Tstrut \Bstrut \\ 
129 & INCH &  &  &  &  &  &  & 2 &  &  &  &  &  &   \Tstrut \Bstrut \\ 
130 & FINDS &  &  &  &  & 1 & 1 &  & 0 & 0 & 0 & 0 & S &   \Tstrut \Bstrut \\ 
131 & PUSH & 1 & 1 & 0 &  &  &  &  &  &  &  &  &  & [(2,0,0,0,0,S)]  \Tstrut \Bstrut \\ 
$\dots$ &  &  &  &  &  &  &  &  &  &  &  &  &  &   \Tstrut \Bstrut \\ 
133 & RET2 & 0 & 0 & $x+1$ &  & 1 & 13 & 2 & 0 & 2 & 0 & 0 & S & []  \Tstrut \Bstrut \\ 
134 & NCS &  &  &  &  & 3 & 1 &  & 0 & 0 &  &  &  &   \Tstrut \Bstrut \\ 
135 & PUSH & 3 & 1 & 0 &  &  &  &  &  &  &  &  &  & [(2,0,0,0,0,S)]  \Tstrut \Bstrut \\ 
$\dots$ &  &  &  &  &  &  &  &  &  &  &  &  &  &   \Tstrut \Bstrut \\ 
137 & RET2 & 0 & 0 & $x+1$ &  & 3 & 13 & 2 & 2 & 1 & 0 & 0 & S & []  \Tstrut \Bstrut \\ 
138 & FINDT &  &  &  &  & 1 & 1 &  &  &  &  &  & T &   \Tstrut \Bstrut \\ 
139 & PUSH & 1 & 1 & 0 &  &  &  &  &  &  &  &  &  & [(2,2,1,0,0,T)]  \Tstrut \Bstrut \\ 
$\dots$ &  &  &  &  &  &  &  &  &  &  &  &  &  &   \Tstrut \Bstrut \\ 
141 & RET2 & 0 & 0 & $x+1$ &  & 13 & 1 & 2 & 2 & 1 & 3 & 2 & T & []  \Tstrut \Bstrut \\ 
142 & NCT &  &  &  &  & 1 & 3 &  &  &  & 0 & 0 &  &   \Tstrut \Bstrut \\ 
143 & PUSH & 1 & 3 & 0 &  &  &  &  &  &  &  &  &  & [(2,2,1,0,0,T)]  \Tstrut \Bstrut \\ 
$\dots$ &  &  &  &  &  &  &  &  &  &  &  &  &  &   \Tstrut \Bstrut \\ 
145 & RET2 & 0 & 0 & $x+1$ &  & 13 & 3 & 2 & 2 & 1 & 2 & 1 & T & []  \Tstrut \Bstrut \\ 
146 & INCH &  &  &  &  &  &  & 3 &  &  &  &  &  &   \Tstrut \Bstrut \\ 
147 & FINDS &  &  &  &  & 1 & 1 &  & 0 & 0 & 0 & 0 & S &   \Tstrut \Bstrut \\ 
148 & PUSH & 1 & 1 & 0 &  &  &  &  &  &  &  &  &  & [(3,0,0,0,0,S)]  \Tstrut \Bstrut \\ 
$\dots$ &  &  &  &  &  &  &  &  &  &  &  &  &  &   \Tstrut \Bstrut \\ 
150 & RET2 & 0 & 0 & $x+1$ &  & 1 & 13 & 3 & 0 & 2 & 0 & 0 & S & []  \Tstrut \Bstrut \\ 
151 & NCS &  &  &  &  & 3 & 1 &  & 0 & 0 &  &  &  &   \Tstrut \Bstrut \\ 
152 & PUSH & 3 & 1 & 0 &  &  &  &  &  &  &  &  &  & [(3,0,0,0,0,S)]  \Tstrut \Bstrut \\ 
$\dots$ &  &  &  &  &  &  &  &  &  &  &  &  &  &   \Tstrut \Bstrut \\ 
154 & RET2 & 0 & 0 & $x+1$ &  & 3 & 13 & 3 & 2 & 1 & 0 & 0 & S & []  \Tstrut \Bstrut \\ 
155 & NCS &  &  &  &  & 9 & 1 &  & 0 & 0 &  &  &  &   \Tstrut \Bstrut \\ 
156 & PUSH & 9 & 1 & 0 &  &  &  &  &  &  &  &  &  & [(3,0,0,0,0,S)]  \Tstrut \Bstrut \\ 
$\dots$ &  &  &  &  &  &  &  &  &  &  &  &  &  &   \Tstrut \Bstrut \\ 
158 & RET2 & 0 & 0 & $x+1$ &  & 9 & 13 & 3 & 3 & 2 & 0 & 0 & S & []  \Tstrut \Bstrut \\ 
159 & FINDT &  &  &  &  & 1 & 1 &  &  &  &  &  & T &   \Tstrut \Bstrut \\ 
160 & PUSH & 1 & 1 & 0 &  &  &  &  &  &  &  &  &  & [(3,3,2,0,0,T)]  \Tstrut \Bstrut \\ 
$\dots$ &  &  &  &  &  &  &  &  &  &  &  &  &  &   \Tstrut \Bstrut \\ 
162 & RET2 & 0 & 0 & $x+1$ &  & 13 & 1 & 3 & 3 & 2 & 3 & 2 & T &   \Tstrut \Bstrut \\ 
163 & INCH &  &  &  &  &  &  & 5 &  &  &  &  &  &   \Tstrut \Bstrut \\ 
164 & $\cong$ &  &  &  & $\cong$ &  &  &  &  &  &  &  &  &   \Tstrut \Bstrut \\ 
\end{longtable}
%\end{tabular}
\end{center}

\begin{table}[h!]
	\caption{Tree Isomorphism: Adjacency matrix of CFG}    
	\begin{center}
		\rotatebox{90}{
\rowcolors{2}{gray!15}{white}
%\begin{tabular}{ | l | c | c | c | c  | c  | c  | c  | c  | c  | c  | c  | c  | c  | c | c  | c  | c | }     
\begin{tabular}{  l | l  l  l l  l  l  l  l  l  l  l  l  l  l  l  l  l }     
 & \rotatebox{90}{INIT} & \rotatebox{90}{NB} & \rotatebox{90}{GC} & \rotatebox{90}{RET} & \rotatebox{90}{$\cong$} & \rotatebox{90}{$\prec$} & \rotatebox{90}{$\succ$} & \rotatebox{90}{PUSH} & \rotatebox{90}{RET2} & \rotatebox{90}{SETH} & \rotatebox{90}{INCH} & \rotatebox{90}{FINDS} & \rotatebox{90}{FINDT} & \rotatebox{90}{NXTS} & \rotatebox{90}{NXTT} & \rotatebox{90}{NCS} & \rotatebox{90}{NCT} \\
 \hline 
INIT &  & \ref{tab:ti_trace1}:1 &  &  & $\star$ & $\star$ & $\star$ &  &  &  &  &  &  &  &  &  & \\
NB &  &  & \ref{tab:ti_trace1}:2 &  &  &  &  &  &  & \ref{tab:ti_trace3}:2 &  &  &  &  &  &  & \\
GC &  & \ref{tab:ti_trace1}:3 &  &  & \ref{tab:ti_trace1}:5 & $\star$ & \ref{tab:ti_trace2}:3 &  &  &  &  &  &  &  &  &  & \\
RET &  & \ref{tab:ti_trace1}:9 &  &  & \ref{tab:ti_trace1}:7 & $\star$ & \ref{tab:ti_trace2}:5 &  &  &  &  &  &  &  &  &  & \\
$\cong$ &  &  &  & \ref{tab:ti_trace1}:6 &  &  &  &  & \ref{tab:ti_trace3}:30 &  &  &  &  &  &  &  & \\
$\prec$ &  &  &  & \ref{tab:ti_trace3}:9 &  &  &  &  & \ref{tab:ti_trace3}:11 &  &  &  &  &  &  &  & \\
$\succ$ &  &  &  & \ref{tab:ti_trace2}:4 &  &  &  &  & \ref{tab:ti_trace3}:109 &  &  &  &  &  &  &  & \\
PUSH &  & \ref{tab:ti_trace3}:5 &  &  & $\star$ & $\star$ & $\star$ &  &  &  &  &  &  &  &  &  & \\
RET2 &  &  &  &  &  &  &  &  &  &  & \ref{tab:ti_trace3}:100 &  & \ref{tab:ti_trace3}:36 & \ref{tab:ti_trace3}:12 & \ref{tab:ti_trace3}:41 & \ref{tab:ti_trace3}:133 & \ref{tab:ti_trace3}:61 \\
SETH &  &  &  &  &  &  &  &  &  &  &  & \ref{tab:ti_trace3}:3 &  &  &  &  & \\
INCH &  &  &  &  & \ref{tab:ti_trace3}:101 &  &  &  &  &  &  & \ref{tab:ti_trace3}:129 &  &  &  &  & \\
FINDS &  &  &  &  &  &  &  & \ref{tab:ti_trace3}:4 &  &  &  &  &  &  &  &  & \\
FINDT &  &  &  &  &  &  &  & \ref{tab:ti_trace3}:37 &  &  &  &  &  &  &  &  & \\
NXTS &  &  &  &  &  &  &  & \ref{tab:ti_trace3}:13 &  &  &  &  &  &  &  &  & \\
NCS &  &  &  &  &  &  &  & \ref{tab:ti_trace3}:151 &  &  &  &  &  &  &  &  & \\
NXTT &  &  &  &  &  &  &  & \ref{tab:ti_trace3}:42 &  &  &  &  &  &  &  &  & \\
NCT &  &  &  &  &  &  &  & \ref{tab:ti_trace3}:62 &  &  &  &  &  &  &  &  & \\
\end{tabular}
}
	\end{center}            
	\label{tab:ti_adj}
\end{table}

\end{document}